\documentclass[twocolumn,a4paper, accepted=2025-06-03]{quantumarticle}

\pdfoutput=1

\let\coloneqq\relax

\usepackage[T1]{fontenc}

\usepackage{xr}
\usepackage{amsthm}
\usepackage{amssymb}
\usepackage{amsmath}
\usepackage{bbold}
\usepackage{bbm}
\usepackage{xcolor}
\definecolor{myrefcolor}{rgb}{0.067,0.5,0.5}
\usepackage[
    breaklinks,
    pdftex,
    colorlinks=true,
    linkcolor=myrefcolor,
    citecolor=myrefcolor,
    urlcolor=myrefcolor
]{hyperref}
\usepackage{braket}
\usepackage{dsfont}
\usepackage{mathdots}
\usepackage{mathtools}
\usepackage{enumerate}
\usepackage[shortlabels]{enumitem}
\usepackage{csquotes}
\usepackage{stmaryrd}
\usepackage[cal=boondox]{mathalfa}
\usepackage{graphicx}
\usepackage{stackengine}
\usepackage{scalerel}
\usepackage{microtype}
\usepackage{array}
\usepackage{makecell}
\newcolumntype{x}[1]{>{\centering\arraybackslash}p{#1}}
\usepackage{tikz}
\usepackage{pgfplots}
\usetikzlibrary{shapes.geometric, shapes.misc, positioning, arrows, arrows.meta, decorations.pathreplacing, decorations.pathmorphing, patterns, angles, quotes, calc}
\usepackage{booktabs}
\usepackage{xfrac}
\usepackage{siunitx}
\usepackage{centernot}
\usepackage{comment}
\usepackage{chngcntr}
\usepackage{breakurl}
\def\clearthms#1{ \@for\tname:=#1\do{\cleartheorem\tname} }

\newtheorem*{thm*}{Theorem}
\newtheorem{thm}{Theorem}
\newtheorem{prop}[thm]{Proposition}
\newtheorem*{prop*}{Proposition}
\newtheorem{lemma}[thm]{Lemma}
\newtheorem*{lemma*}{Lemma}

\newtheorem{cor}[thm]{Corollary}
\newtheorem*{cor*}{Corollary}

\newtheorem*{cj*}{Conjecture}
\newtheorem{Def}[thm]{Definition}
\newtheorem*{Def*}{Definition}

\newtheorem{problem}[thm]{Problem}

\newtheorem{remark}[thm]{Remark}

\makeatletter
\def\thmhead@plain#1#2#3{%
  \thmname{#1}\thmnumber{\@ifnotempty{#1}{ }\@upn{#2}}%
  \thmnote{ {\the\thm@notefont#3}}}
\let\thmhead\thmhead@plain
\makeatother

\theoremstyle{definition}
\newtheorem{rem}[thm]{Remark}

\newcommand{\bb}{\begin{equation}\begin{aligned}\hspace{0pt}}
\newcommand{\be}{\begin{equation}\begin{aligned}\hspace{0pt}}
\newcommand{\bbb}{\begin{equation*}\begin{aligned}}
\newcommand{\ee}{\end{aligned}\end{equation}}
\newcommand{\eee}{\end{aligned}\end{equation*}}
\newcommand*{\coloneqq}{\mathrel{\vcenter{\baselineskip0.5ex \lineskiplimit0pt \hbox{\scriptsize.}\hbox{\scriptsize.}}} =}

\newcommand{\eqt}[1]{\stackrel{\mathclap{\scriptsize \mbox{#1}}}{=}}
\newcommand{\leqt}[1]{\stackrel{\mathclap{\scriptsize \mbox{#1}}}{\leq}}

\newcommand{\ketbra}[1]{\ket{#1}\!\!\bra{#1}}
\newcommand{\ketbraa}[2]{\ket{#1}\!\!\bra{#2}}

\newcommand{\sumno}{\sum\nolimits}
\newcommand{\prodno}{\prod\nolimits}

\newcommand{\G}{\mathrm{\scriptscriptstyle G}}

\newcommand{\id}{\mathds{1}}
\newcommand{\R}{\mathds{R}}
\newcommand{\N}{\mathds{N}}
\newcommand{\C}{\mathds{C}}

\DeclareMathOperator{\Tr}{Tr}

\DeclareMathAlphabet{\pazocal}{OMS}{zplm}{m}{n}

\DeclareMathOperator{\Id}{Id}

\DeclareMathOperator{\diag}{diag}

\newcommand{\NN}{\mathcal{N}}

\newcommand{\lsmatrix}{\left(\begin{smallmatrix}}
\newcommand{\rsmatrix}{\end{smallmatrix}\right)}

\stackMath

\stackMath

\makeatletter
\newcommand*\rel@kern[1]{\kern#1\dimexpr\macc@kerna}
\newcommand*\widebar[1]{%
  \begingroup
  \def\mathaccent##1##2{%
    \rel@kern{0.8}%
    \overline{\rel@kern{-0.8}\macc@nucleus\rel@kern{0.2}}%
    \rel@kern{-0.2}%
  }%
  \macc@depth\@ne
  \let\math@bgroup\@empty \let\math@egroup\macc@set@skewchar
  \mathsurround\z@ \frozen@everymath{\mathgroup\macc@group\relax}%
  \macc@set@skewchar\relax
  \let\mathaccentV\macc@nested@a
  \macc@nested@a\relax111{#1}%
  \endgroup
}

\newcommand{\fakepart}[1]{
 \par\refstepcounter{part}
  \sectionmark{#1}
}

\counterwithin*{equation}{part}
\counterwithin*{thm}{part}
\counterwithin*{figure}{part}

\tikzset{meter/.append style={draw, inner sep=10, rectangle, font=\vphantom{A}, minimum width=30, line width=.8, path picture={\draw[black] ([shift={(.1,.3)}]path picture bounding box.south west) to[bend left=50] ([shift={(-.1,.3)}]path picture bounding box.south east);\draw[black,-latex] ([shift={(0,.1)}]path picture bounding box.south) -- ([shift={(.3,-.1)}]path picture bounding box.north);}}}
\tikzset{roundnode/.append style={circle, draw=black, fill=gray!20, thick, minimum size=10mm}}
\tikzset{squarenode/.style={rectangle, draw=black, fill=none, thick, minimum size=10mm}}

\definecolor{Blues5seq1}{RGB}{239,243,255}
\definecolor{Blues5seq2}{RGB}{189,215,231}
\definecolor{Blues5seq3}{RGB}{107,174,214}
\definecolor{Blues5seq4}{RGB}{49,130,189}
\definecolor{Blues5seq5}{RGB}{8,81,156}

\definecolor{Greens5seq1}{RGB}{237,248,233}
\definecolor{Greens5seq2}{RGB}{186,228,179}
\definecolor{Greens5seq3}{RGB}{116,196,118}
\definecolor{Greens5seq4}{RGB}{49,163,84}
\definecolor{Greens5seq5}{RGB}{0,109,44}

\definecolor{Reds5seq1}{RGB}{254,229,217}
\definecolor{Reds5seq2}{RGB}{252,174,145}
\definecolor{Reds5seq3}{RGB}{251,106,74}
\definecolor{Reds5seq4}{RGB}{222,45,38}
\definecolor{Reds5seq5}{RGB}{165,15,21}

\usepackage{pgfplots}

\usepackage{anyfontsize}

\makeatletter
\newcommand{\raisemath}[1]{\mathpalette{\raisem@th{#1}}}
\newcommand{\raisem@th}[3]{\raisebox{#1}{$#2#3$}}
\makeatother

\newcommand{\TT}{\pazocal{T}}

\newcommand{\PP}{\pazocal{P}}
\newcommand{\QQ}{\pazocal{Q}}

\counterwithin*{figure}{part}
\usepackage[numbers,sort&compress]{natbib}

\counterwithin*{equation}{part}
\counterwithin*{thm}{part}
\usepackage{pgfplots}
\pgfplotsset{width=10cm,compat=1.9}
\usetikzlibrary{decorations.markings}
\newcommand{\fu}{\small Dahlem Center for Complex Quantum Systems, Freie Universit\"{a}t Berlin, 14195 Berlin, Germany}
\usepackage{algorithm}
\usepackage{algpseudocode}
\algrenewcommand\algorithmicrequire{\textbf{Input:}}
\algrenewcommand\algorithmicensure{\textbf{Output:}}

\begin{document}

\title{Optimal estimates of trace distance between bosonic Gaussian states and applications to learning}
\vspace{-1cm}
\author{Lennart Bittel}
\email{bittel@hhu.de}
\affiliation{\fu}

\author{Francesco Anna Mele}
\email{francesco.mele@sns.it}
\affiliation{NEST, Scuola Normale Superiore and Istituto Nanoscienze, Piazza dei Cavalieri 7, IT-56126 Pisa, Italy}

\author{Antonio Anna Mele}
\email{a.mele@fu-berlin.de}
\affiliation{\fu}

\author{Salvatore Tirone}
\email{s.tirone@uva.nl}
\affiliation{QuSoft, Science Park 123, 1098 XG Amsterdam, the Netherlands}
\affiliation{Korteweg--de Vries Institute for Mathematics, University of Amsterdam, Science Park 105-107, 1098 XG Amsterdam, the Netherlands}

\author{Ludovico Lami}
\email{ludovico.lami@gmail.com}
\affiliation{Scuola Normale Superiore, Piazza dei Cavalieri 7, 56126 Pisa, Italy}
\affiliation{QuSoft, Science Park 123, 1098 XG Amsterdam, the Netherlands}
\affiliation{Korteweg--de Vries Institute for Mathematics, University of Amsterdam, Science Park 105-107, 1098 XG Amsterdam, the Netherlands}
\affiliation{Institute for Theoretical Physics, University of Amsterdam, Science Park 904, 1098 XH Amsterdam, the Netherlands}

\begin{abstract}
Gaussian states of bosonic quantum systems enjoy several technological applications and are ubiquitous in nature. Their significance lies in their simplicity, which in turn rests on the fact that they are uniquely determined by two experimentally accessible quantities, their first and second moments. But what if these moments are only known \emph{approximately}, as is inevitable in any realistic experiment? What is the resulting error on the Gaussian state itself, as measured by the most operationally meaningful metric for distinguishing quantum states, namely, the trace distance? In this work, we fully resolve this question by demonstrating that if the first and second moments are known up to an error $\varepsilon$, the trace distance error on the state also scales as $\varepsilon$, and this functional dependence is optimal. To prove this, we establish tight bounds on the trace distance between two Gaussian states in terms of the norm distance of their first and second moments. As an application, we improve existing bounds on the sample complexity of tomography of Gaussian states. In our analysis, we introduce the general notion of derivative of a Gaussian state and uncover its fundamental properties, enhancing our understanding of the structure of the set of Gaussian states.
\end{abstract}

\maketitle

\let\oldaddcontentsline\addcontentsline
\renewcommand{\addcontentsline}[3]{}
\fakepart{Main text}

\noindent \textbf{\em Introduction.}--- Continuous variable (CV) systems~\cite{HOLEVO-CHANNELS-2,BUCCO,weedbrook12,CERF,adesso14,BARNETT-RADMORE}, such as quantum optical systems, play a key role in our quantum technologies, including quantum computation~\cite{Lloyd1999,Gottesman2001,Menicucci2006,Mirrahimi_2014,Ofek_nature2016,error_corr_boson,Guillaud_2019,alexander2024manufacturable,brenner2024factoringintegeroscillatorsqubit}, communication~\cite{Wolf2007,TGW,PLOB,Die-Hard-2-PRL,mele2023maximum, mele2024quantum,Die-Hard-2-PRA,mele2023optical}, and sensing~\cite{SGravi2,PhysRevLett.121.160502, PhysRevLett.86.5870,q_sensing_cv,CVCompressed}.
Additionally, CV systems have become increasingly popular in efforts to demonstrate quantum advantage, particularly through boson sampling~\cite{Borealis,Zhong_2020,SupremacyReview} and quantum simulation experiments~\cite{QuantumPhotoThermodynamics}. 
Among all CV quantum states, \emph{Gaussian states}~\cite{BUCCO} are arguably among the most important ones. 
Indeed, on the one hand Gaussian states are the most common class of states arising in Nature, while on the other they are relatively simple to analyse mathematically.

Gaussian states are the quantum analog of Gaussian probability distributions. Just as a Gaussian probability distribution is uniquely identified by its classical first moment and its classical covariance matrix, a Gaussian state is uniquely identified by its \emph{first moment} $\mathbf{m}$ and its \emph{covariance matrix} $V$~\cite{BUCCO}. Importantly, these quantities can be experimentally estimated through methods such as homodyne or heterodyne detection~\cite{BUCCO}, which are routinely performed in quantum optical laboratories. Thus, estimating a Gaussian state requires only the estimation of its first moment and covariance matrix. In practice, however, we can never estimate the first moment and the covariance matrix \emph{exactly}; instead, we will always have a non-zero error $\varepsilon$ in such estimates, meaning that we can only estimate a Gaussian state up to some error. The goal of this work is to understand how an estimation error $\varepsilon$ in the first moment and the covariance matrix of a Gaussian state propagates into the overall error in estimating the Gaussian state itself:
\begin{problem}[(Informal version)]\label{Problem1}
    Let $\rho(V,\mathbf{m})$ be the Gaussian state with first moment $\mathbf{m}$ and covariance matrix $V$. Let $\tilde{V}$ and $\tilde{\mathbf{m}}$ be $\varepsilon$-approximations of $V$ and $\mathbf{m}$, respectively:
    \bb
\tilde{\mathbf{m}}\approx_\varepsilon\mathbf{m}\,,\qquad \tilde{V}\approx_\varepsilon V\,.
    \ee
    The Gaussian state $\rho(\tilde{V},\tilde{\mathbf{m}})$ with first moment $\tilde{\mathbf{m}}$ and covariance matrix $\tilde{V}$ constitutes an approximation of $\rho(V,\mathbf{m})$:
    \bb\label{eq_approx}
        \rho(\tilde{V},\tilde{\mathbf{m}})\approx\rho(V,\mathbf{m})\,.
    \ee 
    The problem is to understand what is the error incurred in the approximation in Eq.~\eqref{eq_approx} in terms of the error $\varepsilon$ incurred in the approximations of the first moment and of the covariance matrix.
\end{problem}
How can we measure the error involved in estimating a quantum state? The most operationally meaningful approach is to define this error using the \emph{trace distance} between the true state and the estimated state, as established by the well-known Holevo--Helstrom theorem~\cite{HELSTROM, Holevo1976}. This theorem gives the trace distance a strong operational interpretation, identifying it as the appropriate metric for quantifying the probability of error in discriminating $\rho$ and $\sigma$ with quantum measurements~\cite{HELSTROM, Holevo1976}. The trace distance is ubiquitous in quantum information theory, playing a central role in evaluating the performance of numerous quantum algorithms~\cite{NC,MARK,Sumeet_book,WATROUS}. Mathematically, the trace distance between two quantum states $\rho$ and $\sigma$ is given by $\frac12\|\rho-\sigma\|_1$, where $\|\Theta\|_1=\Tr\sqrt{\Theta^\dagger\Theta}$ denotes the \emph{trace norm} of a linear operator $\Theta$. 

Given the operational interpretation of the trace distance, it is thus a fundamental problem to determine what the error incurred in the trace distance is when one estimates the first moment and the covariance matrix of an (unknown) Gaussian state up to a precision $\varepsilon$. Equivalently, it is a fundamental problem to estimate the trace distance between two Gaussian states in terms of the norm distances between their first moments and their covariance matrices:

\begin{problem}[(Formal version of Problem~\ref{Problem1})]\label{prob_2}
    Let $\rho(V,\mathbf{m})$ be the Gaussian state with first moment $\mathbf{m}$ and covariance matrix $V$. Similarly, let $\rho(\tilde{V},\tilde{\mathbf{m}})$ be the Gaussian state with first moment $\tilde{\mathbf{m}}$ and covariance matrix $\tilde{V}$. The problem is to determine stringent bounds on the trace distance  
    \bb
        \frac12\left\|\rho(V,\mathbf{m})-\rho(\tilde{V},\tilde{\mathbf{m}})\right\|_1
    \ee
    between $\rho(V,\mathbf{m})$ and $\rho(\tilde{V},\tilde{\mathbf{m}})$ in terms of the norm distance between their first moments $\| \mathbf{m}-\tilde{\mathbf{m}}\|$ and the norm distance between their covariance matrices $\| V-\tilde{V}\|$. (Note that, since the first moments and covariance matrices are finite-dimensional matrices, the choice of the norm $\|\cdot\|$ used to measure distances between them is not crucial, as all norms are equivalent in finite dimensions.)
\end{problem}
   

Deriving such bounds is of independent interest, as 
observed also by Alexander S.\ Holevo in~\cite{holevo2024estimatestracenormdistancequantum, holevo2024estimatesburesdistancebosonic}. Indeed, they can be useful in analysing tomography of Gaussian states~\cite{mele2024learningquantumstatescontinuous}, characterising quantum Gaussian observables~\cite{holevo2024estimatestracenormdistancequantum}, testing properties of Gaussian systems~\cite{bittel2024optimalestimatestracedistance}, and learning Gaussian processes~\cite{Haah_2023}. 
Problem~\ref{prob_2} 
was initially posed 
in Ref.~\cite{mele2024learningquantumstatescontinuous}, where some preliminary yet unsatisfactory results were also obtained. Recently, some further progress has been made by Holevo~\cite{holevo2024estimatestracenormdistancequantum}. However, neither the results in~\cite{mele2024learningquantumstatescontinuous} nor those in~\cite{holevo2024estimatestracenormdistancequantum} could fully resolve the problem of what a tight upper bound on the trace distance might look like.

In this work, we provide a complete solution to this problem, deriving tight bounds on the trace distance between two Gaussian states based on the norm distance of their first moments and covariance matrices. Our bounds demonstrate that estimating the first moment and covariance matrix of a Gaussian state with precision $\varepsilon$ results in a trace distance error that scales linearly with $\varepsilon$ (and vice versa). We expect that these new trace distance bounds will serve as an important technical tool in advancing the field of quantum optics and CV quantum information. 

To date, CV researchers have primarily relied on fidelity~\cite{Banchi_2015} rather than trace distance to estimate the distance between two Gaussian states. This preference stems from the availability of a well-known closed formula for the fidelity between Gaussian states~\cite{Banchi_2015}. However, this approach has two significant drawbacks: (i)~fidelity lacks the strong operational interpretation that trace distance offers, and (ii)~the fidelity formula is cumbersome, making it challenging to use in analytical calculations. In contrast, the trace distance bounds introduced in this work are both operationally meaningful and analytically tractable, providing a practical alternative to the widely used fidelity formula~\cite{Banchi_2015} for quantifying the distance between Gaussian states.

As an application of our bounds, we establish a tighter sample complexity bound for Gaussian state tomography, outperforming 
current state-of-the-art result~\cite{mele2024learningquantumstatescontinuous}. This contributes to the rapidly growing literature on quantum learning theory with CV systems~\cite{mele2024learningquantumstatescontinuous, anshu2023survey, aolita_reliable_2015, gandhari_precision_2023, becker_classical_2023, oh2024entanglementenabled, fawzi2024optimalfidelityestimationbinary, möbus2023dissipationenabledbosonichamiltonianlearning,upreti2024efficientquantumstateverification, fanizza2024efficienthamiltonianstructuretrace}.

In our analysis, we rigorously define and discover fundamental properties of the \emph{derivative} of a Gaussian state with respect to its first moment and covariance matrix. Remarkably, we demonstrate that this derivative can be expressed through a relatively simple formula, offering new insights into the underlying structure of the geometry of the set of Gaussian states. We believe that this result will have broad applications in quantum metrology with CV systems~\cite{PhysRevA.98.012114, Giova_metro1, giova_metro2, BUCCO}. In fact, the computation of state derivatives is a central component in the analysis of quantum metrology protocols, essentially because the celebrated quantum Cramér–Rao bound~\cite{BUCCO} inherently relies on state derivatives. Until now, derivatives of Gaussian states in the context of quantum metrology have been computed only for very specific examples~\cite{ref_der_mark_1, ref_der_mark_2, ref_der_mark_3}. In contrast, our results apply to arbitrary Gaussian states, providing new tools that can advance the field of quantum metrology.

\medskip \noindent 
\textbf{\em Prior work.\,}--- Let us recap the recent rapid advances around the solution of Problem~\ref{prob_2}. First, this problem was posed in~\cite{mele2024learningquantumstatescontinuous}, where some upper and lower bounds on the trace distance between Gaussian states in terms of the norm distance between the first moments and covariance matrices were established. Specifically, the bounds proved in~\cite{mele2024learningquantumstatescontinuous} guarantee that an $\varepsilon$-error in the moments leads to a trace distance error that scales at least as $\varepsilon$ and at most as $\sqrt{\varepsilon}$~\cite{mele2024learningquantumstatescontinuous}. Later, Holevo elegantly proved a new trace distance upper bound that can outperform the one in~\cite{mele2024learningquantumstatescontinuous} in some regimes~\cite{holevo2024estimatestracenormdistancequantum}. However, Holevo's new upper bound still leads to a trace distance error that scales as $\sqrt{\varepsilon}$~\cite{holevo2024estimatestracenormdistancequantum} in most cases of practical interest. This suboptimal error scaling has serious drawbacks, as it can lead to a significant overestimation of the performance of quantum algorithms involving Gaussian states, especially in the context of quantum learning~\cite{mele2024learningquantumstatescontinuous}.  Given the central role of Gaussian states in optical quantum technologies, addressing this limitation and determining the optimal error dependence is thus a crucial, pressing problem. In this work, we improve on this dependence and fully resolve the problem, establishing that the trace distance error must scale as $\varepsilon$.

One might hope that a good trace distance upper bound could be derived by looking at the fidelity instead, which can be computed explicitly for Gaussian states~\cite{Banchi_2015}. However, this approach ultimately fails because the fidelity formula derived in~\cite{Banchi_2015} is too complicated to yield a simple bound in terms of the differences between the moments. Also, it would be much worse than our bound due to the square root loss in the Fuchs--van de Graaf inequality~\cite{NC,MARK,WATROUS}.

In the fermionic setting, tight bounds have been recently derived for the trace distance between two fermionic Gaussian states~\cite{bittel2024optimaltracedistanceboundsfreefermionic}. However, while the latter can in principle be computed numerically (for system with a small number of fermionic modes), the trace distance between two bosonic Gaussian states is more challenging to compute numerically (even for a single-mode system), as doing that would require the diagonalisation of an infinite-dimensional matrix. This underscores the critical importance of deriving such trace distance bounds in the bosonic setting, where they are even more essential than in the fermionic case and arguably considerably more challenging to establish due to the inherent difficulties of working with infinite-dimensional quantum systems.

\medskip \noindent 
\textbf{\em Notation.\,}--- By definition, a CV system with $n$ \emph{modes} is represented by the Hilbert space $L^2(\mathbb R^n)$, which comprises all square-integrable complex-valued functions over $\R^n$. The first moment $\mathbf{m}$ and the covariance matrix $V$ of a quantum state $\rho$ are defined as $\mathbf{m}\coloneqq \Tr\big[\mathbf{\hat{R}}\,\rho\big]$ and $V\coloneqq\Tr\big[\big\{\mathbf{(\hat{R}-m\,\hat{\mathbb{1}}),(\hat{R}-m\,\hat{\mathbb{1}})}^{\raisebox{-1.9pt}{\scriptsize $\intercal$}}\big\}\rho\big]$, where $(\cdot)^\intercal$ denotes the transpose operation, $\{\cdot,\cdot\}$ represents the anti-commutator, and $\mathbf{\hat{R}}\coloneqq (\hat{x}_1,\hat{p}_1,\dots,\hat{x}_n,\hat{p}_n)^{\intercal}$ is the quadrature operator vector, with $\hat{x}_1,\hat{p}_1,\dots,\hat{x}_n,\hat{p}_n$ being the well-known position and momentum operators of each mode~\cite{BUCCO}. The quadrature operator vector satisfies the canonical commutation relation $[\mathbf{\hat{R}},\mathbf{\hat{R}}^{\intercal}]=i\,\Omega\,\mathbb{\hat{1}}$, where $\Omega\coloneqq\bigoplus_{i=1}^n \lsmatrix 0&1 \\ -1&0 \rsmatrix$. 

The set of $n$-mode Gaussian states is in one-to-one correspondence with the set of pairs $(V,\mathbf{m})$, where $V$ is a $2n\times 2n$ real matrix satisfying the \emph{uncertainty relation} $V+i\Omega \ge 0$, and $\mathbf{m}$ is a $2n$-dimensional real vector~\cite{BUCCO}. Specifically, given such a pair $(V,\mathbf{m})$ with $V$ satisfying the uncertainty relation, there exists a unique Gaussian state with covariance matrix $V$ and first moment $\mathbf{m}$~\cite{BUCCO}. Conversely, the covariance matrix of any (Gaussian) state satisfies the uncertainty relation~\cite{BUCCO}. For the rest of the paper, we will denote as $\rho(V,\mathbf{m})$ the Gaussian state with covariance matrix $V$ and first moment $\mathbf{m}$.

Given a matrix $V$, we write $|V|\coloneqq \sqrt{V^\dagger V}$ 
for its modulus, $\|V\|_\infty$ 
for its operator norm (i.e.~the maximum singular value of $V$), $\|V\|_1\coloneqq\Tr|V|$ 
for its trace norm, and $\|V\|_2\coloneqq \sqrt{\Tr [V^\dagger V]}$ 
for its Hilbert--Schmidt norm. The Euclidean norm of a vector $\mathbf{m}$ 
is denoted as $\|\mathbf{m}\|_2\coloneqq \sqrt{\mathbf{m}^\intercal\mathbf{m}}$.

\medskip \noindent 
\textbf{\em Tight upper bound.\,}--- The following theorem establishes a tight upper bound on the trace distance between two Gaussian states in terms of the norm distances between their first moments and covariance matrices.

\begin{thm}[(Tight upper bound on the trace distance between Gaussian states)]\label{thm_main000}
    Let $\rho(V,\mathbf{m})$ be the Gaussian state with first moment $\mathbf{m}$ and covariance matrix $V$. Similarly, let $\rho(W,\mathbf{t})$ be the Gaussian state with first moment $\mathbf{t}$ and covariance matrix $W$. Then, their trace distance can be upper bounded as follows:
    \bb
    \label{eq_ineq_main}
        \frac{1}{2}\left\|\rho(V,\mathbf{m})\!-\!\rho(W,\mathbf{t})\right\|_1 &\le \tfrac{1+\sqrt{3}}{8}\Tr\!\left[|V\!-\!W|\,\Omega^\intercal\!\!\left(\tfrac{V+W}{2}\right)\!\Omega\right] \\
        &\ + \sqrt{\tfrac{\min(\|V\|_\infty,\|W\|_\infty)}{2}}\, \|\mathbf{m}\!-\!\mathbf{t}\|_2\, .
    \ee
\end{thm}

To prove the above result, we develop a new mathematical framework based on the notion of \emph{derivative of a Gaussian state}, which can be of independent interest (see e.g.~\cite{huang2024informationgeometrybosonicgaussian}), especially for the field of quantum metrology~\cite{BUCCO,Giova_metro1,giova_metro2}. See Section~\ref{Sec_proof_derivative} of the Supplemental Material (SM)~\cite{SM}.

Crucially, H\"older's inequality implies that the term $\Tr\!\left[|V-W|\Omega^\intercal\left(\frac{V+W}{2}\right)\Omega\right]$ in~\eqref{eq_ineq_main} can be further upper bounded either by \( \max(\|V\|_\infty,\|W\|_\infty) \|V-W\|_1 \), or by $\max(\Tr V, \Tr W) \|V-W\|_\infty $. This is important because it implies that the upper bound in Theorem~\ref{thm_main000} provides a linear dependence on the norm of $V-W$, improving upon the square root dependence previously established in~\cite{mele2024learningquantumstatescontinuous} and~\cite{holevo2024estimatestracenormdistancequantum}.  Remarkably, the lower bound presented in~\cite[Eq.~S265]{mele2024learningquantumstatescontinuous} also exhibits the same linear dependence as our new upper bound. Consequently, it follows that our new bound provides a complete solution to Problem~\ref{prob_2}. This fully answers a fundamental question for the field of quantum optics and CV quantum information~\cite{mele2024learningquantumstatescontinuous}, 
demonstrating that if the first moment and the covariance matrix of a quantum Gaussian state are known within an error $\varepsilon$, the associated trace distance 
uncertainty scales linearly with $\varepsilon$ (and this functional dependence is optimal).

Notably, the bound in Theorem~\ref{thm_main000} is \emph{tight} with respect to its dependence on $\varepsilon\coloneqq \|V-W\|_\infty$ and $a\coloneqq \max(\|V\|_\infty,\|W\|_\infty)$. Specifically, this bound establishes that $\frac12\|\rho(V,0)-\rho(W,0)\|_1\le \frac{1+\sqrt{3}}{8}a\varepsilon$. As proved in Lemma~\ref{lemma_tight} in the SM, the trace distance between the two Gaussian states with covariance matrices \( V \coloneqq \lsmatrix a & 0 \\ 0 & a^{-1} \rsmatrix \) and \( W\coloneqq \lsmatrix a & 0 \\ 0 & a^{-1} + \varepsilon \rsmatrix \) and zero first moments satisfies 
\bb
    \lim\limits_{\varepsilon\rightarrow0^+}\frac{\frac{1}{2}\|\rho(W,0)-\rho(V,0)\|_1}{\varepsilon}=\frac{1+\sqrt{3}}{8}a\quad  \!\!\forall a\ge 1\,,
\ee
establishing that the constant $\frac{1+\sqrt{3}}{8}$ in~\eqref{eq_ineq_main} is optimal.
We can achieve have optimality also for the first moment dependence as all approximation in \eqref{step_same_cov} can be made tight for coherent states.

\medskip
\noindent \textbf{\em Application to quantum state tomography.\,}--- As a direct application of Theorem~\ref{thm_main000}, we can provide an upper bound on the sample complexity of tomography of Gaussian states~\cite{mele2024learningquantumstatescontinuous} that outperforms the one provided in~\cite{mele2024learningquantumstatescontinuous}. Let us start by briefly introducing the fundamental problem of tomography of Gaussian states~\cite{mele2024learningquantumstatescontinuous}.
\begin{problem}[(Tomography of Gaussian states)]\label{prob1}
Given $N$ copies of an unknown $n$-mode Gaussian state $\rho$ satisfying the energy constraint  $\Tr[\rho\hat{E}] \leq E$, the goal is to construct a classical description of a Gaussian state $\tilde{\rho}$ such that  
\bb
    \Pr\left(\frac{1}{2}\|\tilde{\rho} - \rho\|_1 \leq \varepsilon\right) \geq 1 - \delta\,.
\ee
In words, $\tilde{\rho}$ must be such that the probability that it is $\varepsilon$-close to $\rho$ (in trace distance) is not smaller than $1-\delta$. Here, $\hat{E}\coloneqq \frac12\hat{\mathbf{R}}^\intercal \hat{\mathbf{R}}$ denotes the energy operator~\cite{BUCCO}.
\end{problem}
Given the number of modes $n$, the energy constraint $E$, the trace distance error $\varepsilon$, and the failure probability $\delta$, the \emph{sample complexity} of tomography of Gaussian states is defined as the minimum number of copies $N$ that allows one to solve Problem~\ref{prob1}. In the following theorem, we find an upper bound on such a sample complexity.

\begin{thm}[(Upper bound on the sample complexity of tomography of energy-constrained Gaussian states)]
\label{th:main2}
Let $\varepsilon, \delta \in (0,1]$.
Let \(\rho\) be an unknown \(n\)-mode Gaussian state with (unknown) covariance matrix $V$ and arbitrary first moment, satisfying the energy constraint \(\Tr[\rho \hat{E}] \leq E\). A number 
\begin{align}
    N &= O\!\left(\frac{(\|V\|_1\|V\|_{\infty})^2}{\varepsilon^2} \left( n + \log(2/\delta) \right) \right) \\&\le O\!\left(\frac{E^4}{\varepsilon^2} \left( n + \log(2/\delta) \right) \right) 
\end{align}
of copies of \(\rho\) suffices to construct a classical description of a Gaussian state \(\tilde{\rho}\) such that $\Pr\left( \frac{1}{2}\|\tilde{\rho} - \rho\|_1 \leq \varepsilon \right) \geq 1 - \delta$.
\end{thm}
The proof is provided in Theorem~\ref{th:appmain2} and Corollary~\ref{cor:energy_constraint} in the SM~\cite{SM}, where explicit sample-complexity bounds (without big-$O$ notation) are provided. 
The above Theorem~\ref{th:main2} improves upon the upper bound on the sample complexity of tomography of Gaussian states presented in Ref.~\cite{mele2024learningquantumstatescontinuous}, which was $N = O\!\left(\frac{n^3E^4}{\varepsilon^4}\right)$. The improvement is evident in both the scaling with the number of modes $n$ and the scaling with the trace distance error $\varepsilon$, ultimately attributed to our trace distance bound in Theorem~\ref{thm_main000}.

\medskip
\noindent \textbf{\em Discussion and conclusions.\,}--- In this work, we have presented novel and optimal bounds on the trace distance between bosonic Gaussian states in terms of the norm difference of their first moments and covariance matrices. These bounds allow one to understand how errors in estimating the first moment and covariance matrix of a Gaussian state propagate to the trace distance error, which is widely regarded as the most operationally meaningful measure of error in the estimation of a quantum state. This result fully resolves an open question posed in~\cite{mele2024learningquantumstatescontinuous} and subsequently explored by Holevo~\cite{holevo2024estimatestracenormdistancequantum}. Our analysis can be viewed as a quantum analog of the estimates of the total variation distance between Gaussian probability distributions in the classical probability literature~\cite{Barsov1987,devroye2023total,arbas2023polynomial}, which play a significant role in machine learning applications~\cite{arbas2023polynomial}.

As an application of our new trace distance bounds, we have improved the existing bounds on the sample complexity of tomography of Gaussian states~\cite{mele2024learningquantumstatescontinuous}.
This result contributes to the rapidly growing literature on quantum learning theory and state learning~\cite{anshu2023survey, odonnell2015efficient, Haah_2017, kueng2014low, chen2023does, guta2018fast, fanizza2023learning, huang2024learning, arunachalam2023optimal, montanaro2017learning, grewal2023efficient, leone2023learning, hangleiter2024bell, aaronson2023efficient, mele2024efficient}, particularly at its interface with CV systems, which have been relatively unexplored until recently~\cite{aolita_reliable_2015, gandhari_precision_2023, becker_classical_2023, oh2024entanglementenabled, fawzi2024optimalfidelityestimationbinary, möbus2023dissipationenabledbosonichamiltonianlearning,upreti2024efficientquantumstateverification, fanizza2024efficienthamiltonianstructuretrace}. Since Gaussian states 
play a major role both theoretically and experimentally~\cite{BUCCO}, we believe that our results on their learnability 
hold both fundamental 
as well as practical value.

Finally, we have introduced the notion of derivative of Gaussian states and uncovered its fundamental properties, deepening our understanding of the structure of the set of Gaussian states. We believe that these findings will have applications in quantum metrology and sensing~\cite{BUCCO,Giova_metro1,giova_metro2}. In fact, since the first posting of this paper on arXiv, progress in this research direction has already begun~\cite{Mark_der}.

As additional future research directions, we propose extending our results to derive analogous bounds for the distance between Gaussian channels~\cite{BUCCO} and investigating their learnability~\cite{anshu2023survey, Haah_2023}. Additionally, applying our trace-distance bounds to other learning tasks, such as property testing of bosonic Gaussian states~\cite{anshu2023survey, bittel2024optimaltracedistanceboundsfreefermionic}, offers a compelling avenue for exploration. Together with our findings and the potential applications of Gaussian state derivatives in quantum metrology and sensing, these directions present exciting opportunities to further enhance the role of CV systems in advancing quantum technologies.

\begin{acknowledgments}
\smallskip
\noindent \emph{Acknowledgements.} We thank Salvatore F.~E. Oliviero, Lorenzo Leone, Jonathan Conrad, and Vittorio Giovannetti for helpful discussions. We thank Marco Fanizza for help with the sample complexity analysis. L.B. and A.A.M. have been been supported by the BMWK (EniQmA), BMBF (FermiQP, MuniQCAtoms, DAQC), the ERC (DebuQC) and the DFG (CRC 183). F.A.M. acknowledge financial support by MUR (Ministero dell'Istruzione, dell'Universit\`a e della Ricerca) through the following projects: PNRR MUR project PE0000023-NQSTI, PRIN 2017 Taming complexity via Quantum Strategies: a Hybrid Integrated Photonic approach (QUSHIP) Id.\ 2017SRN-BRK, and project PRO3 Quantum Pathfinder. LL acknowledges support from MIUR (Ministero dell'Istruzione, dell'Universit\`{a} e della Ricerca) through the project `Dipartimenti di Eccellenza 2023--2027' of the `Classe di Scienze' department at the Scuola Normale Superiore. 

\smallskip
\noindent \emph{Note.}  Concurrently with our work, another study~\cite{fanizza2024efficienthamiltonianstructuretrace} demonstrates a method for learning Gaussian states in trace distance, achieving a quadratic scaling in precision and a polynomial dependence on the number of modes by leveraging a newly introduced trace distance bound for Gaussian states, which holds under certain additional restrictions on the class of Gaussian states considered compared to our work.
\end{acknowledgments}

\medskip


\bibliographystyle{apsrev4-2}
\bibliography{biblio}

\appendix
\noindent \textbf{\em Technical tools.}---
To establish our trace distance bounds, we introduce numerous technical tools that may hold independent interest for the field of quantum optics, CV quantum information, and quantum metrology. Among these tools, we introduce the concept of the \emph{derivative of a Gaussian state} (see Section~\ref{Sec_Der}). Specifically, for an \( n \)-mode Gaussian state \( \rho(V, \mathbf{m}) \) with covariance matrix \( V \) and first moment \( \mathbf{m} \), and given a \( 2n \times 2n \) real symmetric matrix \( X \) and a \( 2n \)-dimensional real vector \( \mathbf{x} \), we define the derivative of \( \rho(V, \mathbf{m}) \) in the direction \( (X, \mathbf{x}) \) as the following limit of incremental ratios:
\bb
    &\frac{\mathrm{d}}{\mathrm{d}\alpha}\rho(V+\alpha X,\textbf{m}+\alpha \textbf{x})\Bigg|_{\alpha=0}\\
    &\,\coloneqq \lim_{\alpha \to 0^+} 
        \frac{\rho(V + \alpha X, \mathbf{m} + \alpha \mathbf{x}) 
        - \rho(V, \mathbf{m})}{\alpha}\,.
\ee
The existence of such a limit is rigorously proved in Section~\ref{sec:proof_diff} in the SM from a functional analysis perspective. Notably, the following lemma, proved in Theorem~\ref{lemma_der_state} in the SM, demonstrates that the derivative of a Gaussian state can be expressed with a relatively simple formula:
\begin{lemma}[(Compact formula for the derivative of a Gaussian state)]
\label{lemma_formula_comp}
    The derivative of the Gaussian state $\rho(V,\mathbf{m})$ along the direction $(X,\mathbf{x})$ can be expressed as:
    \begin{align}\label{eq_der_main_end}
        &\frac{\mathrm{d}}{\mathrm{d}\alpha}\rho(V+\alpha X,\textbf{m}+\alpha \textbf{x})\Bigg|_{\alpha=0}\\
        &\, = -\frac{1}{4} \sum_{k,j=1}^{2n} \tilde{X}_{kj} 
        [\hat{R}_k, [\hat{R}_j, \rho(V, \mathbf{m})]]  + i \sum_{j=1}^{2n} \tilde{x}_j 
        [\hat{R}_j, \rho(V, \mathbf{m})],
        \nonumber
    \end{align}
    where $\tilde{X} \coloneqq \Omega^\intercal X \Omega$, 
    $\tilde{\mathbf{x}} \coloneqq \Omega \mathbf{x}$, and $\hat{R}_j$ denotes the $j$th quadrature operator.
\end{lemma} 
Such a compact formula is crucial for the proof of our trace-distance bounds and allows to shed light on the geometry of set of Gaussian states. Remarkably, in the following lemma, proved in Theorem~\ref{exact_der} in the SM, we show how to calculate the trace norm of the derivative of a pure Gaussian state \emph{exactly} in terms of its covariance matrix:
\begin{lemma}[(Closed formula for the trace norm of the derivative of a pure Gaussian state)]\label{exact_der_main}
Let $\rho(V,\textbf{m})$ be a pure Gaussian state with a (pure) covariance matrix $V$ and first moment $\textbf{m}$. Then, the trace norm of the derivative of the Gaussian state $\rho(V,\textbf{m})$ along the direction $(X,0)$ is given by 
	\bb
	&\left\|  \frac{\mathrm{d}}{\mathrm{d}\alpha}\rho(V+\alpha X,\textbf{m})\Bigg|_{\alpha=0} \right\|_1\\&= \frac14\sqrt{ (\Tr[\tilde{X}V])^2+2\Tr[\tilde{X}V\tilde{X}V]+2\Tr[\Omega \tilde{X}\Omega \tilde{X}]
 }\\&\quad+\frac14\|\sqrt{V+i\Omega}\tilde{X}\sqrt{V+i\Omega}\|_1\,,
	\ee
	where $\tilde{X}\coloneqq \Omega X\Omega^\intercal$.
\end{lemma}
In the general case of \emph{mixed} Gaussian states, we show that the trace norm of the derivative can be bounded in terms of the covariance matrix as follows (see Lemma~\ref{lemma_upp_mixed} in the SM):
\bb
    \left\|  \frac{\mathrm{d}}{\mathrm{d}\alpha}\rho(V+\alpha X,\textbf{m})\Bigg|_{\alpha=0} \right\|_1\leq \frac{1+\sqrt{3}}{4}\Tr[|X|\Omega^\intercal V\Omega]\,,
\ee
which is crucial for the proof of our main result in Theorem~\ref{thm_main000}.

Previous studies have explored the derivatives of specific examples of Gaussian states (e.g.,~\cite{ref_der_mark_1,ref_der_mark_2,ref_der_mark_3}). In contrast, our work introduces a comprehensive framework for defining the derivative of Gaussian states and establishes a general formula applicable to any Gaussian state (see Lemma~\ref{lemma_formula_comp}). Additionally, another work~\cite{Mark_der}, posted after the first version of this work, establishes a formula for the derivative
of a Gaussian state with respect to its first moment and its Hamiltonian matrix~\cite{Mark_der}, which is related to our analysis.

As additional technical contributions of independent interest, in Section~\ref{Sec_proof_LL} in the SM we discover new properties of the \emph{Gaussian noise channel}~\cite[Chapter 5]{BUCCO}, which is widely recognised as one of the most important Gaussian channels. Moreover, we derive an exact formula for the second moment of a quadratic observable on a Gaussian state: 
\begin{lemma}[(Second moment of a quadratic observable on a Gaussian state)]\label{lemma_xx_p_gauss_main}
	Let $\rho(V,\mathbf{m})$ be an $n$-mode Gaussian state with covariance matrix $V$ and first moment $\mathbf{m}$. Let $X$ be an $2n\times2n$ real symmetric matrix. Then, it holds that
	\begin{align}
	&\Tr\!\left[ (\hat{\mathbf{R}}^\intercal X \hat{\mathbf{R}})^2\, \rho(V,\mathbf{m}) \right]\\&  =  \left(\frac{1}{2}\Tr[VX]+\mathbf{m}^\intercal X\mathbf{m}\right)^2 + \frac{1}{2}\Tr[XVXV] \\&\quad+ \frac{1}{2}\Tr[\Omega X\Omega X]+2\mathbf{m}^\intercal XVX\mathbf{m} \,.
\end{align}
\end{lemma}
The proof is detailed in Lemma~\ref{lemma_xx_p_gauss} in the SM. While this formula was previously established for the special case of Gaussian states with vanishing first moments~\cite{Wilde_second_moment}, we present an independent proof that extends its validity to Gaussian states with non-zero first moments.
 
\subsection{Alternative bounds}
\label{sub:alternative}
An improved upper bound, surpassing the one in Theorem~\ref{thm_main000}, can be achieved by applying the \emph{Williamson decomposition} to the covariance matrices~\cite{BUCCO}. Any covariance matrix $V$ can be written in the Williamson decomposition as follows~\cite{BUCCO}:
\bb\label{def_will}
    V=SDS^\intercal\,,
\ee
where $S$ is a \emph{symplectic} matrix, i.e.~a matrix satisfying $S\Omega S^\intercal=\Omega$, and $D$ is a diagonal matrix of the form $D=\bigoplus_{j=1}^n  \left(\begin{matrix}d_j&0\\0&d_j\end{matrix}\right)$, with $d_1,d_2,\ldots , d_n\ge1$.
\begin{thm}[(Improved upper bound)]\label{thm_improvement}
    Let $\rho(V,\mathbf{m}),\rho(W,\mathbf{t})$ be Gaussian states with covariance matrices $V,W$ and first moments $\mathbf{m},\mathbf{t}$, respectively. Moreover, let $V = S_1 D_1 S_1^\intercal$ and $W = S_2 D_2 S_2^\intercal$ be the Williamson decompositions of $V$ and $W$, as in~\eqref{def_will}. Then, the following bound holds:
\bb\label{eq_ineq_main2}
& \frac{1}{2}\left\|\rho(V,\mathbf{m})-\rho(W,\mathbf{t})\right\|_1\\&    
 \quad\le\frac{1+\sqrt{3}}{8}\Tr\!\left[|V-W|\Omega^\intercal\left(\frac{S_1  S_1^\intercal+ S_2 S_2^\intercal}{2}\right)\Omega\right]\\&\quad\quad+\frac{\min(\|S_1\|_\infty,\|S_2\|_\infty)}{\sqrt{2} }\|\mathbf{m}-\mathbf{t}\|_2\,. 
\ee
\end{thm}
The proof can be found in Theorem~\ref{thm_improved_sm} in the SM. Note that the bound in Theorem~\ref{thm_improvement} is strictly stronger than the one in Theorem~\ref{thm_main000} since $S_1 S_1^\intercal \le V$ and $S_2 S_2^\intercal \le W$. Furthermore, the improvements are especially significant for highly mixed Gaussian states. 

Additionally, in the forthcoming Lemma~\ref{bound_trace_distance_thm_main_LL}, proved in Section~\ref{Sec_proof_LL} in the SM, we derive an alternative upper bound that outperforms the above bounds in the regime where the moments are sufficiently far apart. Interestingly, this bound is always finite, even when the distance between the covariance matrices is very large.
\begin{thm}[(Alternative upper bound)] \label{bound_trace_distance_thm_main_LL}
Let $\widebar{\phi}: [0,\infty) \to \R^+$ be a concave function such that $\widebar{\phi}(x) \geq \phi(x)$ for all $x\geq 0$, where the function $\phi: [0,\infty) \to [0,1]$ defined by
\bb
\phi(x) \coloneqq \frac{1-e^{-x}}{2} + e^{-x/2} \sqrt{\sinh^2(x/2) + 1 - e^{-\frac{x^2}{1+4x}}}\, .
\ee
    Let $\rho(V,\mathbf{m})$ and $\rho(W,\mathbf{t})$ be Gaussian states with covariance matrices $V,W$ and first moments $\mathbf{t},\mathbf{m}$, respectively. Then, it holds that
\bb
&\frac12 \left\| \rho(V,\mathbf{m}) - \rho(W,\mathbf{t}) \right\|_1 \\ &\quad\leq   2\, \widebar{\phi}\!\left(\frac18 \max\big\{\|V\|_\infty,\,\|W\|_\infty\big\}\, \|V-W\|_1 \right)\\& \qquad+\sqrt{1 - e^{-\frac12 \min\{\|V\|_\infty,\,\|W\|_\infty\}\, \|\mathbf{t}-\mathbf{m}\|_2^2}}\,.
\ee
\end{thm}
By looking at a plot of the function $\phi$, it seems clear that $\phi$ itself is concave. However, proving this explicitly seems quite difficult. If this could be done, then in the above theorem we could take $\widebar{\phi} \coloneqq \phi$.

\let\addcontentsline\oldaddcontentsline

\clearpage
\fakepart{Supplemental Material}

\onecolumngrid
\begin{center}
\vspace*{\baselineskip}
{\textbf{\large Supplemental Material}}\\
\end{center}

\renewcommand{\theequation}{S\arabic{equation}}
\renewcommand{\thethm}{S\arabic{thm}}
\renewcommand{\thefigure}{S\arabic{figure}}
\setcounter{page}{1}
\makeatletter

\setcounter{secnumdepth}{2}

\tableofcontents

The supplemental material is organised as follows. Section~\ref{sub:prelCV} provides essential preliminaries on continuous variable systems necessary for understanding the proofs of our results. In Section~\ref{Sec_proof_derivative}, we present our results on the upper bounds on the trace distance between Gaussian states, leveraging properties of the \emph{derivative of a Gaussian state}. Section~\ref{sec:proof_diff} rigorously defines the derivative of a Gaussian state from a functional analysis perspective, providing the foundation for certain steps used in Section~\ref{Sec_proof_derivative}. Moreover, Section~\ref{Sec_proof_LL} introduces an alternative, insightful method for obtaining upper bounds on the trace distance between Gaussian states by exploring new properties of the Gaussian noise channel, which may be of independent interest. Finally, in Section~\ref{sec_sample} we apply our trace-distance bounds to improve on the existing bounds on the sample complexity of tomography of Gaussian states~\cite{mele2024learningquantumstatescontinuous},

\section{Preliminaries on continuous variable systems}
\label{sub:prelCV}
In this section, we briefly review the basics of quantum information with continuous variable systems~\cite{BUCCO}. By definition, a continuous variable system is a quantum system associated with the Hilbert space $L^2(\R^n)$, which comprises all square-integrable complex-valued functions over $ \R^n$. Since $L^2(\R^n) = (L^2(\R))^{\otimes n}$, a continuous variable system can be viewed as comprising $n$ subsystems, or \emph{modes}. In this setting, $n$ is called the \emph{number of modes}, and a quantum state on $L^2(\R^n)$ is referred to as an \emph{$n$-mode state}.

Let $\hat{x}_1,\hat{p}_1,\dots,\hat{x}_n,\hat{p}_n$ be the well-known position and momentum operators of each mode~\cite{BUCCO}. The \emph{quadrature operator vector} is defined as
\bb
    \mathbf{\hat{R}}\coloneqq (\hat{x}_1,\hat{p}_1,\dots,\hat{x}_n,\hat{p}_n)^{\intercal}
\ee
and it satisfies the following relation, known as \emph{canonical commutation relation}:
\bb
    [\mathbf{\hat{R}},\mathbf{\hat{R}}^{\intercal}]=i\,\Omega\,\mathbb{\hat{1}}\,,
\ee
where 
\bb
    \Omega\coloneqq\bigoplus_{i=1}^n \left(\begin{matrix}0&1\\-1&0\end{matrix}\right)\,.
\ee
The characteristic function $\chi_\Theta: \R^{2n}\to \mathbb{C}$ of a bounded linear operator $\Theta$ on $L^2(\R^n)$ is given by 
\bb 
    \chi_\Theta(\mathbf{r})\coloneqq \Tr[ \Theta  \hat{D}_{\mathbf{r}} ] \, , 
\ee
where the displacement operator is defined as 
\bb
\hat{D}_{\mathbf{r}}\coloneqq e^{-i {\mathbf{r}}^{\intercal}\Omega \mathbf{\hat{R}}}\,.
\ee
Any bounded linear operator $\Theta$ on $L^2(\R^n)$ can be expressed in terms of its characteristic function as~\cite{BUCCO,Groenewold_1946}:
\bb\label{eq_fourier_weyl}
\Theta = \mathcal{W}\left(\chi_{\Theta}(\mathbf{r})\right) \coloneqq \int_{\R^{2n}}\frac{\mathrm{d}\mathbf{r}}{(2\pi)^n}\chi_\Theta(\mathbf{r}) \hat{D}_{\mathbf{r}}^\dagger \,,
\ee
which is known as the \emph{Fourier-Weyl transform}. This establishes a one-to-one correspondence between bounded operators and their characteristic functions. The first moment $\mathbf{m}(\rho)$ and the covariance matrix $V(\rho)$ of a quantum state $\rho$ are defined as follows:
\bb 
	\mathbf{m}(\rho)&\coloneqq \Tr\!\left[\mathbf{\hat{R}}\,\rho\right]\,,\\
	V(\rho)&\coloneqq\Tr\!\left[\left\{\mathbf{(\hat{R}-m(\rho)\,\hat{\mathbb{1}}),(\hat{R}-m(\rho)\,\hat{\mathbb{1}})}^{\intercal}\right\}\rho\right]\, ,
\ee
where $(\cdot)^\intercal$ denotes the transpose operation, $\{\hat{A},\hat{B}\} \coloneqq \hat{A}\hat{B}+\hat{B}\hat{A}$ denotes the anti-commutator. Any covariance matrix $V\!(\rho)$ satisfies the following matrix inequality, known as \emph{uncertainty relation}:
\bb
V\!(\rho)+i\Omega\ge0\,.
\ee
As a consequence, since $\Omega$ is skew-symmetric, any covariance matrix $V\!(\rho)$ is positive semi-definite on $\R^{2n}$. Conversely, for any symmetric $2n\times 2n$ matrix $W$ such that $W+i\Omega\ge0$ there exists an $n$-mode (Gaussian) state $\rho$ with covariance matrix $V\!(\rho)=W$~\cite{BUCCO}. A displacement transformation shifts the first moments without affecting the covariance matrices, i.e.
\bb
    &\mathbf{m}\!\left(\hat{D}_\mathbf{r}\rho \hat{D}_\mathbf{r}^\dagger\right)=\mathbf{m}(\rho)+\mathbf{r}\,,\\
    &V\!\left(\hat{D}_\mathbf{r}\rho \hat{D}_\mathbf{r}^\dagger\right)=V\!(\rho)\,,
\label{eq:bille}
\ee
as it follows from the fact that
\begin{align}
    \label{eq:displ}\hat{D}_\mathbf{r}\mathbf{\hat{R}}\hat{D}_\mathbf{r}^\dagger=\mathbf{\hat{R}}+\mathbf{r}\mathbb{\hat{1}}\,.
\end{align} 
Let us proceed with the definition of Gaussian states.
\begin{Def}[(Gaussian state)]\label{def_gauss_sm}
An $n$-mode state $\rho$ is said to be a Gaussian state if it can be written as a Gibbs state of a quadratic Hamiltonian $\hat{H}$ in the quadrature operator vector:
\bb
    \hat{H}\coloneqq \frac{1}{2}(\mathbf{\hat{R}}-\mathbf{m})^{\intercal}H(\mathbf{\hat{R}}-\mathbf{m})
\ee
for some $2n\times 2n$ symmetric positive-definite matrix $H$ and some vector $\mathbf{m}\in\R^{2n}$. The Gibbs states associated with the Hamiltonian $\hat{H}$ are given by 
\begin{equation}
\rho= \left(\frac{e^{-\beta \hat{H}}}{\Tr[e^{-\beta \hat{H}}]}\right)_{\beta\in(0,\infty]}\,,
\end{equation}
where the parameter $\beta$ is called the `inverse temperature'.
\end{Def}
\begin{remark}
Definition~\ref{def_gauss_sm} also covers pathological cases where both \( \beta \) and specific terms in \( H \) diverge (for instance, in tensor products of pure and mixed Gaussian states). To formalise this rigorously, one can define the set of Gaussian states as the trace-norm closure of the set of Gibbs states of quadratic Hamiltonians~\cite{G-resource-theories}.
\end{remark} 
The characteristic function of a Gaussian state $\rho$ is the Fourier transform of a Gaussian probability distribution, evaluated at $\Omega \mathbf{r}$, which can be written in terms of $\mathbf{m}(\rho) $ and $V\!(\rho)$ as~\cite{BUCCO}
\bb\label{eq:charact_Gauss_Prel}
\chi_{\rho}(\mathbf{r})=\exp\!\left( -\frac{1}{4}(\Omega \mathbf{r})^{\intercal}V\!(\rho)\Omega \mathbf{r}+i(\Omega \mathbf{r})^{\intercal}\mathbf{m}(\rho) \right)\,.
\ee
Since any quantum state is uniquely identified by its characteristic function, it follows that any Gaussian state is uniquely identified by its first moment and covariance matrix.

Notably, any covariance matrix $V$ can be written in the so-called \emph{Williamson decomposition} as follows~\cite{BUCCO}:
\bb\label{eq:will}
    V=SDS^\intercal\,,
\ee
where $S$ is a \emph{symplectic} matrix, i.e.~a matrix satisfying $S\Omega S^\intercal=\Omega$, and $D$ is a diagonal matrix of the form $D=\bigoplus_{j=1}^n  \left(\begin{matrix}d_j&0\\0&d_j\end{matrix}\right)$, with $d_1,d_2,\ldots , d_n\ge1$ being the so-called \emph{symplectic eigenvalues}.

A Gaussian state $\rho$ is pure if and only if all the symplectic eigenvalues of its covariance matrix $V(\rho)$ are exactly one~\cite{BUCCO}. In particular, the Williamson decomposition of the covariance matrix of a pure Gaussian state $\psi$ reads $V(\psi)=SS^\intercal$, for some symplectic matrix $S$.
\subsection{Gaussian noise channel}\label{subsubsec_gauss_noise}
The \emph{Gaussian noise channel} is a Gaussian channel that will turn out to be crucial for our analysis. For any $2n\times 2n$ positive semi-definite matrix $K$, the \emph{Gaussian noise channel} $\NN_K$ is the convex combination of displacement transformations $(\cdot)\mapsto \hat{D}_u(\cdot)\hat{D}_u^\dagger$, where $u\in\R^{2n}$ is distributed according to a Gaussian probability distribution with vanishing mean value and covariance matrix equal to $K$~\cite[Chapter 5]{BUCCO}. Mathematically, $\NN_K$ can be written as
\bb\label{Gaussian_noise}
\NN_K(\Theta) \coloneqq \int_{\R^{2n}} \mathrm{d}u\ P_K(u)\, \hat{D}_u \Theta \hat{D}_{u}^\dagger\,,
\ee
where $P_K(u)$ is a Gaussian probability distribution with vanishing mean and covariance matrix equal to $K$, and $\hat{D}_{u}\coloneqq e^{i u^{\intercal}\Omega \mathbf{\hat{R}}}$ is the displacement operator. Of course, in the case when $K$ is strictly positive, $P_K(u)$ can be written as
\bb\label{def_pku}
    P_K(u)\coloneqq \frac{e^{-\frac12 u^\intercal K^{-1} u}}{(2\pi)^{n}\sqrt{\det K}}\,.
\ee
Instead, when $K$ has some zero eigenvalues, the definition of $P_K(u)$ involves Dirac deltas as follows. Let $K=\sum_{i=1}^{2n}\lambda_iv_iv_i^\intercal$ be a spectral decomposition of $K$, where $(v_i)_{i=1,\ldots,2n}$ are orthonormal eigenvectors, while $\lambda_1,\ldots,\lambda_r>0$ and $\lambda_{r+1},\ldots,\lambda_{2n}=0$ are its eigenvalues, with $r$ being the rank of $K$. Then, $P_K(u)$ can be written as
\bb
    P_K(u)\coloneqq \frac{e^{-\frac12 u^\intercal\left(\sum_{i=1}^{r}\lambda_i^{-1} v_i v_i^\intercal \right)u
 }}{(2\pi)^{(2n-r)/2}\sqrt{\lambda_1\ldots\lambda_r}}\prod_{i=r+1}^{2n}\delta(v_i^\intercal u)\,,
\ee
where $\delta(\cdot)$ denotes the Dirac delta distribution.

It turns out that $\NN_K$ is a Gaussian channel that leaves the first moments unchanged, and it acts on the covariance matrix of the input state by adding $K$, as stated by the following Lemma~\cite[Chapter 5]{BUCCO}. 
\begin{lemma}\emph{(\cite[Chapter 5]{BUCCO})}\label{lemma_add_gauss}
    Let $K\in\R$ be a $2n\times 2n$ positive semi-definite matrix. The Gaussian noise channel $\NN_K$, defined in~\eqref{Gaussian_noise}, is a Gaussian channel that acts on the first moments and covariance matrices as 
    \begin{eqnarray}
        \mathbf{m}\!\left(\NN_K(\rho)\right)&=&\mathbf{m}\!\left(\rho\right)\,,\\
        V\!\left(\NN_K(\rho)\right)&=& V\!(\rho)+K\,\,,
    \end{eqnarray}
    for any $n$-mode input state $\rho$.
\end{lemma}

Lemma~\ref{lemma_add_gauss} directly implies the following useful known result~\cite{BUCCO}.
\begin{lemma}\label{lemma_mixed_gauss}
    Let $\rho$ be a Gaussian state and let us consider the Williamson decomposition of its covariance matrix, $V(\rho)=SDS^\intercal$. Additionally, let $\psi$ be the Gaussian state with the same first moment as $\rho$ and with covariance matrix $SS^\intercal$. Then, $\rho$ can be written as the output of the Gaussian noise channel $\NN_{S(D-\mathbb{1})S^\intercal}$ applied to $\psi$:
    \bb
        \rho=\NN_{S(D-\mathbb{1})S^\intercal}(\psi)\,.
    \ee
    In particular, any mixed Gaussian state with covariance matrix $V$ can be written as a convex combination of pure Gaussian states with all the same covariance matrix $W$, which satisfies the matrix inequality $W\le V$.
\end{lemma}

\section{Proof of our results}\label{Sec_proof_derivative}
In this section, we address the problem of deriving upper bounds on the trace distance between Gaussian states in terms of the norm distance between their first moments and covariance matrices. The section is organised as follows:
\begin{itemize}
    \item In Subsection~\ref{Subsec_exact_exp}, we derive an elegant formula for the second moment of a quadratic observable on a Gaussian state, expressed solely in terms of the covariance matrix and first moment.
    \item In Subsection~\ref{Subsec_van_third}, we prove that all odd moments of a Gaussian state with vanishing first moment are zero.
    \item In Subsection~\ref{Sec_Der}, we introduce the concept of the \emph{derivative of a Gaussian state} and establish some useful properties. Specifically, we show that: (i) this derivative can be represented through a compact formula involving commutators between quadrature operators and the Gaussian state; (ii) we derive a simple upper bound on the trace norm of the derivative of a Gaussian state in terms of the first moment and covariance matrix; and (iii) we find an exact formula for the trace norm of the derivative of a \emph{pure} Gaussian state in terms of its first moment and covariance matrix.
    \item In Subsection~\ref{Sec_upppp_bounds}, leveraging the tools developed in the previous subsections, we prove upper bounds on the trace distance between Gaussian states and demonstrate that these bounds are tight.

\end{itemize}

\subsection{Exact formula for the second moment of a quadratic observable on a Gaussian state}\label{Subsec_exact_exp}
In this subsection, we derive a closed formula for the second moment of a quadratic observable on a Gaussian state in terms of its covariance matrix and its first moment. This formula was already known in the case of Gaussian states with vanishing first moment~\cite{Wilde_second_moment}. Here, we report an alternative proof.
\begin{lemma}[(Second moment of a quadratic observable on a Gaussian state)]\label{lemma_xx_p_gauss}
	Let $\rho(V,\mathbf{m})$ be an $n$-mode Gaussian state with covariance matrix $V$ and first moment $\mathbf{m}$. Let $X$ be an $2n\times2n$ real symmetric matrix. Then, it holds that
	\begin{equation}
	\Tr\!\left[ (\hat{\mathbf{R}}^\intercal X \hat{\mathbf{R}})^2\, \rho(V,\mathbf{m}) \right]  =  \left(\frac{1}{2}\Tr[VX]+\mathbf{m}^\intercal X\mathbf{m}\right)^2 + \frac{1}{2}\Tr[XVXV] + \frac{1}{2}\Tr[\Omega X\Omega X]+2\mathbf{m}^\intercal XVX\mathbf{m} \,.
	\end{equation}
\end{lemma}
Before proving Lemma~\ref{lemma_xx_p_gauss}, let us establish some preliminary results.

\begin{lemma}\label{lemma_xx_p}
	Let $\rho$ be an $n$-mode state and let $X$ be a $2n\times2n$ real symmetric matrix. Then, it holds that
	\bb
	\Tr\!\left[ (\hat{\mathbf{R}}^\intercal X \hat{\mathbf{R}})^2 \rho \right]  =  \frac{1}{8}\sum_{i,j,l,k=1}^{2n}X_{ij}X_{kl}\Tr\!\left[ \{\hat{R}_i,\{\hat{R}_j,\{\hat{R}_k,\hat{R}_l\} \}\}\rho \right]   + \frac{1}{2}\Tr[\Omega X\Omega X] 
	\ee
\end{lemma}
\begin{proof}
Note that
\bb
    \Tr\!\left[ (\hat{\mathbf{R}}^\intercal X \hat{\mathbf{R}})^2 \rho \right]  = \sum_{i,j,l,k=1}^{2n}X_{ij}X_{kl}\Tr\!\left[ \hat{R}_i\hat{R}_j\hat{R}_k\hat{R}_l \rho \right] \,. 
\ee
By exploiting that $X$ is symmetric, we have that
\bb\label{eq_gr}
\sum_{i,j,l,k=1}^{2n}X_{ij}X_{kl}\Tr\!\left[ \{\hat{R}_i,\{\hat{R}_j,\{\hat{R}_k,\hat{R}_l\} \}\}\rho \right]   
	&=2\sum_{i,j,l,k=1}^{2n}X_{ij}X_{kl}\Tr\!\left[ \{\hat{R}_i,\{\hat{R}_j,\hat{R}_k\hat{R}_l \}\}\rho \right]   \\
	&=2\sum_{i,j,l,k=1}^{2n}X_{ij}X_{kl}\Tr\!\left[ \hat{R}_i\hat{R}_j\hat{R}_k\hat{R}_l \rho \right]+2\sum_{i,j,l,k=1}^{2n}X_{ij}X_{kl}\Tr\!\left[  \hat{R}_i\hat{R}_k\hat{R}_l\hat{R}_j
	\rho \right]   \\
	&\,\quad+2\sum_{i,j,l,k=1}^{2n}X_{ij}X_{kl}\Tr\!\left[
	\hat{R}_j\hat{R}_k\hat{R}_l\hat{R}_i  \rho \right] +2\sum_{i,j,l,k=1}^{2n}X_{ij}X_{kl}\Tr\!\left[ \hat{R}_k\hat{R}_l\hat{R}_j \hat{R}_i\rho \right]   \\
	&=4\sum_{i,j,l,k=1}^{2n}X_{ij}X_{kl}\Tr\!\left[ \hat{R}_i\hat{R}_j\hat{R}_k\hat{R}_l \rho \right]+4\sum_{i,j,l,k=1}^{2n}X_{ij}X_{kl}\Tr\!\left[  \hat{R}_i\hat{R}_k\hat{R}_l\hat{R}_j
	\rho \right]   \,.
\ee
Moreover, the canonical commutation relation implies that
\bb
    \Tr\!\left[  \hat{R}_i\hat{R}_k\hat{R}_l\hat{R}_j
	\rho \right]&=   \Tr\!\left[  \hat{R}_i\hat{R}_k\hat{R}_j\hat{R}_l
	\rho \right]+i\Omega_{lj}\Tr\!\left[  \hat{R}_i\hat{R}_k 
	\rho \right]\\
 &=   \Tr\!\left[  \hat{R}_i\hat{R}_j\hat{R}_k\hat{R}_l
	\rho \right]+i\Omega_{kj}\Tr\!\left[  \hat{R}_i\hat{R}_l
	\rho \right]+i\Omega_{lj}\Tr\!\left[  \hat{R}_i\hat{R}_k 
	\rho \right]\,.
\ee
Using the fact that $X$ is symmetric and $\Omega$ is skew-symmetric, we get
\bb
    i\sum_{i,j,l,k=1}^{2n}X_{ij}X_{kl}\left(\Omega_{kj}\Tr\!\left[  \hat{R}_i\hat{R}_l
	\rho \right]+\Omega_{lj}\Tr\!\left[  \hat{R}_i\hat{R}_k 
	\rho \right]\right)&=i\sum_{i,j,l,k=1}^{2n}X_{ij}X_{kl}\left(\Omega_{jl}\Tr\!\left[  \hat{R}_k\hat{R}_i
	\rho \right]+\Omega_{lj}\Tr\!\left[  \hat{R}_i\hat{R}_k 
	\rho \right]\right)\\
 &=i\sum_{i,j,l,k=1}^{2n}X_{ij}X_{kl}\Omega_{jl}\left(\Tr\!\left[  [\hat{R}_k,\hat{R}_i]
	\rho \right]\right)\\
 &=-\sum_{i,j,l,k=1}^{2n}X_{ij}X_{kl}\Omega_{jl}\Omega_{ki}\\
 &=-\Tr[X\Omega X\Omega ]\,.
\ee
Consequently, we have that
\bb
    \sum_{i,j,l,k=1}^{2n}X_{ij}X_{kl}\Tr\!\left[  \hat{R}_i\hat{R}_k\hat{R}_l\hat{R}_j
	\rho \right]=\sum_{i,j,l,k=1}^{2n}X_{ij}X_{kl}\Tr\!\left[  \hat{R}_i\hat{R}_j\hat{R}_k\hat{R}_l
	\rho \right]-\Tr[X\Omega X\Omega ]
\ee
By substituting in~\eqref{eq_gr}, we finally obtain 
\bb
    \sum_{i,j,l,k=1}^{2n}X_{ij}X_{kl}\Tr\!\left[ \{\hat{R}_i,\{\hat{R}_j,\{\hat{R}_k,\hat{R}_l\} \}\}\rho \right]   = 8\sum_{i,j,l,k=1}^{2n}X_{ij}X_{kl}\Tr\!\left[  \hat{R}_i\hat{R}_j\hat{R}_k\hat{R}_l
	\rho \right]-4\Tr[X\Omega X\Omega ]\,,
\ee
which, after rearranging, proves the desired result.
\end{proof}

\begin{lemma}\label{lemma_anticomm_dr}
Let \( j \in \{1,2,\ldots,2n\} \). 
The partial derivative of the displacement operator \( \hat{D}_{\mathbf{r}} \coloneqq e^{-i \mathbf{r}^\intercal \Omega \hat{R}} \) with respect to the variable $r_j$ is given by
\begin{align}
    \frac{\partial}{\partial r_{j}} \hat{D}_{\mathbf{r}} = -\frac{i}{2} \{ (\Omega \hat{\mathbf{R}})_{j}, \hat{D}_{\mathbf{r}} \}.
\end{align}
\end{lemma}
\begin{proof}
Since the displacement operator is the following tensor product of single-mode displacement operators
\begin{align}
    \hat{D}_\mathbf{r} &\coloneqq \hat{D}_{(r_1,r_2)}\otimes\hat{D}_{(r_3,r_4)}\otimes\ldots\otimes\hat{D}_{(r_{2n-1},r_{2n})}\,,
\end{align}
it suffices to consider the single-mode case $n=1$. By exploiting that 
	\bb
	e^{A+B}=e^Ae^Be^{ -\frac{1}{2}[A,B] }\,.
	\ee
for any $A, B$ such that $[A, B] \propto \mathbb{1}$, we obtain that
	\bb
	\frac{\partial}{\partial r_{1}} \hat{D}_{\mathbf{r}}&=\frac{\partial}{\partial r_{1}}\left[ e^{-i r_1\hat{R}_2+i r_2\hat{R}_1}   \right]\\
	&=\frac{\partial}{\partial r_{1}}\left[ e^{-i r_1\hat{R}_2}e^{i r_2\hat{R}_1}e^{-i\frac{1}{2}r_1r_2   }   \right]\\
	&=-i\hat{R}_2 \hat{D}_{\mathbf{r}}-i\frac{1}{2}r_2\hat{D}_{\mathbf{r}}\,.
	\ee
Similarly, by using that $e^{A+B}=e^Be^Ae^{ \frac{1}{2}[A,B] }$ for any $A, B$ such that $[A, B] \propto \mathbb{1}$, we also obtain that
	\bb
	\frac{\partial}{\partial r_{1}} \hat{D}_{\mathbf{r}}&=\frac{\partial}{\partial r_{1}}\left[ e^{-i r_1\hat{R}_2+i r_2\hat{R}_1}   \right]\\
	&=\frac{\partial}{\partial r_{1}}\left[ e^{i r_2\hat{R}_1}e^{-i r_1\hat{R}_2}e^{i\frac{1}{2}r_1r_2   }   \right]\\
	&=-i \hat{D}_{\mathbf{r}}\hat{R}_2 +i\frac{1}{2}r_2\hat{D}_{\mathbf{r}}\,.
	\ee
	By adding both the above identities, we obtain that 
	\bb\label{exxxx1}
	\frac{\partial}{\partial r_{1}} \hat{D}_{\mathbf{r}} =-\frac{i}{2}\{\hat{R}_2 ,\hat{D}_{\mathbf{r}}\}\,.
	\ee
	Analogously, we obtain that
	\bb
	\frac{\partial}{\partial r_{2}} \hat{D}_{\mathbf{r}}&=\frac{\partial}{\partial r_{2}}\left[ e^{-i r_1\hat{R}_2+i r_2\hat{R}_1}   \right]\\
	&=\frac{\partial}{\partial r_{2}}\left[ e^{-i r_1\hat{R}_2}e^{i r_2\hat{R}_1}e^{-i\frac{1}{2}r_1r_2   }   \right]\\
    &=i \hat{D}_{\mathbf{r}}\hat{R}_1-i\frac{1}{2}r_1\hat{D}_{\mathbf{r}}	
	\ee
	and that
	\bb
	\frac{\partial}{\partial r_{2}} \hat{D}_{\mathbf{r}}&=\frac{\partial}{\partial r_{2}}\left[ e^{-i r_1\hat{R}_2+i r_2\hat{R}_1}   \right]\\
	&=\frac{\partial}{\partial r_{2}}\left[ e^{i r_2\hat{R}_1}e^{-i r_1\hat{R}_2}e^{i\frac{1}{2}r_1r_2   }   \right]\\
	&=i \hat{R}_1 \hat{D}_{\mathbf{r}} +i\frac{1}{2}r_1\hat{D}_{\mathbf{r}}\,.
	\ee
By adding these two expressions, we get 
 \begin{align}\label{exxxx2}
    \frac{\partial}{\partial r_{2}} \hat{D}_{\mathbf{r}} =\frac{i}{2}\{\hat{R}_1,\hat{D}_{\mathbf{r}} \}\,.
 \end{align}
Note that~\eqref{exxxx1} and~\eqref{exxxx2} can be rewritten in a compact form as 
\bb
    \frac{\partial}{\partial r_{j}} \hat{D}_{\mathbf{r}} = -\frac{i}{2} \{ (\Omega \hat{\mathbf{R}})_{j}, \hat{D}_{\mathbf{r}} \}\,.
\ee
This concludes the proof.
\end{proof} 
As an example, Lemma~\ref{lemma_anticomm_dr} is useful to calculate the laplacian $\triangle\coloneqq \sum_{j=1}^{2n}\frac{\partial^2}{\partial r_j^2}$ of the displacement operator. Indeed, by exploiting such a lemma, one can prove that 
\bb
    \triangle \hat{D}_{\textbf{r}}=-\left( \hat{\textbf{R}}^\intercal\hat{\textbf{R}}-\textbf{r}^\intercal\hat{\textbf{R}}+\frac14 \textbf{r}^\intercal\textbf{r}\,\mathbb{1}\right)\hat{D}_{\textbf{r}}\,.
\ee
Moreover, Lemma~\ref{lemma_anticomm_dr} directly implies a simple expression for the derivative of the characteristic function. 
\begin{lemma}[(Derivative of a characteristic function)]\label{der_charact_general}
Let $\Theta$ be a linear operator on $L^2(\mathbb{R}^n)$ and let \( j \in \{1,2,\ldots,2n\} \).  The partial derivative of the characteristic function $\chi_{\Theta}(\mathbf{r}) \coloneqq \Tr[\Theta e^{-i \mathbf{r}^\intercal \Omega \hat{R}}]$ with respect to the variable $r_j$ is given by
\begin{align}
    \frac{\partial}{\partial r_{j}} \chi_{\Theta}(\mathbf{r})= -\frac{i}{2} \chi_{\{ \Theta, (\Omega \hat{\mathbf{R}})_{j} \}}(\textbf{r})\qquad\forall\,\textbf{r}\in\mathbb{R}^{2n}\,,
\end{align}
i.e.~it is proportional to the characteristic function of the anti-commutator between $\Theta$ itself and the quadrature $(\Omega \hat{\textbf{R}})_j$.
\end{lemma}
Lemma~\ref{der_charact_general} is  useful in order to calculate (symmetrised) moments of a quantum state given its characteristic function. This is true because the characteristic function of an operator evaluated in zero is exactly equal to the trace of the operator. As an example, Lemma~\ref{der_charact_general} implies that
\bb\label{nested_anticom}
    \Tr\!\left[\{(\Omega \hat{\mathbf{R}})_k,\{(\Omega \hat{\mathbf{R}})_l, \Theta  \}\}\right] =-4\,\frac{\partial}{\partial r_{k}} \frac{\partial}{\partial r_{l}}\chi_{\Theta}(0)\,. 
\ee
Note that, in general, the moment \(\Tr[\Theta \hat{R}_k \hat{R}_l]\) cannot be proportional to the second derivative \(\frac{\partial}{\partial r_k} \frac{\partial}{\partial r_l} \chi_{\Theta}(0)\). This discrepancy arises because \(\Tr[\Theta \hat{R}_k \hat{R}_l]\) lacks symmetry under the exchange \(k \leftrightarrow l\), whereas the second derivative \(\frac{\partial}{\partial r_k} \frac{\partial}{\partial r_l} \chi_{\Theta}(0)\) is symmetric under such an exchange (as partial derivatives commute). Consequently, a fully symmetrised structure --- such as one involving a nested anti-commutator as in~\eqref{nested_anticom} --- is necessary.

\begin{lemma}\label{lemma_4_moment}
    Let $\rho(V,\mathbf{m})$ be a Gaussian state with covariance matrix $V$ and first moment $\mathbf{m}$. Then, for any $i,j,k,l\in\{1,2,\ldots,2n\}$, it holds that
    \bb
        \Tr[\rho(V,\mathbf{m})\,\{\hat{R}_i,\{\hat{R}_j,\{\hat{R}_k,\hat{R}_l\} \}\}  ] &=2( V_{ij} V_{kl}+ V_{ik}  V_{jl}+ V_{il}  V_{jk})\\
 &\quad+4\left({V}_{kl}{m}_i{m}_j+{V}_{jl}{m}_i{m}_k+{V}_{kj}{m}_i{m}_l+{V}_{il}{m}_j{m}_k+{V}_{ik}{m}_j{m}_l+{V}_{ij}{m}_k{m}_l\right)\\
 &\quad+8{m}_i{m}_j{m}_k{m}_l\,.
    \ee
\end{lemma}
\begin{proof}
    On the one hand, by performing an explicit calculation of the partial derivatives of the characteristic function 
    \bb
        \chi_{\rho(V,\mathbf{m})}(\mathbf{r})=e^{ -\frac{1}{4}(\Omega \mathbf{r})^{\intercal}V\Omega \mathbf{r} }\,,
    \ee
    we obtain
    \bb
         \frac{\partial}{\partial r_{i}}\frac{\partial}{\partial r_{j}}\frac{\partial}{\partial r_{k}}\frac{\partial}{\partial r_{l}}\chi_{\rho(V,0)}(\mathbf{r})|_{\mathbf{r}=0}&= \frac{1}{4}(\tilde V_{ij}\tilde V_{kl}+\tilde V_{ik} \tilde V_{jl}+\tilde V_{il} \tilde V_{jk})\\
 &\quad+\frac{1}{2}\left(\tilde{V}_{kl}\tilde{m}_i\tilde{m}_j+\tilde{V}_{jl}\tilde{m}_i\tilde{m}_k+\tilde{V}_{kj}\tilde{m}_i\tilde{m}_l+\tilde{V}_{il}\tilde{m}_j\tilde{m}_k+\tilde{V}_{ik}\tilde{m}_j\tilde{m}_l+\tilde{V}_{ij}\tilde{m}_k\tilde{m}_l\right)\\
 &\quad+\tilde{m}_i\tilde{m}_j\tilde{m}_k\tilde{m}_l\,,
	\ee
	where 
    \bb
	\tilde{V}&\coloneqq   \Omega^{\intercal}V\Omega \,,\\
        \tilde{\mathbf{m}}&\coloneqq \Omega\mathbf{m}\,.
	\ee
In particular, it also holds that
\bb\label{eq_part1}
	&\sum_{i,j,k,l=1}^n\Omega_{ii'}\Omega_{jj'}\Omega_{kk'}\Omega_{ll'}\frac{\partial}{\partial r_{i'}}\frac{\partial}{\partial r_{j'}}\frac{\partial}{\partial r_{k'}}\frac{\partial}{\partial r_{l'}}\chi_{\rho(V,\mathbf{m})}(\mathbf{r})|_{\mathbf{r}=0}\\
 &\quad=\frac{1}{4}( V_{ij} V_{kl}+ V_{ik}  V_{jl}+ V_{il}  V_{jk})\\
 &\qquad+\frac{1}{2}\left({V}_{kl}{m}_i{m}_j+{V}_{jl}{m}_i{m}_k+{V}_{kj}{m}_i{m}_l+{V}_{il}{m}_j{m}_k+{V}_{ik}{m}_j{m}_l+{V}_{ij}{m}_k{m}_l\right)\\
 &\qquad+{m}_i{m}_j{m}_k{m}_l\,.
	\ee 
    On the other hand, Lemma~\ref{lemma_anticomm_dr} implies that
	\bb
	\frac{\partial}{\partial r_{i}}\frac{\partial}{\partial r_{j}}\frac{\partial}{\partial r_{k}}\frac{\partial}{\partial r_{l}}\chi_{\rho(V,0)}(\mathbf{r})&=\Tr\!\left[\rho(V,0)\frac{\partial}{\partial r_{i}}\frac{\partial}{\partial r_{j}}\frac{\partial}{\partial r_{k}}\frac{\partial}{\partial r_{l}}\hat{D}_{\mathbf{r}}\right]\\
    &=\frac{1}{16}\Tr\!\left[\rho(V,0)\,\{(\Omega \hat{\mathbf{R}})_i,\{(\Omega \hat{\mathbf{R}})_j,\{(\Omega \hat{\mathbf{R}})_k,\{(\Omega \hat{\mathbf{R}})_l, \hat{D}_{\mathbf{r}}  \}\}\}\}   \right]\,.
\ee
Consequently, by exploiting that $\hat{D}_0=\hat{\mathbb{1}}$, we obtain that
	\bb
	\frac{\partial}{\partial r_{i}}\frac{\partial}{\partial r_{j}}\frac{\partial}{\partial r_{k}}\frac{\partial}{\partial r_{l}}\chi_{\rho(V,0)}(\mathbf{r})|_{\mathbf{r}=0}&=\frac{1}{8}\Tr[\rho(V,0)\,\{(\Omega \hat{\mathbf{R}})_i,\{(\Omega \hat{\mathbf{R}})_j,\{(\Omega \hat{\mathbf{R}})_k,(\Omega \hat{\mathbf{R}})_l\} \}\}  ]\,.
\ee
In particular, we get
\bb\label{eq_part2}
	\sum_{i,j,k,l=1}^n\Omega_{ii'}\Omega_{jj'}\Omega_{kk'}\Omega_{ll'}\frac{\partial}{\partial r_{i'}}\frac{\partial}{\partial r_{j'}}\frac{\partial}{\partial r_{k'}}\frac{\partial}{\partial r_{l'}}\chi_{\rho(V,0)}(\mathbf{r})|_{\mathbf{r}=0}&= \frac{1}{8}\Tr[\rho(V,0)\,\{\hat{R}_i,\{\hat{R}_j,\{\hat{R}_k,\hat{R}_l\} \}\}  ]  \,.
\ee 
By putting~\eqref{eq_part1} and~\eqref{eq_part2} together, we conclude the proof.
\end{proof}

Now, we are now ready to prove Lemma~\ref{lemma_xx_p_gauss}.
\begin{proof}[Proof of Lemma~\ref{lemma_xx_p_gauss}]
     It holds that 
     \bb
	\Tr\!\left[ (\hat{\mathbf{R}}^\intercal X \hat{\mathbf{R}})^2 \rho(V,\mathbf{m}) \right]  &\eqt{(i)} \frac{1}{2}\Tr[\Omega X\Omega X]+ \frac{1}{8}\sum_{i,j,l,k=1}^{2n}X_{ij}X_{kl}\Tr\!\left[ \{\hat{R}_i,\{\hat{R}_j,\{\hat{R}_k,\hat{R}_l\} \}\}\rho(V,\mathbf{m}) \right]    \\
 &\eqt{(ii)}\frac{1}{2}\Tr[\Omega X\Omega X]+ \sum_{i,j,l,k=1}^{2n}X_{ij}X_{kl}\Bigg(\frac{1}{4}( V_{ij} V_{kl}+ V_{ik}  V_{jl}+ V_{il}  V_{jk})\\
 &\qquad+\frac{1}{2}\left({V}_{kl}{m}_i{m}_j+{V}_{jl}{m}_i{m}_k+{V}_{kj}{m}_i{m}_l+{V}_{il}{m}_j{m}_k+{V}_{ik}{m}_j{m}_l+{V}_{ij}{m}_k{m}_l\right)\\
 &\qquad+{m}_i{m}_j{m}_k{m}_l\Bigg)\\
 &=\frac{1}{2}\Tr[\Omega X\Omega X]+ \frac{1}{4}(\Tr[VX])^2 + \frac{1}{2}\Tr[XVXV] \\
 &\quad + \mathbf{m}^\intercal X\mathbf{m}\Tr[VX]+2\mathbf{m}^\intercal XVX\mathbf{m}+(\mathbf{m}^\intercal X\mathbf{m})^2\\
 &=  \left(\frac{1}{2}\Tr[VX]+\mathbf{m}^\intercal X\mathbf{m}\right)^2 + \frac{1}{2}\Tr[XVXV] + \frac{1}{2}\Tr[\Omega X\Omega X]+2\mathbf{m}^\intercal XVX\mathbf{m} \,.
\ee
Here, in (i), we exploited Lemma~\ref{lemma_xx_p}, and in (ii), we applied Lemma~\ref{lemma_4_moment}.
 \end{proof}

\subsection{A Gaussian state with vanishing first moment also has vanishing odd moments}\label{Subsec_van_third}
Thanks to the tools regarding developed above, we can easily show the following rather intuitive result.
\begin{lemma}[(A Gaussian state with vanishing first moment also has vanishing third moments)]\label{lemma_3_moment}
    Let $\rho(V,0)$ be a Gaussian state with zero first moment and covariance matrix $V$. Then, for any $j,k,l\in\{1,2,\ldots,2n\}$, it holds that
    \bb
    \Tr[\rho(V,0)\,\hat{R}_j\hat{R}_k\hat{R}_l   ]&=0\,.
    \ee 
\end{lemma}
\begin{proof}
    On the one hand, since the characteristic function of $\rho(V,0)$ is given by
    \bb
        \chi_{\rho(V,0)}(\mathbf{r})=e^{ -\frac{1}{4}(\Omega \mathbf{r})^{\intercal}V\Omega \mathbf{r} }\,,
    \ee
    it holds that
    \bb\label{eq_part1m}
    \frac{\partial}{\partial r_{j}}\frac{\partial}{\partial r_{k}}\frac{\partial}{\partial r_{l}}\chi_{\rho(V,0)}(\mathbf{r})|_{\mathbf{r}=0}&= 0\,.
\ee
On the other hand, Lemma~\ref{lemma_anticomm_dr} implies that
\bb
    \frac{\partial}{\partial r_{j}}\frac{\partial}{\partial r_{k}}\frac{\partial}{\partial r_{l}}\chi_{\rho(V,0)}(\mathbf{r})&=\Tr\!\left[\rho(V,0)\frac{\partial}{\partial r_{j}}\frac{\partial}{\partial r_{k}}\frac{\partial}{\partial r_{l}}\hat{D}_{\mathbf{r}}\right]=\frac{i}{8}\Tr\!\left[\rho(V,0)\,\{(\Omega \hat{\mathbf{R}})_j,\{(\Omega \hat{\mathbf{R}})_k,\{(\Omega \hat{\mathbf{R}})_l, \hat{D}_{\mathbf{r}}  \}\}\}   \right]\,.
\ee
Consequently, by exploiting that $\hat{D}_0=\hat{\mathbb{1}}$, we obtain that
	\bb
\frac{\partial}{\partial r_{j}}\frac{\partial}{\partial r_{k}}\frac{\partial}{\partial r_{l}}\chi_{\rho(V,0)}(\mathbf{r})|_{\mathbf{r}=0}&=\frac{i}{4}\Tr[\rho(V,0)\,\{(\Omega \hat{\mathbf{R}})_j,\{(\Omega \hat{\mathbf{R}})_k,(\Omega \hat{\mathbf{R}})_l\} \}  ]\,.
\ee
In particular, we get
\bb \label{eq_part2m}\sum_{j,k,l=1}^n\Omega_{jj'}\Omega_{kk'}\Omega_{ll'}\frac{\partial}{\partial r_{j}}\frac{\partial}{\partial r_{k}}\frac{\partial}{\partial r_{l}}\chi_{\rho(V,0)}(\mathbf{r})|_{\mathbf{r}=0}&=\frac{i}{4}\Tr[\rho(V,0)\,\{\hat{R}_j,\{\hat{R}_k,\hat{R}_l\} \}  ]\,.
\ee 
By putting~\eqref{eq_part1m} and~\eqref{eq_part2m} together, we obtain that 
\bb\label{ahahm}
    \Tr[\rho(V,0)\,\{\hat{R}_j,\{\hat{R}_k,\hat{R}_l\} \}  ]&=0\,.
\ee
Finally, note that
\bb
    \Tr[\rho(V,0)\,\hat{R}_j\hat{R}_k\hat{R}_l   ]&=\frac{\Tr[\rho(V,0)\,\hat{R}_j\{\hat{R}_k,\hat{R}_l\}   ]}{2}\\
    &=\frac{\Tr[\rho(V,0)\,\{\hat{R}_j\{\hat{R}_k,\hat{R}_l\}\}   ]}{4}\\
    &=0\,,
\ee
where the two equalities follow from the canonical commutation relation and from the fact that $\rho(V,0)$ has zero first moment, while the last equality follows from~\eqref{ahahm}.
\end{proof}
It is worth noting that the approach used in the lemma above can be easily generalised to show that \emph{any} odd moment of a Gaussian state with zero first moment vanishes. However, in order to prove such a more general result, we use a much simpler approach.
\begin{lemma}
    A Gaussian state with vanishing first moments has also vanishing odd moments.
\end{lemma}
\begin{proof}
    We need to show that the overlap between a Gaussian state with vanishing first moment and a product of an odd number of quadratures vanishes.

    To this end, note that any Gaussian unitary transformation preserves the degree of any polynomial in the quadratures. Hence, thanks to Williamson decomposition, we can assume without loss of generality that the Gaussian state is a tensor product of thermal states. Additionally, taking the trace of each mode, we can further assume without loss of generality the Gaussian state to be a single-mode thermal state. Thus, it suffices to demonstrate that any single-mode thermal state has vanishing odd moments.

    Since thermal states are invariant under phase-space inversion and the product of an odd number of quadratures changes sign under such inversion, it follows that the overlap between a single-mode thermal state and a product of an odd number of quadratures must vanish.  
\end{proof}

\subsection{Derivative of a Gaussian state}\label{Sec_Der}
In this subsection, we first define the \emph{derivative of a Gaussian state} and then prove relevant properties. 
Throughout the rest of the subsection, we will use the following notation. For any $2n\times 2n$ real positive definite matrix $V$ and any real $2n$-dimensional vector $\mathbf{m}$, we define the operator $\rho(V,\mathbf{m})$ as 
\bb\label{def_V_gtr0}
\rho(V,\mathbf{m})\coloneqq\int_{\R^{2n}}\frac{\mathrm{d}\mathbf{r}}{(2\pi)^n}   e^{ -\frac{1}{4}(\Omega \mathbf{r})^{\intercal}V\Omega \mathbf{r} -i\mathbf{r}^\intercal \Omega \mathbf{m} }\hat{D}_{\mathbf{r}}^\dagger\,.
\ee 
If $V$ satisfies the uncertainty principle $V+i\Omega\ge0$, the operator $\rho(V,\mathbf{m})$ is exactly equal to the Gaussian state with first moment $\mathbf{m}$ and covariance matrix $V$. This follows from the Fourier-Weyl relation in~\eqref{eq_fourier_weyl} and from the expression of the characteristic function of a Gaussian state in~\eqref{eq:charact_Gauss_Prel}. For completeness, note that if $V$ does not satisfy the uncertainty principle, the operator $\rho(V,\mathbf{m})$ fails to be positive semi-definite (this can be seen e.g.~by contradiction and by observing that the covariance matrix of $\rho(V,\mathbf{m})$ is exactly $V$), and hence it is not a quantum state.

We begin with the definition of the directional derivative of a Gaussian state. 
\begin{Def}[(Derivative of a Gaussian state)]\label{def:dir_dev}
    Let $\rho(V,\mathbf{m})$ be an $n$-mode Gaussian state with first moment $\mathbf{m}$ and covariance matrix $V$. Let $X$ be a $2n\times 2n$ real symmetric matrix $X$ and let $\mathbf{x}$ be a real $2n$-dimensional vector. We define the directional derivative of the Gaussian state $\rho(V,\mathbf{m})$ along the direction $(X,\mathbf{x})$ as
    \bb\label{eq_der_def_0}
        \partial_{X,\mathbf{x}}\,\rho(V,\mathbf{m}) \coloneqq \lim\limits_{\alpha\to 0^+}\frac{\rho(V + \alpha X,\mathbf{m}+\alpha\mathbf{x}) - \rho(V,\mathbf{m})}{\alpha} \, ,
    \ee 
    where the limit is taken with respect to the trace norm. 
\end{Def}
Note that since $V$ is positive definite, the matrix $V+\alpha X$ is positive definite for sufficiently small values of $\alpha>0$, and thus the operator $\rho(V + \alpha X,\mathbf{m}+\alpha\mathbf{x})$ in~\eqref{eq_der_def_0} is well defined. The existence of this limit is guaranteed by the fact that, for any $V$ and $\mathbf{m}$, the state $\rho(V,\mathbf{m})$ is continuously differentiable, as proven in Theorem~\ref{thm:gauss_regularity} of Section~\ref{sec:proof_diff}.

\subsubsection{Compact formula for the derivative of a Gaussian state}
Here, we prove a compact formula for the derivative of a Gaussian state.
\begin{thm}[(Compact formula for the derivative of a Gaussian state)]\label{lemma_der_state}
    The derivative of the Gaussian state $\rho(V,\mathbf{m})$ along the direction $(X,\mathbf{x})$ can be expressed as
    \begin{align}
    \partial_{X,\mathbf{x}}\,\rho(V,\mathbf{m}) 
    = -\frac{1}{4} \sum_{i,j=1}^{2n} \tilde{X}_{ij} [\hat{R}_i, [\hat{R}_j, \rho(V, \mathbf{m})]] + i \sum_{j=1}^{2n} \tilde{x}_j [\hat{R}_j, \rho(V, \mathbf{m})],
    \end{align}
    where
    \begin{align}
    \tilde{X} &\coloneqq \Omega^\intercal X \Omega, \\
    \tilde{\mathbf{x}} &\coloneqq \Omega \mathbf{x}.
    \end{align}
\end{thm}
Before proving Theorem~\ref{lemma_der_state}, let us prove some preliminary results.

\begin{lemma}\label{identity_commutator}
    Let $j \in \{1,2,\ldots,2n\}$ and $\mathbf{r}\in\R^{2n}$. The commutator between the quadrature operator $\hat{R}_i$ and the displacement operator $\hat{D}_\mathbf{r} \coloneqq e^{-i \mathbf{r}^\intercal \Omega \hat{R}}$ can be expressed as
    \begin{align}
        [\hat{R}_j,\hat{D}_\mathbf{r}] = r_j \hat{D}_\mathbf{r}\,,
    \end{align}
    where $r_j$ is the $j$th component of the vector $\mathbf{r}$.
\end{lemma}
\begin{proof}
    It holds that
    \bb
    [\hat{R}_j, \hat{D}_\mathbf{r}] &= [\hat{R}_j, e^{-i \mathbf{r}^\intercal \Omega \hat{R}}]\\
    &= \sum_{k=1}^\infty \frac{1}{k!} \left[\hat{R}_j, \left(-i \mathbf{r}^\intercal \Omega \hat{R}\right)^k \right]\\
    &\eqt{(i)} \left[\hat{R}_j, -i \mathbf{r}^\intercal \Omega \hat{R} \right] \sum_{k=1}^\infty \frac{1}{(k-1)!} \left(-i \mathbf{r}^\intercal \Omega \hat{R}\right)^{k-1}\\
    &=\left[\hat{R}_j, -i \mathbf{r}^\intercal \Omega \hat{R} \right] \hat{D}_\mathbf{r},\\
    &\eqt{(ii)} r_j \hat{D}_\mathbf{r}\,.
    \ee
    Here, in (i) we exploited that 
    \bb
    \left[A, B^k\right] = k[A, B] B^{k-1}
    \ee
    for any $A,B$ satisfying $[A,[A, B]]=[B,[A, B]]=0$. Moreover, in (ii), by leveraging the canonical commutation relation $[\hat{R}_j, \hat{R}_k] = i \Omega_{jk}$ and the fact that $\Omega \Omega^\intercal = \mathbb{1}$, we observed that
    \bb
    \left[\hat{R}_j, -i \mathbf{r}^\intercal \Omega \hat{R} \right] &= -i \sum_{l,k=1}^{2n} r_l \Omega_{lk} \left[\hat{R}_j, \hat{R}_k\right]\\
    &=\sum_{l,k=1}^{2n} r_l \Omega_{lk} \Omega_{jk}\hat{\mathbb{1}}\\
    &=\sum_{l=1}^{2n} r_l \delta_{lj} \hat{\mathbb{1}}\\
    &= r_j\hat{\mathbb{1}}\,.
    \ee
\end{proof}

\begin{lemma}\label{lemma_comm_char}
    Let $i \in \{1,2,\ldots,2n\}$. Let $\Theta$ be a bounded linear operator on $L^2(\R^{2n})$ such that $[\hat{R}_i,\Theta]$ is also bounded. Then, it holds that
    \bb
    \int_{\R^{2n}} \frac{\mathrm{d} \mathbf{r}}{(2\pi)^n} r_i \chi_\Theta(\mathbf{r}) \hat{D}_{\mathbf{r}}^\dagger = [\Theta,\hat{R}_i]\,.
    \ee
\end{lemma}

\begin{proof}
	Note that the characteristic function of $[\Theta,\hat{R}_i]$ satisfies
	\bb
	\chi_{[\Theta,\hat{R}_i]}(\mathbf{r})=\Tr\!\left[ [\Theta,\hat{R}_i]\hat{D}_{\mathbf{r}}  \right]=\Tr\!\left[ \Theta [\hat{R}_i,\hat{D}_{\mathbf{r}}]  \right]\eqt{(i)} r_i\Tr\!\left[ \Theta \hat{D}_{\mathbf{r}}  \right]=r_i\chi_{\Theta}(\mathbf{r})\,,
	\ee
	where in (i) we exploited Lemma~\ref{identity_commutator}. Hence, the Fourier-Weyl relation in~\eqref{eq_fourier_weyl} concludes the proof.
\end{proof}
Now, we are ready to prove Theorem~\ref{lemma_der_state}.
\begin{proof}[Proof of Theorem~\ref{lemma_der_state}]
By exploiting~\eqref{def_V_gtr0} and the expression of the characteristic function of a Gaussian state in~\eqref{eq:charact_Gauss_Prel}, we have that
\bb
	\rho(V+\alpha X,\mathbf{m}+\alpha\textbf{x})&=\int_{\R^{2n}}\frac{\mathrm{d}\mathbf{r}}{(2\pi)^n} \chi_{\rho(V,\mathbf{m})}(\mathbf{r})  e^{-\alpha\frac{1}{4}(\Omega \mathbf{r})^{\intercal}X\Omega \mathbf{r}-i\alpha\mathbf{r}^\intercal \Omega \textbf{x}}\hat{D}_{\mathbf{r}}^\dagger \,.
\ee
Consequently, it holds that
\bb
    \partial_{X,\mathbf{x}}\,\rho(V,\mathbf{m})&=\lim\limits_{\alpha\to 0^+}\int_{\R^{2n}}\frac{\mathrm{d}\mathbf{r}}{(2\pi)^n} \chi_{\rho(V,\mathbf{m})}(\mathbf{r})  \left( \frac{e^{-\alpha\frac{1}{4}(\Omega \mathbf{r})^{\intercal}X\Omega \mathbf{r}-i\alpha\mathbf{r}^\intercal \Omega \textbf{x}}-1}{\alpha} \right)\hat{D}_{\mathbf{r}}^\dagger \\
    &\eqt{(i)}\int_{\R^{2n}}\frac{\mathrm{d}\mathbf{r}}{(2\pi)^n} \chi_{\rho(V,\mathbf{m})}(\mathbf{r})  \left( \lim\limits_{\alpha\to 0^+}\frac{e^{-\alpha\frac{1}{4}(\Omega \mathbf{r})^{\intercal}X\Omega \mathbf{r}-i\alpha\mathbf{r}^\intercal \Omega \textbf{x}}-1}{\alpha} \right)\hat{D}_{\mathbf{r}}^\dagger \\
    &=\int_{\R^{2n}}\frac{\mathrm{d}\mathbf{r}}{(2\pi)^n} \chi_{\rho(V,\mathbf{m})}(\mathbf{r})  \left( -\frac{1}{4}(\Omega \mathbf{r})^{\intercal}X\Omega \mathbf{r}-i\mathbf{r}^\intercal \Omega \textbf{x}\right)\hat{D}_{\mathbf{r}}^\dagger \\
    &= -\int_{\R^{2n}}\frac{\mathrm{d}\mathbf{r}}{(2\pi)^n}  \chi_{\rho(V,\mathbf{m})}(\mathbf{r}) 
	\frac{1}{4}\mathbf{r}^{\intercal} \tilde{X} \mathbf{r} \,  \hat{D}_{\mathbf{r}}^\dagger-i\int_{\R^{2n}}\frac{\mathrm{d}\mathbf{r}}{(2\pi)^n} \chi_{\rho(V,\mathbf{m})}(\mathbf{r})  
	\mathbf{r}^\intercal \tilde{\textbf{x}}   \,\hat{D}_{\mathbf{r}}^\dagger  \\
    &=-\frac{1}{4}\sum_{i,j=1}^{2n}\tilde{X}_{ij} \int_{\R^{2n}}\frac{\mathrm{d}\mathbf{r}}{(2\pi)^n}   \chi_{\rho(V,\mathbf{m})}(\mathbf{r})
	r_ir_j    \hat{D}_{\mathbf{r}}^\dagger-i\sum_{j=1}^{2n}\tilde{x}_j\int_{\R^{2n}}\frac{\mathrm{d}\mathbf{r}}{(2\pi)^n} \chi_{\rho(V,\mathbf{m})}(\mathbf{r})  
	r_j  \hat{D}_{\mathbf{r}}^\dagger \,.
\ee
Here, the equality (i) is justified by the differentiability of Gaussian states proved in Theorem~\ref{thm:gauss_regularity}. Additionally, Lemma~\ref{lemma_comm_char} guarantees that
\bb
    \int_{\R^{2n}}\frac{\mathrm{d}\mathbf{r}}{(2\pi)^n} \chi_{\rho(V,\mathbf{m})}(\mathbf{r})  
	r_j  \hat{D}_{\mathbf{r}}^\dagger &=[\rho(V,\mathbf{m}),\hat{R}_j]\,,\\
    \int_{\R^{2n}}\frac{\mathrm{d}\mathbf{r}}{(2\pi)^n}   \chi_{\rho(V,\mathbf{m})}(\mathbf{r})
	r_ir_j    \hat{D}_{\mathbf{r}}^\dagger&=[[\rho(V,\mathbf{m}),\hat{R}_i],\hat{R}_j]\,.
\ee
The assumptions of Lemma~\ref{lemma_comm_char} are satisfied because $[\rho(V,\mathbf{m}),\hat{R}_i]$ and $[[\rho(V,\mathbf{m}),\hat{R}_i],\hat{R}_j]$ are bounded for any $i,j\in\{1,2,\ldots,2n\}$, as a consequence of Proposition 3.18 in the arXiv version of~\cite{schwartz_op}. In particular, we conclude that
\bb
\partial_{X,\mathbf{x}}\,\rho(V,\mathbf{m})&=-\frac{1}{4}\sum_{i,j=1}^{2n}\tilde{X}_{ij} [[\rho(V,\mathbf{m}),\hat{R}_j],\hat{R}_i]+i\sum_{j=1}^{2n}\tilde{x}_j[\hat{R}_j,\rho(V,\mathbf{m})]  \\
	&=-\frac{1}{4}\sum_{i,j=1}^{2n}\tilde{X}_{ij} [\hat{R}_i,[\hat{R}_j,\rho(V,\mathbf{m})]]+i\sum_{j=1}^{2n}\tilde{x}_j[\hat{R}_j,\rho(V,\mathbf{m})]\,.
\ee
\end{proof}
\subsubsection{Upper bound on the trace norm of the derivative of a Gaussian state}
Here, we derive a simple upper bound on the trace norm of the derivative of a Gaussian state.
\begin{lemma}[(Upper bound on the trace norm of the derivative of a Gaussian state)]\label{lemma_upp_mixed}
    Let $\rho(V,0)$ be a Gaussian state with covariance matrix $V$ and zero first moment. Then, the trace norm of the derivative of the Gaussian state $\rho(V,0)$ along the direction $(X,0)$ can be upper bounded as
	\bb
	\|  \partial_{X,0}\rho(V,0) \|_1&\leq \frac{1+\sqrt{3}}{4}\Tr[|X|\Omega^\intercal V\Omega]\,.
	\ee
\end{lemma}
\begin{proof}
We have that
	\bb\label{two_terms_to_bound}
	\|  \partial_{X,0}\rho(V,0) \|_1&\eqt{(i)}\frac{1}{4}\left\|  \sum_{ij}\tilde{X}_{ij} [\hat{R}_i,[\hat{R}_j,\rho]]  \right\|_1\\
	&\eqt{(ii)}\frac14\left\|   \hat{\mathbf{R}}^\intercal \tilde{X} \hat{\mathbf{R}}\rho  +\rho\hat{\mathbf{R}}^\intercal \tilde{X} \hat{\mathbf{R}}-2\sum_{i,j}\tilde{X}_{ij}\hat{R}_i\rho\hat{R}_j \right\|_1\\
 &\leqt{(iii)} \frac{1}{2}\left\|   \hat{\mathbf{R}}^\intercal \tilde{X} \hat{\mathbf{R}}\rho\right\|_1+\frac12\left\|\sum_{i,j}\tilde{X}_{ij}\hat{R}_i\rho\hat{R}_j \right\|_1\,,
\ee
where in (i) we employed Lemma~\ref{lemma_der_state}, (ii) is a consequence of the fact $X$ is symmetric, and in (iii) we used triangle inequality and the fact that 
\bb
\|\hat{\mathbf{R}}^\intercal \tilde{X} \hat{\mathbf{R}}\rho\|_1=\|(\hat{\mathbf{R}}^\intercal \tilde{X} \hat{\mathbf{R}}\rho)^\dagger\|_1=\|\rho\hat{\mathbf{R}}^\intercal \tilde{X} \hat{\mathbf{R}}\|_1\,.
\ee
Let us upper bound each of the two terms in~\eqref{two_terms_to_bound}. The first term can be upper bounded as follows:
\bb\label{bound0000000001}
    \left\|   \hat{\mathbf{R}}^\intercal \tilde{X} \hat{\mathbf{R}}\rho\right\|_1&=\left\|   \hat{\mathbf{R}}^\intercal \tilde{X} \hat{\mathbf{R}}\sqrt{\rho}\sqrt{\rho}\right\|_1\\
    &\leqt{(iv)} \left\|   \hat{\mathbf{R}}^\intercal \tilde{X} \hat{\mathbf{R}}\sqrt{\rho}\right\|_2\|\sqrt{\rho}\|_2\\
    &= \sqrt{\Tr\left[(\hat{\mathbf{R}}^\intercal \tilde{X} \hat{\mathbf{R}})^2\rho\right]}\\
    &\eqt{(v)} \frac{\sqrt{(\Tr[\tilde{X}V])^2+2\Tr[\tilde{X}V\tilde{X}V]+2\Tr[\Omega \tilde{X}\Omega \tilde{X}]}}{2}\\
    &\leqt{(vi)} \frac{\sqrt{\left(\Tr[|X|\Omega^\intercal V\Omega]\right)^2+2\Tr[\tilde{X}V\tilde{X}V]+2\Tr[\Omega \tilde{X}\Omega \tilde{X}]}}{2}\\
    &\leqt{(vii)}\frac{\sqrt{3}}{2}\Tr[|X|\Omega^\intercal V \Omega]\,.
\ee
Here, in (iv), we exploited H\"older's inequality. In (v), we exploited Lemma~\ref{lemma_xx_p_gauss}. In (vi), we used that
	\bb\label{eq_1_re}
	(\Tr[\tilde{X}V])^2\le (\Tr[|\tilde{X}|V])^2 = \left(\Tr[ |X|\Omega^\intercal V\Omega]\right)^2.
	\ee
In (vii), we first considered the real spectral decomposition of the real symmetric matrix $\tilde{X}$, i.e.
\bb
    \tilde{X}=\sum_{j=1}^{2n}\lambda_j \textbf{e}_j\textbf{e}_j^\intercal
\ee
where $\lambda_1,\lambda_2,\ldots,\lambda_{2n}\in\R$ are the eigenvalues and $\textbf{e}_1,\textbf{e}_2,\ldots,\textbf{e}_{2n}\in\R^{2n}$ are the real orthonormal eigenvectors, and second we observed that
\bb\label{eq_2_re}
	\Tr[\tilde{X}V\tilde{X}V]+\Tr[\Omega \tilde{X}\Omega \tilde{X}]&=\Tr[\tilde{X}(V+i\Omega)\tilde{X}(V-i\Omega)]\\
 &\eqt{(viii)}\,\Tr\!\left[\sqrt{V-i\Omega}\tilde{X}\sqrt{V+i\Omega}\sqrt{V+i\Omega}\tilde{X}\sqrt{V-i\Omega}\right]\\
 &= \left\|\sqrt{V-i\Omega}\tilde{X}\sqrt{V+i\Omega}\right\|_2^2\\
 &\le \left(\sum_{j=1}^{2n}|\lambda_j| \left\|\sqrt{V-i\Omega}\textbf{e}_j\textbf{e}_j^\intercal \sqrt{V+i\Omega}\right\|_2\right)^2\\
 &\eqt{(ix)} \left(\sum_{j=1}^{2n}|\lambda_j| \sqrt{\textbf{e}_j ^\intercal (V-i\Omega) \textbf{e}_j \textbf{e}_j ^\intercal (V+i\Omega) \textbf{e}_j}\right)^2\\
 &\eqt{(x)}\left(\sum_{j=1}^{2n}|\lambda_j| \textbf{e}_j ^\intercal V \textbf{e}_j\right)^2\\
 &= \left(\Tr[|\tilde{X}|V]\right)^2\\
&= \left(\Tr[ |X|\Omega^\intercal V\Omega]\right)^2\,.
\ee
Here, in (viii) we exploited the uncertainty relation $V+i\Omega\ge0$, which implies that $V-i\Omega=(V+i\Omega)^\intercal\ge0$. In (ix), we exploited that $\|\ketbraa{v}{w}\|_2=\sqrt{\braket{v|v}\braket{w|w}}$. Finally, in (x), we exploited that $\textbf{e}_j$ is a real vector and $\Omega$ is skew-symmetric, implying that $\textbf{e}^{\intercal}_j\Omega\textbf{e}_j=0$. This completes the proof of $\left\|   \hat{\mathbf{R}}^\intercal \tilde{X} \hat{\mathbf{R}}\rho\right\|_1\le\frac{\sqrt{3}}{2}\Tr[|X|\Omega^\intercal V \Omega]$.

Now, let us proceed with the bound of the second term in~\eqref{two_terms_to_bound}. Similarly as above, let us write the real spectral decomposition of $\tilde{X}$ as $\tilde{X}=O^\intercal DO$, with $D=\diag(\lambda_1,\ldots,\lambda_{2n})$ and $O$ orthogonal. Then, by defining $\hat{Y}_l\coloneqq \sum_{i=1}^{2n} O_{li}\hat{R}_i$, we have that
    \bb\label{bound0000000002}
        \left\|  \sum_{i,j=1}^{2n}\tilde{X}_{ij} \hat{R}_i\rho \hat{R}_j \right\|_1&=\left\|  \sum_{l=1}^{2n}\lambda_l\hat{Y}_l\rho \hat{Y}_l \right\|_1\\
        &\le \sum_{l=1}^{2n}|\lambda_l|\left\| \hat{Y}_l\rho \hat{Y}_l \right\|_1\\
        &\leqt{(xi)}\sum_{l=1}^{2n}|\lambda_l|\| \hat{Y}_l\sqrt{\rho}\|_2\|\sqrt{\rho} \hat{Y}_l \|_2\\
        &\le\sum_{l=1}^{2n}|\lambda_l|\Tr[\rho \hat{Y}_l^2]\\
        &=\sum_{l=1}^{2n}|\lambda_l|\sum_{i,j=1}^{2n}O_{li}O_{lj}\Tr[\rho \hat{R}_i\hat{R}_j]\\
        &=\sum_{i,j=1}^{2n}|\tilde{X}|_{ij}\Tr[\rho \hat{R}_i\hat{R}_j]\\
        &=\frac{1}{2}\sum_{i,j=1}^{2n}|\tilde{X}|_{ij}\Tr[\rho \{\hat{R}_i,\hat{R}_j\}]\\
        &=\frac{1}{2}\sum_{i,j=1}^{2n}|\tilde{X}|_{ij}V_{ij}\\
        &=\frac{1}{2}\Tr[V|\tilde{X}|]\,\\
        &=\frac12\Tr[|X|\Omega^\intercal V \Omega]\,,
    \ee
    where in (xi) we employed H\"older's inequality. Consequently, by using the inequalities~\eqref{two_terms_to_bound},~\eqref{bound0000000001}, and~\eqref{bound0000000002}, we conclude that
    \bb
        \|  \partial_{X,0}\rho(V,0) \|_1\le \frac{\sqrt{3}+1}{4} \Tr[|X|\Omega^\intercal V \Omega] \,.
    \ee
\end{proof}
\subsubsection{Exact formula for the trace norm of the derivative of a pure Gaussian state}
In this subsection, we determine a closed formula for the trace norm of the derivative of a pure Gaussian state $\rho(V,0)$ in terms of its covariance matrix $V$.
\begin{thm}[(Closed formula for the trace norm of the derivative of a pure Gaussian state)]\label{exact_der}
Let $\rho(V,0)$ be a pure Gaussian state with a (pure) covariance matrix $V$ and zero first moment. Then, the trace norm of the derivative of the Gaussian state $\rho(V,0)$ along the direction $(X,0)$ is given by 
	\bb
	\|  \partial_{X,0}\rho(V,0) \|_1= \frac{\sqrt{ (\Tr[\tilde{X}V])^2+2\Tr[\tilde{X}V\tilde{X}V]+2\Tr[\Omega \tilde{X}\Omega \tilde{X}]
 }+\|\sqrt{V+i\Omega}\tilde{X}\sqrt{V+i\Omega}\|_1}{4}\,,
	\ee
	where $\tilde{X}\coloneqq \Omega X\Omega^\intercal$.
\end{thm}
\begin{proof}
Since $\rho(V,0)$ is pure, we can write $\rho(V,0)=\ketbra{\psi}$ for some $\ket{\psi}$.
Lemma~\ref{lemma_der_state} establishes that
\bb
    \partial_{X,0}\rho(V,0)&=-\frac{1}{4}\sum_{ij}\tilde{X}_{ij} [\hat{R}_i,[\hat{R}_j,\ketbra{\psi}]]  \\
    &\eqt{(i)}-\frac{1}{4}\left(\hat{\mathbf{R}}^\intercal \tilde{X} \hat{\mathbf{R}}\ketbra{\psi}  +\ketbra{\psi}\hat{\mathbf{R}}^\intercal \tilde{X} \hat{\mathbf{R}}-2\sum_{i,j}\tilde{X}_{ij}\hat{R}_i\ketbra{\psi}\hat{R}_j \right)\\
    &=-\frac{1}{4}\left(\Theta_1-\Theta_2\right)
\ee
where in (i) we exploited that $\tilde{X}$ is symmetric, and in (ii) we defined the hermitian operators 
\bb
    \Theta_1&\coloneqq \hat{\mathbf{R}}^\intercal \tilde{X} \hat{\mathbf{R}}\ketbra{\psi}  +\ketbra{\psi}\hat{\mathbf{R}}^\intercal \tilde{X} \hat{\mathbf{R}}\,,\\
    \Theta_2&\coloneqq2\sum_{i,j}\tilde{X}_{ij}\hat{R}_i\ketbra{\psi}\hat{R}_j \,.
\ee
Note that these operators are orthogonal to each other, i.e.~$\Theta_1\Theta_2=\Theta_2\Theta_1=0$. This follows from the fact that the first moment of $\rho(V,0)$ vanishes and from Lemma~\ref{lemma_3_moment}, which implies that $\Tr[\psi\hat{R}_i\hat{R}_j\hat{R}_k]=0$ for all $i,j,k\in\{1,2,\ldots,2n\}$. Consequently, we have that
	\bb\label{expr_to_calcul}
	\|  \partial_{X,0}\rho(V,0) \|_1&=\frac{\|\Theta_1\|_1}{4}+\frac{\|\Theta_2\|_1}{2}\,,
	\ee
	where 
	\bb
	\|\Theta_1\|_1 & =\left\|   \hat{\mathbf{R}}^\intercal \tilde{X} \hat{\mathbf{R}}\ketbra{\psi}  +\ketbra{\psi}\hat{\mathbf{R}}^\intercal \tilde{X} \hat{\mathbf{R}}\right\|_1\,\\
	\|\Theta_2\|_1 &= \left\|\sum_{i,j}\tilde{X}_{ij}\hat{R}_i\ketbra{\psi}\hat{R}_j \right\|_1\,.
	\ee
	Let us proceed to calculate the  term $\|\Theta_1\|_1$ in~\eqref{expr_to_calcul}. By defining the non-normalised vector $\ket{v}\coloneqq \hat{\mathbf{R}}^\intercal \tilde{X} \hat{\mathbf{R}}\ket{\psi}$, we have that
 \bb
    \|\Theta_1\|_1=\|\ketbraa{v}{\psi}+\ketbraa{\psi}{v}\|_1\,.
\ee 
The dimension of the support of the operator $\ketbraa{v}{\psi}+\ketbraa{\psi}{v}$ is two. Hence, its trace norm can be easily calculated and it reads: 
	\begin{align}
	\|\Theta_1\|_1=\|\ketbraa{v}{\psi}+\ketbraa{\psi}{v}\|_1 
    &=2\sqrt{\langle v | v\rangle-\mathrm{Im}(\langle \psi | v\rangle)^2}\,.
    \end{align}
    Note that $\langle v | v\rangle= \Tr[  (\hat{\mathbf{R}}^\intercal \tilde{X} \hat{\mathbf{R}})^2\psi   ]$ and that $\mathrm{Im}(\langle \psi | v\rangle)=\mathrm{Im}(\Tr[\hat{\mathbf{R}}^\intercal \tilde{X} \hat{\mathbf{R}}\psi])=0$. To show the latter, note that $\hat{\mathbf{R}}^\intercal \tilde{X} \hat{\mathbf{R}}$ is an hermitian operator and thus the expectation value $\Tr[\hat{\mathbf{R}}^\intercal \tilde{X} \hat{\mathbf{R}}\psi]$ is real.
    Consequently, we have that
	\bb
	\|\Theta_1\|_1&=\sqrt{   4\Tr[  (\hat{\mathbf{R}}^\intercal \tilde{X} \hat{\mathbf{R}})^2\psi   ] }\\
        &=\sqrt{ (\Tr[\tilde{X}V])^2+2\Tr[\tilde{X}V\tilde{X}V]+2\Tr[\Omega \tilde{X}\Omega \tilde{X}]
    }\,,
	\ee
 where in the last line we employed Lemma~\ref{lemma_xx_p_gauss}.

 Let us now calculate the term $\|\Theta_2\|_1$ in~\eqref{expr_to_calcul}. To this end, we exploit the simple-to-prove fact that, given (non-orthonormal) vectors $(\ket{\psi_i})_i$ and a matrix $A$, the non-zero singular values of the operator $\sum_{ij} A_{ij}\ketbraa{\psi_i}{\psi_j}$ are the same of the matrix $\sqrt{G}A\sqrt{G}$, where $G$ is the \emph{Gram matrix} defined as $G_{ij}\coloneqq \langle \psi_i|\psi_j\rangle$. Consequently, we have that
	\bb
	\|\Theta_2\|_1 &=\left\|\sum_{i,j}\tilde{X}_{ij}\hat{R}_i\ketbra{\psi}\hat{R}_j \right\|_1=  \left\|\sqrt{G}\tilde{X}\sqrt{G} \right\|_1\,,
	\ee
	where 
	\bb
	G_{ij}\coloneqq \bra{\psi} \hat{R}_i\hat{R}_j\ket{\psi}=\frac{V_{ij}+i\Omega_{ij}}{2},.
	\ee
	Hence, we have that $\|\Theta_2\|_1=\frac12\|\sqrt{V+i\Omega}\tilde{X}\sqrt{V+i\Omega}\|_1$. Consequently, by substituting into~\eqref{expr_to_calcul}, we conclude that
	\bb
	\|  \partial_{X,0}\rho(V,0) \|_1= \frac{\sqrt{ (\Tr[\tilde{X}V])^2+2\Tr[\tilde{X}V\tilde{X}V]+2\Tr[\Omega \tilde{X}\Omega \tilde{X}] }+\|\sqrt{V+i\Omega}\tilde{X}\sqrt{V+i\Omega}\|_1}{4}\,.
	\ee
\end{proof}

\subsection{Upper bound on the trace distance between Gaussian states}\label{Sec_upppp_bounds}
 In this subsection, we derive our main result: an upper bound on the trace distance between Gaussian states in terms of the norm distance of their first moments and covariance matrices.
\begin{thm}[(Upper bound on the trace distance between Gaussian states)]\label{thm_sm}
Let $\rho(V,\mathbf{m})$ be the Gaussian state with first moment $\mathbf{m}$ and covariance matrix $V$. Similarly, let $\rho(W,\mathbf{t})$ be the Gaussian state with first moment $\mathbf{t}$ and covariance matrix $W$. Then, the trace distance can be upper bounded as follows:
    \bb
        \frac12\left\|\rho(V,\mathbf{m})-\rho(W,\mathbf{t})\right\|_1\le \frac{1+\sqrt{3}}{8}\Tr\!\left[|V-W|\Omega^\intercal\left(\frac{V+W}{2}\right)\Omega\right]+\sqrt{\frac{\min(\|V\|_\infty,\|W\|_\infty)}{2} }\|\mathbf{m}-\mathbf{t}\|_2\,.
\ee
In particular, H\"older's inequality implies that
    \bb\label{eq_smsmsm}
	\frac{1}{2}\left\|\rho(V,\mathbf{m})-\rho(W,\mathbf{t})\right\|_1&\le \frac{1+\sqrt{3}}{8}\max(\|V\|_\infty,\|W\|_\infty) \,\|V-W\|_1+ \sqrt{\frac{\min(\|V\|_\infty,\|W\|_\infty)}{2} }\|\mathbf{m}-\mathbf{t}\|_2\,,\\
 \frac{1}{2}\left\|\rho(V,\mathbf{m})-\rho(W,\mathbf{t})\right\|_1&\le \frac{1+\sqrt{3}}{8}\max(\Tr V ,\Tr W) \,\|V-W\|_\infty+ \sqrt{\frac{\min(\|V\|_\infty,\|W\|_\infty)}{2} }\|\mathbf{m}-\mathbf{t}\|_2\,.
    \ee
\end{thm}
\begin{proof}
    We have that
    \bb\label{hey}
        \left\|\rho(V,\mathbf{m})-\rho(W,\mathbf{t})\right\|_1&\le\|\rho(V,\mathbf{m})-\rho(W,\mathbf{m})\|_1+\|\rho(W,\mathbf{m})-\rho(W,\mathbf{t})\|_1\\
        &=\|\rho(V,0)-\rho(W,0)\|_1+\|\rho(W,\mathbf{m})-\rho(W,\mathbf{t})\|_1\,
    \ee
    where in the last line we exploited the invariance of the trace norm under (displacement) unitaries. The first term in~\eqref{hey} can be upper bounded as follows:
    \bb\label{step_same_first}
        \|\rho(W,0)-\rho(V,0)\|_1&\eqt{(i)}\left\|\lim\limits_{m\rightarrow\infty}\frac1m\sum_{i=1}^m \partial_{W-V,0}\rho\!\left(V+\frac{i-1}{m}(W-V),0\right)   \right\|_1\\
        &\leqt{(ii)}\lim\limits_{m\rightarrow\infty}\frac1m\sum_{i=1}^m\left\| \partial_{W-V,0}\rho\!\left(V+\frac{i-1}{m}(W-V),0\right) \right\|_1\\
        &\leqt{(iii)} \frac{\sqrt{3}+1}{4} \lim\limits_{m\rightarrow\infty}\frac1m\sum_{i=1}^m \Tr\!\left[|V-W|\Omega^\intercal \left(V+\frac{i-1}{m}(W-V)\right) \Omega\right]\\
        &=\frac{\sqrt{3}+1}{4}  \Tr\!\left[|V-W|\Omega^\intercal \left(\frac{V+W}{2}\right) \Omega\right]\,.
    \ee
    Here, equality (i) is a consequence of Corollary~\ref{cor:int_rem} (together with a `discretization' of the integral), in (ii), we exploited triangle inequality; 
    and in (iii), we employed Lemma~\ref{lemma_upp_mixed}.

    Now, let us upper bound the second term in~\eqref{hey}. Let us consider the Williamson decomposition $W=SDS^\intercal$, as in~\eqref{def_will}. Then, Lemma~\ref{lemma_add_gauss} guarantees that 
    \bb
        \rho(W,\mathbf{m})&=\mathcal{N}_{S(D-{\mathbb{1}})S^\intercal}\left(\rho(SS^\intercal,\mathbf{m})\right)\\
        \rho(W,\mathbf{t})&=\mathcal{N}_{S(D-{\mathbb{1}})S^\intercal}\left(\rho(SS^\intercal,\mathbf{t})\right)\,,
    \ee
    where $\mathcal{N}_{S(D-\mathbb{1})S^\intercal}$ is the Gaussian noise channel defined in Subsection~\ref{subsubsec_gauss_noise}. Consequently, we have that
\bb\label{step_same_cov}
    \frac12\|\rho(W,\mathbf{m})-\rho(W,\mathbf{t})\|_1&=\frac12\|\mathcal{N}_{S(D-1)S^\intercal}\left(\rho(SS^\intercal,\mathbf{m})-\rho(SS^\intercal,\mathbf{t})\right)\|_1\\
    &\leqt{(iv)}\frac12\|\rho(SS^\intercal,\mathbf{m})-\rho(SS^\intercal,\mathbf{t})\|_1\\
&\eqt{(v)} \sqrt{1 - \Tr [\rho(SS^\intercal,\mathbf{m})\,\rho(SS^\intercal,\mathbf{t})]} \\
&\eqt{(vi)} \sqrt{1 - e^{-\frac12 (\mathbf{m}-\mathbf{t})^\intercal (SS^\intercal)^{-1} (\mathbf{m}-\mathbf{t})}} \\
&\leqt{(vii)} \sqrt{\frac{(\mathbf{m}-\mathbf{t})^\intercal (SS^\intercal)^{-1} (\mathbf{m}-\mathbf{t})}{2}}\\
&\le \sqrt{\frac{ \|(SS^\intercal)^{-1}\|_\infty}{2}}\|\mathbf{m}-\mathbf{t}\|_2\\
&\eqt{(viii)}\sqrt{\frac{ \|SS^\intercal\|_\infty}{2}}\|\mathbf{m}-\mathbf{t}\|_2\\
&\leqt{(ix)}\sqrt{\frac{ \|W\|_\infty}{2}}\|\mathbf{m}-\mathbf{t}\|_2\,.
\ee
Here, in (iv) we applied the monotonicity of the trace norm under quantum channels; in (v) we exploited that $\rho_{SS^\intercal,\mathbf{m}}$ and $\rho_{SS^\intercal,\mathbf{t}}$ are pure (Gaussian) states; in (vi), we employed the known formula for the overlap between Gaussian states~\cite[Eq~(4.51)]{BUCCO}; in (vii), we exploited that $e^x\ge 1+x$ for all $x\in\R$; in (viii), we exploited that for any symplectic matrix $T$ it holds that 
\bb
    \|T^{-1}\|_\infty= \|\Omega T^\intercal\Omega^\intercal\|_\infty=\|T^\intercal\|_\infty=\|T\|_\infty\,,
\ee
where we exploited that $T\Omega T^\intercal =\Omega$ and the fact that $\Omega$ is orthogonal; finally, in (ix), we used the fact that $SS^\intercal\le SDS^\intercal=W$.

By substituting~\eqref{step_same_first} and~\eqref{step_same_cov} into~\eqref{hey}, we obtain the desired bound.
\end{proof}

\subsubsection{Improved upper bound on the trace distance between Gaussian states}
In this section, we derive an improved upper bound on the trace distance between Gaussian states, which involves the Williamson decomposition of the covariance matrices.
\begin{thm}[(Improved upper bound)]\label{thm_improved_sm}
    Let $\rho(V,\mathbf{m})$ be the Gaussian state with first moment $\mathbf{m}$ and covariance matrix $V$. Similarly, let $\rho(W,\mathbf{t})$ be the Gaussian state with first moment $\mathbf{t}$ and covariance matrix $W$. Moreover, let $V = S_1 D_1 S_1^\intercal$ and $W = S_2 D_2 S_2^\intercal$ be the Williamson decompositions of $V$ and $W$, as in~\eqref{def_will}. Then, the following better bound holds:
\bb 
 \frac{1}{2}\left\|\rho(V,\mathbf{m})-\rho(W,\mathbf{t})\right\|_1\le    
 \frac{1+\sqrt{3}}{8}\Tr\!\left[|V-W|\Omega^\intercal\left(\frac{S_1  S_1^\intercal+ S_2 S_2^\intercal}{2}\right)\Omega\right]+\frac{\min(\|S_1\|_\infty,\|S_2\|_\infty)}{\sqrt{2} }\|\mathbf{m}-\mathbf{t}\|_2\,. 
\ee
In particular,  H\"older's inequality implies that
 	\bb
	       \frac{1}{2}\left\|\rho(V,\mathbf{m})-\rho(W,\mathbf{t})\right\|_1&\le    \frac{1+\sqrt{3}}{8}\max(\|S_1\|^2_\infty,\|S_2\|^2_\infty) \|V-W\|_1+\frac{\min(\|S_1\|_\infty,\|S_2\|_\infty)}{\sqrt{2} }\|\mathbf{m}-\mathbf{t}\|_2\,,\\
            \frac{1}{2}\left\|\rho(V,\mathbf{m})-\rho(W,\mathbf{t})\right\|_1&\le    \frac{1+\sqrt{3}}{8}\max(\|S_1\|^2_2,\|S_2\|^2_2) \|V-W\|_\infty+\frac{\min(\|S_1\|_\infty,\|S_2\|_\infty)}{\sqrt{2} }\|\mathbf{m}-\mathbf{t}\|_2\,.
	\ee
\end{thm}
\begin{proof}
    Similarly to the proof of Theorem~\ref{thm_sm}, we have that
    \bb\label{hey2}
        \frac12\left\|\rho(V,\mathbf{m})-\rho(W,\mathbf{t})\right\|_1&\leqt{(i)}\frac12\|\rho(V,0)-\rho(W,0)\|_1+\frac12\|\rho(W,\mathbf{m})-\rho(W,\mathbf{t})\|_1\,\\
        &\leqt{(ii)} \frac12\|\rho(V,0)-\rho(W,0)\|_1+\sqrt{\frac{ \|S_2S_2^\intercal\|_\infty}{2}}\|\mathbf{m}-\mathbf{t}\|_2 \\
        &=\frac12\|\rho(V,0)-\rho(W,0)\|_1+\frac{ \|S_2\|_\infty}{\sqrt{2}}\|\mathbf{m}-\mathbf{t}\|_2\,.
        \ee
    Here, in (i), we exploited~\eqref{hey}, and in (ii), we employed step~(viii) of~\eqref{step_same_cov}. Hence, we just need to prove that
    \bb
        \|\rho(V,0)-\rho(W,0)\|_1\le  \frac{1+\sqrt{3}}{4}\Tr\!\left[|V-W|\Omega^\intercal\left(\frac{S_1  S_1^\intercal+ S_2 S_2^\intercal}{2}\right)\Omega\right]\,.
    \ee
    To this end, note that the triangle inequality implies that
    \bb\label{eqz}
        \|\rho(W,0)-\rho(V,0)\|_1&\le \lim\limits_{m\rightarrow\infty}\sum_{i=1}^m\left\| \rho\!\left(V+\frac{i}{m}(W-V),0\right)-\rho\!\left(V+\frac{i-1}{m}(W-V),0\right)   \right\|_1\,.
    \ee
    Moreover, since $V-S_1S_1^\intercal\ge0$ and $W-S_2S_2^\intercal\ge0$, it holds that
\bb
    K_i\coloneqq  \left(1-\frac{i-1}{m}\right)(V-S_1S_1^\intercal)+\frac{i-1}{m}(W-S_2S_2^\intercal)\ge 0\,.
\ee
Additionally, note that
\bb
    \rho\!\left(V+\frac{i-1}{m}(W-V),0\right)&=\mathcal{N}_{K_i}\left( S_1S_1^\intercal+\frac{i-1}{m}(S_2S_2^\intercal-S_1S_1^\intercal) \right)\,,\\
    \rho\!\left(V+\frac{i}{m}(W-V),0\right)&=\mathcal{N}_{K_i}\left( S_1S_1^\intercal+\frac{i-1}{m}(S_2S_2^\intercal-S_1S_1^\intercal)+\frac{W-V}{m} \right)\,,
\ee
where we denoted as $\mathcal{N}_{K_i}$ the Gaussian noise channel defined in Subsection~\ref{subsubsec_gauss_noise} and we employed Lemma~\ref{lemma_add_gauss}. Consequently, by employing~\eqref{eqz} and the monotonicity of the trace norm under quantum channels, we have that
    \bb
        \|\rho(W,0)-\rho(V,0)\|_1&\le \lim\limits_{m\rightarrow\infty}\sum_{i=1}^m\Bigg\| \rho\!\left(S_1S_1^\intercal+\frac{i-1}{m}(S_2S_2^\intercal-S_1S_1^\intercal)+\frac{W-V}{m},0\right)\\
        &\qquad\qquad\qquad-\rho\!\left(S_1S_1^\intercal+\frac{i-1}{m}(S_2S_2^\intercal-S_1S_1^\intercal),0\right)   \Bigg\|_1\,.
    \ee
Hence, we obtain that
\bb
    \|\rho(W,0)-\rho(V,0)\|_1&\leqt{(iii)}\lim\limits_{m\rightarrow\infty}\frac{1}{m}\sum_{i=1}^m\left\| \partial_{ W-V,0 }\rho\!\left(S_1S_1^\intercal+\frac{i-1}{m}(S_2S_2^\intercal-S_1S_1^\intercal),0\right)  \right\|_1\\
    &\leqt{(iv)} \frac{\sqrt{3}+1}{4} \lim\limits_{m\rightarrow\infty}\frac{1}{m}\sum_{i=1}^m \Tr\!\left[|W-V|\Omega^\intercal \left(  S_1S_1^\intercal+\frac{i-1}{m}(S_2S_2^\intercal-S_1S_1^\intercal\right)\right]\\
    &=\frac{\sqrt{3}+1}{4}\Tr\!\left[|W-V|\Omega^\intercal \left(  \frac{S_1S_1^\intercal+S_2S_2^\intercal}{2}
 \right) \Omega\right]\,,
\ee
where in (iii) we employed Theorem~\ref{thm:gauss_regularity}, while in (iv) we exploited Lemma~\ref{lemma_upp_mixed}. This, together with~\eqref{hey2}, concludes the proof.
\end{proof}
\subsubsection{Tightness of the upper bound}
In this subsection, we prove that the bounds on the trace distance derived in Theorem~\ref{thm_sm} are \emph{tight}. The key tool here is to employ the closed formula for the trace norm of the derivative of a pure Gaussian state proved in Theorem~\ref{exact_der}.
\begin{lemma}[(Tightness of the upper bounds)]\label{lemma_tight}
Let $\varepsilon\in(0,1)$ and $a>1$. Let $\rho(V_{a},0)$ and $\rho(W_{a,\varepsilon},0)$ be Gaussian states with zero first moments and covariance matrices given by
\bb
    V_{a}&\coloneqq\begin{pmatrix}
        a &0 \\ 0& a^{-1}
    \end{pmatrix}
    \,,\\
    W_{a,\varepsilon}&\coloneqq \begin{pmatrix}
        a &0 \\ 0 &a^{-1}+\epsilon
    \end{pmatrix}\,.
\ee
Remarkably, the trace distance satisfies
\bb
    \lim\limits_{\varepsilon\rightarrow0^+}\frac{\frac{1}{2}\|\rho(W_{a,\varepsilon},0)-\rho(V_{a},0)\|_1}{\frac{1+\sqrt{3}}{8}\max(\|W_{a,\varepsilon}\|_\infty,\|V_{a}\|_\infty) \,\|W_{a,\varepsilon}-V_{a}\|_1}=1\,,
\ee
and also
\bb
    \lim\limits_{a\rightarrow+\infty}\lim\limits_{\varepsilon\rightarrow0^+}\frac{\frac{1}{2}\|\rho(W_{a,\varepsilon},0)-\rho(V_{a},0)\|_1}{\frac{1+\sqrt{3}}{8}\max(\Tr W_{a,\varepsilon}, \Tr V_{a} ) \,\|W_{a,\varepsilon}-V_{a}\|_\infty}=1\,,
\ee
establishing the tightness of the upper bounds provided in Theorem~\ref{thm_sm}. 
\end{lemma}
\begin{proof}
Note that
\bb
    W_{a,\varepsilon}=V_a+\varepsilon X\,,
\ee
where
\bb
    X\coloneqq \begin{pmatrix}
        0 &0 \\ 0 & 1
    \end{pmatrix}\,.
\ee
Consequently, it holds that
\bb
    \lim\limits_{\varepsilon\rightarrow 0^+}\frac{ \|\rho(W_{a,\varepsilon},0)-\rho(V_{a},0)\|_1}{\varepsilon}&=\lim\limits_{\varepsilon\rightarrow 0^+}\frac{ \|\rho(V_a+\varepsilon X,0)-\rho(V_{a},0)\|_1}{\varepsilon}\\
    &= \|\partial_{X,0}\rho(V_a,0)\|_1\\
    &\eqt{(i)} \frac{\sqrt{ (\Tr[\tilde{X}V_a])^2+2\Tr[\tilde{X}V_a\tilde{X}V_a]+2\Tr[\Omega \tilde{X}\Omega \tilde{X}] }+\|\sqrt{V_a+i\Omega}\tilde{X}\sqrt{V_a+i\Omega}\|_1}{4}\,,
\ee
where in (i) we used the closed formula for the derivative of a pure Gaussian state reported in Lemma~\ref{exact_der} and we denoted
\bb
    \tilde{X}\coloneqq \Omega^\intercal X\Omega=\begin{pmatrix}
        1 &0 \\ 0 & 0
    \end{pmatrix}\,.
\ee
Moreover, we have that
\bb
    \Tr[\tilde{X}V_a]&=a\,,\\
    \Tr[\tilde{X}V_a\tilde{X}V_a]&=a^2\,,\\
    \Tr[\Omega \tilde{X}\Omega \tilde{X}]&=0\,,\\
    \|\sqrt{V_a+i\Omega}\tilde{X}\sqrt{V_a+i\Omega}\|_1&=a\,.
\ee
As a consequence, we obtain that
\bb
    \lim\limits_{\varepsilon\rightarrow 0^+}\frac{ \|\rho(W_{a,\varepsilon},0)-\rho(V_{a},0)\|_1}{\varepsilon}=\frac{1+\sqrt{3}}{4}a
\ee
Since 
\bb
    \|W_{a,\varepsilon}-V_{a}\|_\infty&=\varepsilon\,,\\
    \max(\|W_{a,\varepsilon}\|_\infty,\|V_{a}\|_\infty)&=\max(a,a^{-1} + \varepsilon)\,,
\ee
we thus obtain that
\bb
    \lim\limits_{\varepsilon\rightarrow0^+}\frac{\frac{1}{2}\|\rho(W_{a,\varepsilon},0)-\rho(V_{a},0)\|_1}{\frac{1+\sqrt{3}}{8}\max(\|W_{a,\varepsilon}\|_\infty,\|V_{a}\|_\infty) \,\|W_{a,\varepsilon}-V_{a}\|_1}=1\,,
\ee
Additionally, since
\bb
    \|W_{a,\varepsilon}-V_{a}\|_1&=\varepsilon\,,\\
    \max(\Tr W_{a,\varepsilon}, \Tr V_{a} )&=a+a^{-1}+\varepsilon\,,
\ee 
we obtain that
\bb
\lim\limits_{\varepsilon\rightarrow0^+}\frac{\frac{1}{2}\|\rho(W_{a,\varepsilon},0)-\rho(V_{a},0)\|_1}{\frac{1+\sqrt{3}}{8}\max(\Tr W_{a,\varepsilon}, \Tr V_{a} ) \,\|W_{a,\varepsilon}-V_{a}\|_\infty}=\frac{1}{1+a^{-2}}\,,
\ee
and hence
\bb
\lim\limits_{a\rightarrow+\infty}\lim\limits_{\varepsilon\rightarrow0^+}\frac{\frac{1}{2}\|\rho(W_{a,\varepsilon},0)-\rho(V_{a},0)\|_1}{\frac{1+\sqrt{3}}{8}\max(\Tr W_{a,\varepsilon}, \Tr V_{a} ) \,\|W_{a,\varepsilon}-V_{a}\|_\infty}=1\,.
\ee
\end{proof}

\section{Differentiability of Gaussian states} \label{sec:proof_diff}
In this section we study the regularity properties of Gaussian states, when they are seen as functions of their statistical moments. In particular, we prove that they are continuously differentiable. The section is organised as follows:

\begin{itemize}
    \item In Subsection~\ref{sec:prel_func_an} we introduce all the preliminaries required to prove the main result of this section. 
    \item In Subsection~\ref{subsec:diff} we give the full proof of the regularity of Gaussian states.
\end{itemize}

\subsection{Preliminaries on functional analysis} \label{sec:prel_func_an}
In this subsection we review some basic facts about functional analysis. In particular, we have to define Fr\'{e}chet spaces and some properties characterizing them. For the sake of brevity, in the following section the Hilbert space $L^2(\R^{n})$ will be called $\mathcal{H}$.
As it is well-known, we can deal with vector spaces in infinite dimensions through Banach spaces. In Banach spaces we can define a norm, so that the distance between vectors can be written as a difference in that norm. Normed vector spaces have useful properties, unfortunately there are instances of vector spaces which are not \textit{normable}, i.~e.~vector spaces whose topology cannot be derived from a norm. One example of those spaces are Fr\'{e}chet spaces. Fortunately, Fr\'{e}chet spaces are manageable, because their continuity properties can be derived by a countable number of seminorms. Formally we say that: 

\begin{Def}[(Fr\'{e}chet space~\cite{Narici2010})]  \label{def:frechet_sp}
    A Fr\'{e}chet space is a topological vector space which is complete with respect to a countable family of seminorms $\{\left\| \cdot \right\|_k\}_{k\in\mathbb{N}}$. Moreover, a complete translationally-invariant metric for the space can be defined by the quantity
    \bb \label{eq:fre_dist}
        d(x,y) \coloneqq \sum_{k\in\mathbb{N}}2^{-k}\frac{\left\| x - y \right\|_k}{1 + \left\| x - y \right\|_k} \, .
    \ee 
\end{Def}
As we have seen above Fr\'{e}chet spaces are complete metric spaces, a distance is translationally-invariant if, for any $x,y,z$ vectors, it satisfies 
\bb 
    d(x+z,y+z) = d(x,y) \, .
\ee 
Now, we have to clarify what we mean when we talk about continuity in Fr\'{e}chet spaces. In particular, we need to define the concept of limit in these sets. As we see in the following Definition, analogously to Banach spaces, where the limit is defined on the space norm, in Fr\'{e}chet spaces we have to take into account the whole countable family of seminorms.

\begin{Def}[{(Continuity in Fr\'{e}chet spaces)}] \label{def:fre_cont}
    We say that a function $f : A \to B$, where $A$ and $B$ are Fr\'{e}chet spaces with families of seminorms $\{\|\cdot \|_{A,k}\}_{k\in\mathbb{N}}$ and $\{\| \cdot \|_{B,k}\}_{k\in\mathbb{N}}$, respectively, is continuous at point $\bar{a}\in A$ in the Fr\'{e}chet topology, if $\exists \, \bar{b}\in B$ such that $f(\bar{a}) = \bar{b}$ and for any sequence $\{a_n\}_{n\in\mathbb{N}}$ satisfying the following condition for any $k\in\mathbb{N}$
    \bb
        \lim_{n\to\infty}\| \bar{a} - a_n \|_{A,k} = 0 \, ,
    \ee 
    it holds that
    \bb
        \lim_{n\to\infty}\| \bar{b} - f(a_n) \|_{B,k} = 0 \, ,
    \ee 
    for any $k\in\mathbb{N}$. Equivalently, we can define continuity in Fr\'{e}chet topology by employing the distance introduced in Equation~\eqref{eq:fre_dist} of Definition~\ref{def:frechet_sp}, and applying the standard notion of continuity in metric spaces. If the above condition is met we write that
    \bb
        \lim_{a\to\bar{a}} f(a) = \bar{b} \, .
    \ee 
\end{Def}
Even though Fr\'{e}chet space may seem obscure mathematical constructions, a lot of widely used vector spaces are of this kind. One example is the space of infinitely differentiable functions over $\R^{n}$, here the countable family of seminorms is given by the uniform norms of the derivatives of the function. Another important case is represented by the \textit{regular and rapidly decreasing functions}, also known as Schwartz functions. Here, we define the Schwartz space of functions as:

\begin{Def}[(Schwartz space~\cite{Hormander1990})] \label{def:schwartz_space}
	For any $k\in\mathbb{N}$, we call Schwartz space $\pazocal{S}\left(\R^k\right)$ the set of $C^{\infty}$ complex-valued functions $f:\R^k \to \C$ such that
	\bb 
	\sup_{\mathbf{r}\in\R^k}\left|\mathbf{r}^a\partial_bf(\mathbf{r})\right| < \infty \, ,
	\ee 
	for any multi-indices $a,b\in\mathbb{N}^k$. Here, $\mathbf{r}^a \coloneqq r_1^{a_1}r_2^{a_2} \ldots r^{a_k}_k$ and $\partial_b f \coloneqq \frac{\partial^{b_1}}{\partial r_{1}^{b_1}}\frac{\partial^{b_2}}{\partial r_{2}^{b_2}} \ldots \frac{\partial^{b_k}}{\partial r_{k}^{b_k}} f$.
\end{Def}

Up to now, we have discussed sets of complex-valued functions, but in this work we deal with quantum states, so we need to define an analogue of the Schwartz space for infinite-dimensional linear operators. In~\cite{schwartz_op} authors introduce precisely the set that is needed in the following, i.~e.~the set of Schwartz operators. Besides, they prove many different properties of these operators. Here, we review the definition of this vector space and of the properties required for our work. 

\begin{Def}[(Schwartz operators~\cite{schwartz_op})] \label{def:schwartz_op}
    For any number of bosonic modes $n\in\mathbb{N}$, we define the set of Schwartz operators $\pazocal{S}(\mathcal{H})$ as the linear operators acting on a dense subspace of $\mathcal{H}$ such that all the following seminorms are bounded:
	\bb
	\left\|\hat{T}\right\|_{a,b} \coloneqq \sup\left\{\left|\langle \mathbf{\hat{R}}^a \psi| \hat{T} \mathbf{ \hat{R}}^b\varphi\rangle \right| 
	: \varphi,\psi\in\pazocal{S}\left(\R^{2n}\right), \left\|\psi\right\|_{L^2}, \left\|\varphi\right\|_{L^2} \leq 1\right\} \; ,
	\ee
	for any multi-indices $a,b\in\mathbb{N}^{2n}$.
\end{Def}
The set $\pazocal{S}(\mathcal{H})$ is a Fr\'{e}chet space. Moreover, as we are going to see, $\pazocal{S}(\mathcal{H}) \subset \TT(\mathcal{H})$, is the set of trace-class operators on $\mathcal{H}$. The seminorms defined above are slightly contrived, so we now want to find a better way to express them.

\begin{prop} \label{prop:schwartz_norm_equiv}
    For any linear operator $\hat{T}\in\pazocal{S}(\mathcal{H})$, and for any $a,b$ multi-indices, it holds that
	\bb
	       \| \hat{T} \|_{a,b}  = \| \mathbf{\hat{R}}^a\hat{T}\mathbf{\hat{R}}^b \|_{\infty} \, .
	\ee 
\end{prop}
\begin{proof}
The proof is a consequence of Lemma 3.2 in the arXiv version of~\cite{schwartz_op}.
\end{proof}

Thanks to Proposition~\ref{prop:schwartz_norm_equiv}, we have a nice way to express the seminorms of the space $\pazocal{S}(\mathcal{H})$. These quantities are expressed in terms of operator norms (or Schatten infinity-norms), since in this work we are studying quantum states, it is interesting to have estimates in terms of the trace-norm. This is precisely the topis of the following Proposition: 

\begin{prop}[(Lemma 3.6 in the arXiv version of~\cite{schwartz_op})] \label{prop:schw_estimate}
    For any complex separable Hilbert space, the set of the Schwartz operators $\pazocal{S}(\mathcal{H})$ is contained in the set of the trace-class operators $\TT(\mathcal{H})$. Moreover, for any $a,b$ multi-indices it is holds that
    \bb 
        \|\mathbf{\hat{R}}^a \hat{T} \mathbf{\hat{R}}^b\|_1 \leq \| \hat{E}_n^{-2} \|_1\| \hat{E}_n^2 \mathbf{\hat{R}}^a \hat{T} \mathbf{\hat{R}}^b\|_{\infty} < \infty \, ,
    \ee 
    where $\hat{E}_n \coloneqq \frac12\sum_{j=1}^{2n}\hat{R}_j^2$ is the energy operator. 
\end{prop}

\begin{prop}[{(Proposition 3.18 in the arXiv version of~\cite{schwartz_op})}] \label{prop:schwartz_iso}
	For any $n\in\mathbb{N}$ the so Fourier--Weyl transform is a continuous linear homeomorphism between the Schwartz space of functions $\pazocal{S}\left(\R^{2n}\right)$ and the set of Schwartz operators $\pazocal{S}\left(\mathcal{H}\right)$. The map is continuous in the natural Fr\'{e}chet-space topology as in Definition~\ref{def:fre_cont}.
\end{prop}

Proposition~\ref{prop:schwartz_iso} gives us a way to connect Schwartz operators to Schwartz functions by competely preserving the structure of Fr\'{e}chet spaces. As we are going to see in the next Section, this result allows us to deal with classical objects (Schwartz functions) in order to prove properties of quantum objects (Schwartz operators).

\subsection{Proof of differentiability} \label{subsec:diff}

In this subsection we prove that any quantum Gaussian state $\rho(V,\mathbf{m})$ is continuously differentiable as a function of its statistical moments. First of all we notice that the characteristic function of any Gaussian state, which takes the form $\chi_{\rho(V,\mathbf{m})}(\mathbf{r}) = \exp\left( -\frac14 \left(\Omega\mathbf{r}\right)^{\intercal}V\Omega\mathbf{r} + i\mathbf{m}^{\intercal}\Omega\mathbf{r} \right)$, belongs to the Schwartz space for any $V$ and $\mathbf{m}$, so any Gaussian state is a Schwartz operator. The Fourier--Weyl transform of a Gaussian function is a Schwartz operator even if $V$ is strictly positive but it is not an admissible quantum covariance matrix. In addition, it is also important to remark that any polynomial multiplied by a Gaussian function is a Schwartz function. 
First, we explain a sketch of the proof. We want to prove that Gaussian states are continuously differentiable, and we know that Gaussian states can be expressed as the Fourier--Weyl transform of Gaussian functions, so we can first prove that Gaussian functions are continuously differentiable in the Fr\'{e}chet topology of the Schwartz space. Then, we can apply Proposition~\ref{prop:schwartz_iso} in order to lift up the regularity properties of Gaussian states to Gaussian functions.

We begin by stating this simple proposition on the Taylor expansion of Gaussian functions:

\begin{prop} \label{prop:taylor_exp}
    Given a Gaussian function of the form 
    \bb 
        \exp\left( -\frac14 \left(\Omega\mathbf{r}\right)^{\intercal}(V+X)\Omega\mathbf{r} + i(\mathbf{m}+\mathbf{l})^{\intercal}\Omega\mathbf{r} \right) \, ,
    \ee    
        where $V$ is a $2n \times 2n$ strictly positive real symmetric matrix, $X$ is a $2n\times 2n$ real and symmetric matrix, and $\mathbf{m},\;\mathbf{l}\in\R^{2n}$, then for any $\mathbf{r}\in\R^{2n}$ its Taylor polynomial reads
    \bb
        \exp\left( -\frac14 \left(\Omega\mathbf{r}\right)^{\intercal}V\Omega\mathbf{r} + i\mathbf{m}^{\intercal}\Omega\mathbf{r} \right) \sum_{k,\ell\in\mathbb{N}}\frac{(-1)^k [(\Omega\mathbf{r})^{\intercal}X\Omega\mathbf{r}]^k i^{\ell}(\mathbf{l}^{\intercal}\Omega\mathbf{r})^{\ell} }{4^kk!\ell !} \, .
    \ee 

\end{prop}

The following Lemma is important because it tells us that the integral remainder of the Taylor formula of a Gaussian at any order is a Schwartz function. This is going to be the first ingredient of our proof.

\begin{lemma} \label{lemma:domination}
For any $n,k\in\N$, for any $\mathbf{m},\mathbf{l}\in\R^{2n}$, for any polynomial with complex coefficients $\PP(\mathbf{r})$, for any, multi-indices $a,b$, for any $V>0$, and any $X$ symmetric matrix such that $V+X > 0$, the function

\bb 
    f_{\PP,V,X,\mathbf{m},\mathbf{l},a,b,k}(\mathbf{r}) \coloneqq \int_{0}^{1}  (1-\alpha)^{k}\PP(\mathbf{r})  \exp\left( -\frac14 \left(\Omega\mathbf{r}\right)^{\intercal}(V+\alpha X)\Omega\mathbf{r} + i(\mathbf{m}+\alpha \mathbf{l})^{\intercal}\Omega\mathbf{r} \right)  {\rm d}\alpha 
\ee 
belongs to $\pazocal{S}(\R^{2n})$.

\end{lemma}

\begin{proof}
    In order to prove this Lemma we need to show that for any $k$, for any polynomial $\PP$, for any multi-indices $a$ and $b$, for any covariance matrix $V$, for any mean $\mathbf{m}$, and for any variations $X$ and $\mathbf{l}$ the following holds true:
    \bb \label{eq:abs_int}
        \sup_{\mathbf{r}\in\R^{2n}} \left| \mathbf{r}^a\partial_b \int_{0}^{1}  (1-\alpha)^{k}\PP(\mathbf{r})  \exp\left( -\frac14 \left(\Omega\mathbf{r}\right)^{\intercal}(V+\alpha X)\Omega\mathbf{r} + i(\mathbf{m}+\alpha \mathbf{l})^{\intercal}\Omega\mathbf{r} \right)  {\rm d}\alpha \right| < \infty \, .
    \ee 
    So, we need to find a clever upper bound. The key point to prove this lemma is to find a function of $\alpha$ which dominates the following function:
    \bb \label{eq:to_be_dom}
       (1-\alpha)^k \left|\mathbf{r}^a \partial_b \PP(\mathbf{r}) \exp\left( -\frac14 \left(\Omega\mathbf{r}\right)^{\intercal}(V+\alpha X)\Omega\mathbf{r} + i(\mathbf{m}+\alpha \mathbf{l})^{\intercal}\Omega\mathbf{r} \right) \right| .
    \ee 
    To this end, we begin by considering the least eigenvalue of the positive matrix $V + \alpha X$, calling it $\lambda_{min}(\alpha)$. It is a continuous function of $\alpha$, and by positivity of the aformentioned matrix we can conclude that 
    \bb 
        \lambda_{\min}^* \coloneqq \min_{\alpha\in[0,1]}\lambda_{\min}(\alpha) > 0 \, .
    \ee
    Moreover, we know that, for any $\alpha\in[0,1]$, $V +\alpha X \geq \lambda^*_{\min}\mathbb{1}$. The previous statement holds true for any $V$ and $X$. So we obtain this simple upper bound on the exponential part of~\eqref{eq:to_be_dom}:
    \bb 
        \exp\left( -\frac14 \left(\Omega\mathbf{r}\right)^{\intercal}(V+\alpha X)\Omega\mathbf{r} + i(\mathbf{m}+\alpha \mathbf{l})^{\intercal}\Omega\mathbf{r} \right) \leq \exp\left( -\frac14 \left(\Omega\mathbf{r}\right)^{\intercal}\lambda^*_{\min}\Omega\mathbf{r} + i(\mathbf{m}+\alpha \mathbf{l})^{\intercal}\Omega\mathbf{r} \right) \, .
    \ee 
    The above upper bound is independent of $\alpha$. We now focus on the derivatives of the exponential, in particular we want to dominate the following expression:
    \bb 
        \left | \partial_b \exp\left( -\frac14 \left(\Omega\mathbf{r}\right)^{\intercal}(V+\alpha X)\Omega\mathbf{r} + i(\mathbf{m}+\alpha \mathbf{l})^{\intercal}\Omega\mathbf{r} \right) \right | \, .
    \ee 
    To do so, we compute the polynomial generated by the derivatives, i.e.~we define $\QQ_b(\alpha,\mathbf{r})$ through the following formula:
    \bb 
       \partial_b \exp\left( -\frac14 \left(\Omega\mathbf{r}\right)^{\intercal}(V+\alpha X)\Omega\mathbf{r} + i(\mathbf{m}+\alpha \mathbf{l})^{\intercal}\Omega\mathbf{r} \right) = \QQ_b(\alpha,\mathbf{r})\exp\left( -\frac14 \left(\Omega\mathbf{r}\right)^{\intercal}(V+\alpha X)\Omega\mathbf{r} + i(\mathbf{m}+\alpha \mathbf{l})^{\intercal}\Omega\mathbf{r} \right)  \, .
    \ee 
    It is easy to notice that $\QQ_b(\alpha,\mathbf{r})$ is continuous in $\alpha$, so we define
    \bb 
        \QQ^*_b(\mathbf{r}) \coloneqq \max_{\alpha\in[0,1]}|\QQ_b(\alpha,\mathbf{r})| \, .
    \ee 
    The above function is non-negative and scales at most polynomially in $\mathbf{r}$. In addition, it is independent of $\alpha$. Thus, we have 
    \bb 
        &\left | \partial_b \exp\left( -\frac14 \left(\Omega\mathbf{r}\right)^{\intercal}(V+\alpha X)\Omega\mathbf{r} + i(\mathbf{m}+\alpha \mathbf{l})^{\intercal}\Omega\mathbf{r} \right) \right | \\ 
        &\leq \QQ^*_b(\mathbf{r}) \left |\exp\left( -\frac14 \left(\Omega\mathbf{r}\right)^{\intercal}(V+\alpha X)\Omega\mathbf{r} + i(\mathbf{m}+\alpha \mathbf{l})^{\intercal}\Omega\mathbf{r} \right)\right|\, .
    \ee 
    Therefore, for any $\mathbf{r}\in\R^{2n}$, we are able to give the follwing upper bound:
    \bb 
         &(1-\alpha)^k\left|\mathbf{r}^a \partial_b \PP(\mathbf{r}) \exp\left( -\frac14 \left(\Omega\mathbf{r}\right)^{\intercal}(V+\alpha X)\Omega\mathbf{r} + i(\mathbf{m}+\alpha \mathbf{l})^{\intercal}\Omega\mathbf{r} \right) \right| \\
         & \leq (1-\alpha)^k\left | \mathbf{r}^{a}\left(\PP'_{b}(\mathbf{r}) + \QQ^*_b(\mathbf{r}) \right) \exp\left( -\frac14 \left(\Omega\mathbf{r}\right)^{\intercal}\lambda^*_{\min}\Omega\mathbf{r} + i(\mathbf{m}+\alpha \mathbf{l})^{\intercal}\Omega\mathbf{r} \right) \right | \\ 
         & \leq (1-\alpha)^k M_{\PP,V,X,\mathbf{m},\mathbf{l},a,b} \, .
    \ee 
    Here we have $\PP'_{b}(\mathbf{r}) \coloneqq \partial_b\PP(\mathbf{r})$. Furthermore, we obtain
    \bb 
        M_{\PP,V,X,\mathbf{m},\mathbf{l},a,b} \coloneqq \sup_{\mathbf{r}\in\R^{2n}}\left | \mathbf{r}^{a}\left(\PP'_{b}(\mathbf{r}) + \QQ^*_b(\mathbf{r}) \right) \exp\left( -\frac14 \left(\Omega\mathbf{r}\right)^{\intercal}\lambda^*_{\min}\Omega\mathbf{r} + i(\mathbf{m}+\alpha \mathbf{l})^{\intercal}\Omega\mathbf{r} \right) \right | < \infty \, .
    \ee 
    The last quantity is finite because we are taking the supremum of the absolute value of a Schwartz function. We are now able to define the non-negative dominating function as
    \bb 
        \xi_{\PP,V,X,\mathbf{m},\mathbf{l},a,b,k}(\alpha) \coloneqq (1-\alpha)^kM_{\PP,V,X,\mathbf{m},\mathbf{l},a,b} \, .
    \ee 
    It clearly is integrable. Moreover, we have 
    \bb 
       \int_{0}^{1}\xi_{\PP,V,X,\mathbf{m},\mathbf{l},a,b,k}(\alpha) {\rm d}\alpha = \frac{M_{\PP,V,X,\mathbf{m},\mathbf{l},a,b}}{k+1} < \infty \, .
    \ee
    Now, we need to prove that we can exchange the derivatives and the integral. To this aim we initially consider $|b|=1$; so, we have to find a dominating function. Luckily, $\xi_{\PP,V,X,\mathbf{m},\mathbf{l},0,b,k}(\alpha)$ does the job. Then, we can apply the Leibniz integral rule to exchange the derivative and the integral. For $|b|>1$ we iterate the same argument. After that, for any $V$, any $X$, any $a$ and $b$, any $k$, and any $\mathbf{m}$ and $\mathbf{l}$, we can estimate the expression in~\eqref{eq:abs_int} as
    \bb 
        &\sup_{\mathbf{r}\in\R^{2n}} \left| \int_{0}^{1} \mathbf{r}^a\partial_b (1-\alpha)^{k}\PP(\mathbf{r})  \exp\left( -\frac14 \left(\Omega\mathbf{r}\right)^{\intercal}(V+\alpha X)\Omega\mathbf{r} + i(\mathbf{m}+\alpha \mathbf{l})^{\intercal}\Omega\mathbf{r} \right)  {\rm d}\alpha \right| \\
        &\leq \sup_{\mathbf{r}\in\R^{2n}}  \int_{0}^{1}\left| \mathbf{r}^a\partial_b (1-\alpha)^{k}\PP(\mathbf{r})  \exp\left( -\frac14 \left(\Omega\mathbf{r}\right)^{\intercal}(V+\alpha X)\Omega\mathbf{r} + i(\mathbf{m}+\alpha \mathbf{l})^{\intercal}\Omega\mathbf{r} \right) \right| {\rm d}\alpha \\
        &\leq \int_{0}^{1}(1-\alpha)^{k}\sup_{\mathbf{r}\in\R^{2n}} \left| \mathbf{r}^a\partial_b \PP(\mathbf{r})  \exp\left( -\frac14 \left(\Omega\mathbf{r}\right)^{\intercal}(V+\alpha X)\Omega\mathbf{r} + i(\mathbf{m}+\alpha \mathbf{l})^{\intercal}\Omega\mathbf{r} \right) \right| {\rm d}\alpha \\ 
        & \leq \int_{0}^{1}\xi_{\PP,V,X,\mathbf{m},\mathbf{l},a,b}(\alpha) {\rm d}\alpha < \infty \, .
    \ee 
    Hence, we conclude. 
    \end{proof}
    After that, we would like to lift up these regularity properties to the Schwartz operators. We can easily do that thanks to the previous Lemma combined with Proposition~\ref{prop:schwartz_iso}. In this way we obtain the following Corollary:

\begin{cor} \label{coroll:bound_transform}
    For any $n,k\in\N$, for any $\mathbf{m},\mathbf{l}\in\R^{2n}$, for any polynomial with complex coefficients $\PP(\mathbf{r})$, for any $V>0$, and any $X$ symmetric matrix such that $V+X > 0$, it holds that
    \bb 
        \left\|\mathcal{W}\left( \int_{0}^{1}  (1-\alpha)^{k}\PP(\mathbf{r})  \exp\left( -\frac14 \left(\Omega\mathbf{r}\right)^{\intercal}(V+\alpha X)\Omega\mathbf{r} + i(\mathbf{m}+\alpha \mathbf{l})^{\intercal}\Omega\mathbf{r} \right)  {\rm d}\alpha \right) \right\|_1 < \infty \, .
    \ee 
    Here, the linear map $\mathcal{W}$ is the Fourier--Weyl transform defined in Section~\ref{sub:prelCV}. 
\end{cor}

\begin{proof}
    The strategy we employ to prove the above inequality is to prove that the operator
    \bb 
        \mathcal{W}\left( \int_{0}^{1}  (1-\alpha)^{k}\PP(\mathbf{r})  \exp\left( -\frac14 \left(\Omega\mathbf{r}\right)^{\intercal}(V+\alpha X)\Omega\mathbf{r} + i(\mathbf{m}+\alpha \mathbf{l})^{\intercal}\Omega\mathbf{r} \right)  {\rm d}\alpha \right)
    \ee 
    is a Schwartz operator. Then, the bound follows directly from Proposition~\ref{prop:schw_estimate}. Due to Lemma~\ref{lemma:domination} we know that, for any $\PP(\mathbf{r})$, $V$, $X$, $\mathbf{m}$, $\mathbf{l}$, and $k$, the function 
    \bb
        f_{\PP,V,X,\mathbf{m},\mathbf{l},0,0,k}(\mathbf{r}) = \int_{0}^{1} (1-\alpha)^{k}\PP(\mathbf{r})  \exp\left( -\frac14 \left(\Omega\mathbf{r}\right)^{\intercal}(V+\alpha X)\Omega\mathbf{r} + i(\mathbf{m}+\alpha \mathbf{l})^{\intercal}\Omega\mathbf{r} \right) {\rm d}\alpha
    \ee 
    belongs to the Schwartz space. Hence, by invoking Proposition~\ref{prop:schwartz_iso}, we conclude.
    \end{proof}

We have proved that the operator obtained by the Fourier--Weyl transform of any integral remainder of the Taylor expansion of a Gaussian function is a Schwartz operator, therefore it is a trace-class operator. Now, we have to prove some regularity properties of the terms of the Taylor expansion of a generic Gaussian function. To do so, we define the positive cone of the real symmetric matrices as 
\bb
    \mathcal{P}^+_{2n}(\R) \coloneqq \left\{ M\in{\R}^{2n, 2n} : M^{\intercal} = M, \, M > 0 \right\} \, . 
\ee
Then, we state the following Lemma:

\begin{lemma} \label{lemma:cont_pol_schw}
    For any $n\in\mathbb{N}$, symmetric $V > 0$, $\mathbf{m}\in\R^{2n}$, and any polynomial $\PP(\mathbf{r})$ with complex coefficients, the map $\Phi_{\PP} : \mathcal{P}^+_{2n}(\R) \times \R^{2n} \to \pazocal{S}(\mathcal{\R}^{2n})$, defined as
    \bb
        \Phi_{\PP}[V,\mathbf{m}] \coloneqq \pazocal{P}(\mathbf{r})\exp\left( -\frac14 \left(\Omega\mathbf{r}\right)^{\intercal}V\Omega\mathbf{r} + i\mathbf{m}^{\intercal}\Omega\mathbf{r} \right) \, ,
    \ee
    is continuous in the Fr\'{e}chet topology of the Schwartz space.
\end{lemma}

\begin{proof}
    According to Definition~\ref{def:fre_cont}, in order to prove continuity in the Schwartz space, we need to show that the above map is continuous with respect to all of the seminorms in~\eqref{def:schwartz_space}; that is, we need to show that for all $a,b\in \N^k$ it holds that
    \bb \label{eq:lim_fre}
        \lim_{(X,\mathbf{l})\to (0,0)} \sup_{\mathbf{r}\in\R^{2n}} \bigg| \mathbf{r}^a\partial_b \PP(\mathbf{r})\bigg( &\exp\left( -\frac14 \left(\Omega\mathbf{r}\right)^{\intercal}(V + X)\Omega\mathbf{r} + i(\mathbf{m} + \mathbf{l})^{\intercal}\Omega\mathbf{r} \right) \\
        &- \exp\left( -\frac14 \left(\Omega\mathbf{r}\right)^{\intercal}V\Omega\mathbf{r} + i\mathbf{m}^{\intercal}\Omega\mathbf{r} \right) \bigg) \bigg| = 0 \, .
    \ee 
    Here $X$ is always a $2n\times 2n$ real symmetric matrix. In Proposition~\ref{prop:taylor_exp} we write the Taylor expansion of a Gaussian function, then by truncating this expression to the first order and using the integral remainder of the Taylor formula, we obtain
    \bb \label{eq:expanded}
        &\exp\left( -\frac14 \left(\Omega\mathbf{r}\right)^{\intercal}(V + X)\Omega\mathbf{r} + i(\mathbf{m} + \mathbf{l})^{\intercal}\Omega\mathbf{r} \right) \\
        &= \exp\left( -\frac14 \left(\Omega\mathbf{r}\right)^{\intercal}V\Omega\mathbf{r} + i\mathbf{m}^{\intercal}\Omega\mathbf{r} \right) \\
        &\quad +\int_{0}^{1}\bigg[\exp\left( -\frac14 \left(\Omega\mathbf{r}\right)^{\intercal}(V+\alpha X)\Omega\mathbf{r} + i(\mathbf{m}+\alpha \mathbf{l})^{\intercal}\Omega\mathbf{r} \right)
        \left(- \frac14 \left(\Omega\mathbf{r}\right)^{\intercal}X\Omega\mathbf{r}  + i \mathbf{l}^{\intercal}\Omega\mathbf{r} \right)\bigg]{\rm d}\alpha \, .   
    \ee
    Before we can prove the expression~\eqref{eq:lim_fre}, we need to prove that for any multi-index $b$ the following relation holds
    \bb
        &\partial_b \PP(\mathbf{r})\int_{0}^{1} \exp\left( -\frac14 \left(\Omega\mathbf{r}\right)^{\intercal}(V+\alpha X)\Omega\mathbf{r} + i(\mathbf{m}+\alpha \mathbf{l})^{\intercal}\Omega\mathbf{r} \right) \left(- \frac14 \left(\Omega\mathbf{r}\right)^{\intercal}X\Omega\mathbf{r}  + i \mathbf{l}^{\intercal}\Omega\mathbf{r} \right) {\rm d}\alpha \\
        &= \int_{0}^{1}\partial_b\PP(\mathbf{r}) \exp\left( -\frac14 \left(\Omega\mathbf{r}\right)^{\intercal}(V+\alpha X)\Omega\mathbf{r} + i(\mathbf{m}+\alpha \mathbf{l})^{\intercal}\Omega\mathbf{r} \right) \left(- \frac14 \left(\Omega\mathbf{r}\right)^{\intercal}X\Omega\mathbf{r} + i\mathbf{l}^{\intercal}\Omega\mathbf{r} \right) {\rm d}\alpha \, .
    \ee
    To do so, we recall the first part of the proof of Lemma~\ref{lemma:domination}, where the above expression is proved in greater generality. 
    We are now ready to prove the main statement of this Lemma. By exchanging the integral and the derivative and then plugging equation~\eqref{eq:expanded} into~\eqref{eq:lim_fre}, we obtain
     \bb
        &\lim_{(X,\mathbf{l})\to(0,0)} \sup_{\mathbf{r}\in\R^{2n}} \left | \int_{0}^{1}\mathbf{r}^a \partial_b\PP(\mathbf{r}) \exp\left( -\frac14 \left(\Omega\mathbf{r}\right)^{\intercal}(V+\alpha X)\Omega\mathbf{r} + i(\mathbf{m}+\alpha \mathbf{l})^{\intercal}\Omega\mathbf{r} \right) \left(- \frac14 \left(\Omega\mathbf{r}\right)^{\intercal}X\Omega\mathbf{r} + i\mathbf{l}^{\intercal}\Omega\mathbf{r} \right) {\rm d}\alpha\right | \\
        &\quad \leq \lim_{(X,\mathbf{l})\to(0,0)} \frac14\sum_{j,k=1}^{2n}|\tilde{X}_{jk}|\sup_{\mathbf{r}\in\R^{2n}} \left | \int_{0}^{1}\mathbf{r}^a \partial_b\PP(\mathbf{r}) \exp\left( -\frac14 \left(\Omega\mathbf{r}\right)^{\intercal}(V+\alpha X)\Omega\mathbf{r} + i(\mathbf{m}+\alpha \mathbf{l})^{\intercal}\Omega\mathbf{r} \right)r_jr_k {\rm d}\alpha\right | \\
        & \qquad +\lim_{(X,\mathbf{l})\to(0,0)} \sum_{j=1}^{2n}|\tilde{l}_j|\sup_{\mathbf{r}\in\R^{2n}}\left | \int_{0}^{1}\mathbf{r}^a \partial_b\PP(\mathbf{r}) \exp\left( -\frac14 \left(\Omega\mathbf{r}\right)^{\intercal}(V+\alpha X)\Omega\mathbf{r} + i(\mathbf{m}+\alpha \mathbf{l})^{\intercal}\Omega\mathbf{r} \right)r_j {\rm d}\alpha\right | \, .
     \ee 
    Here $\tilde{X}\coloneqq \Omega^{\intercal}X\Omega$, and $\tilde{\mathbf{l}}\coloneqq \Omega^{\intercal}\mathbf{l}$. Since we can choose $X$ arbitrarily small in any given norm, we have that $V + \alpha X > 0$ for any $\alpha\in[0,1]$. Both $\tilde{X}$ and $\tilde{\mathbf{l}}$ vanish as the pair $(X,\mathbf{l})$ goes to zero, while the suprema of the integrals are always bounded thanks to Lemma~\ref{lemma:domination}. Hence, we conclude.
\end{proof}

\begin{cor} \label{coroll:schwartz_op_cont}
    For any $n\in\mathbb{N}$, $V > 0$, $\mathbf{m}\in\R^{2n}$ and any polynomial $\PP(\mathbf{r})$ with complex coefficients the map $\Psi_{\PP} : \mathcal{P}^+_{2n}(\R) \times \R^{2n} \to \pazocal{S}(\mathcal{\mathcal{H}})$, defined as
    \bb  
        \Psi_{\PP}[V,\mathbf{m}] &\coloneqq \mathcal{W} \left(\pazocal{P}(\mathbf{r})\exp\left( -\frac14 \left(\Omega\mathbf{r}\right)^{\intercal}V\Omega\mathbf{r} + i\mathbf{m}^{\intercal}\Omega\mathbf{r} \right)\right) \\
        &= (2\pi)^{-n}\int_{\R^{2n}}\pazocal{P}(\mathbf{r})\exp\left( -\frac14 \left(\Omega\mathbf{r}\right)^{\intercal}V\Omega\mathbf{r} + i\mathbf{m}^{\intercal}\Omega\mathbf{r} \right)\hat{D}_{-\mathbf{r}} {\rm d}\mathbf{r} \, ,
    \ee
    is continuous in the Fr\'{e}chet topology of the Schwartz operators.
\end{cor}

\begin{proof}
    Due to Proposition~\ref{prop:schwartz_iso} we know that it is a continuous isomorpohism between the Schwartz space and the set of Schwartz operators. In Lemma~\ref{lemma:cont_pol_schw} we have proved that a polynomial times a Gaussian function is continuous in the topology of the Schwartz space. Hence, we conclude.
\end{proof}

Equipped with these results we can now prove the main result of this section concerning the regularity of the Gaussian states as functions of their statistical moments.

\begin{thm} \label{thm:gauss_regularity}
    For any number of modes $n\in\mathbb{N}$ the map acting as $(V,\mathbf{m}) \mapsto \omega(V,\mathbf{m})$, where
    \bb 
        \omega(V,\mathbf{m}) \coloneqq (2\pi)^{-n}\int_{\R^{2n}} \exp\left( -\frac14 (\Omega\mathbf{r})^{\intercal}V(\Omega\mathbf{r}) + i\mathbf{m}^{\intercal}\Omega\mathbf{r} \right) \hat{D}_{-\mathbf{r}}{\rm d}\mathbf{r} \, ,
    \ee
    is continuously differentiable; that is, it belongs to $C^{1}\left( \mathcal{P}^+_{2n}(\R)\times \R^{2n} \to \TT(\mathcal{H}) \right)$. Moreover, for any $V$ and $\mathbf{m}$ the Jacobian operator $\pazocal{J}_{V,\mathbf{m}}$ acts on the variations $X$ and $\mathbf{l}$ as
    \bb 
        \pazocal{J}_{V,\mathbf{m}}[X,\mathbf{l}] = \mathcal{W}\left( \exp\left( -\frac14 \left(\Omega\mathbf{r}\right)^{\intercal}V\Omega\mathbf{r} + i\mathbf{m}^{\intercal}\Omega\mathbf{r} \right) \left( -\frac14\left(\Omega\mathbf{r}\right)^{\intercal}X\Omega\mathbf{r} + i\mathbf{l}^{\intercal}\Omega\mathbf{r} \right) \right) = \partial_{X,\mathbf{l}}\,\rho(V,\mathbf{m}) \, .
    \ee 
    Where $X$ is a generic $2n \times 2n$ symmetric matrix, $\mathbf{l\in\R^{2n}}$, and $\partial_{X,\mathbf{l}}\,\rho(V,\mathbf{m})$ is the directional derivative defined in Equation~\ref{eq_der_def_0}.
\end{thm}

\begin{proof}
    First, we do not specify the norm of the set $\mathcal{P}_{2n}^+(\R)\times \R^{2n}$, because it is an open subset of a finite-dimensional vector space, so all norms are equivalent; while in the set of trace-class operators $\TT(\mathcal{H})$ we use the operator trace-norm $\| \cdot\|_1$. We start by noticing that $\omega(V,\mathbf{m})$ is a Schwartz operator, so for any $V$ and $\mathbf{m}$ it is a trace-class operator. In Proposition~\ref{prop:taylor_exp} we recall the Taylor expansion of the Gaussian function, for any symmetric real matrix $X$ and vector $\mathbf{l}$. We call the first  of the expansion as
    \bb
        g^{(1)}_{V,X,\mathbf{m},\mathbf{l}}(\mathbf{r}) \coloneqq \exp\left( -\frac14 \left(\Omega\mathbf{r}\right)^{\intercal}V\Omega\mathbf{r} + i\mathbf{m}^{\intercal}\Omega\mathbf{r} \right) \left( -\frac14\left(\Omega\mathbf{r}\right)^{\intercal}X\Omega\mathbf{r} + i\mathbf{l}^{\intercal}\Omega\mathbf{r} \right) \, .
    \ee
    While, we take the second-order integral remainder and we call it     
    \bb
        g^{(2)}_{V,X,\mathbf{m},\mathbf{l}}(\mathbf{r}) \coloneqq \int_{0}^{1}\bigg[&(1-\alpha)\exp\left( -\frac14 \left(\Omega\mathbf{r}\right)^{\intercal}(V+\alpha X)\Omega\mathbf{r} + i(\mathbf{m}+\alpha\mathbf{l})^{\intercal}\Omega\mathbf{r} \right) \\
        &\times\left( \frac{1}{32}\left[(\Omega\mathbf{r})^{\intercal}X\Omega\mathbf{r}\right]^2 - \frac12 \left(\mathbf{l}^{\intercal}\Omega\mathbf{r} \right)^2 -\frac{i}{4}\left[\left(\Omega\mathbf{r}\right)^{\intercal}X\Omega\mathbf{r}\right]\left[\mathbf{l}^{\intercal}\Omega\mathbf{r}\right] \right)\bigg] {\rm d}\alpha \, .
    \ee
    In order to prove the differentiability of $\omega(V,\mathbf{m})$ we have to evaluate the following limit and prove that it vanishes,
    \bb \label{eq:lim_cen}
        \lim_{\left\|(X,\mathbf{l}) \right\|\to 0} \frac{\left\|\omega(V+X,\mathbf{m}+\mathbf{l})-\omega(V,\mathbf{m})-\mathcal{W}\left(g^{(1)}_{V,X,\mathbf{m},\mathbf{l}}\right)\right\|_1}{\left\|(X,\mathbf{l}) \right\|} \, ,
    \ee
    where $X$ is always a real $2n\times 2n$ matrix. 
    We know that any polynomial multiplied for a Gaussian function is a Schwartz function, so, if we apply the Fourier--Weyl transform, we always obtain a Schwartz operator, therefore a trace-class operator. In particular, this is true for any term of the Taylor expansion in Proposition~\ref{prop:taylor_exp}; therefore $\mathcal{W}\left(g^{(1)}_{V,X,\mathbf{m},\mathbf{l}}\right)$ is continuous. 
    Moreover, by Lemma~\ref{coroll:schwartz_op_cont} every term of the expansion is also continuous in the set of Schwartz operators as a function of the statistical moments, in addition Proposition~\ref{prop:schw_estimate} tells us that continuity in the set of Schwartz operators implies continuity of trace-class operators, then every term of the expansion is continuous in trace-norm as a function of the statistical moments $(V,\mathbf{m})$. 
    Finally $g^{(1)}_{V,X,\mathbf{m},\mathbf{l}}$ is linear in $(V,\mathbf{m})$ by construction. We now have to deal with the term $\mathcal{W}\left(g^{(2)}_{V,X,\mathbf{m},\mathbf{l}}\right)$. We want to give an estimate to the following norm: 
    \bb 
        &\left\|\mathcal{W}\left(g^{(2)}_{V,X,\mathbf{m},\mathbf{l}}\right)\right\|_1 \\
        &\leq \left\|\mathcal{W}\left(\int_{0}^{1}(1-\alpha)\exp\left( -\frac14 \left(\Omega\mathbf{r}\right)^{\intercal}(V+\alpha X)\Omega\mathbf{r} + i(\mathbf{m}+\alpha\mathbf{l})^{\intercal}\Omega\mathbf{r}\right) \frac{1}{32}\left[(\Omega\mathbf{r})^{\intercal}X\Omega\mathbf{r}\right]^2 {\rm d}\alpha  \right)\right\|_1 \\
        & \quad +  \left \| \mathcal{W}\left(\int_{0}^{1}(1-\alpha)\exp\left( -\frac14 \left(\Omega\mathbf{r}\right)^{\intercal}(V+\alpha X)\Omega\mathbf{r} + i(\mathbf{m}+\alpha\mathbf{l})^{\intercal}\Omega\mathbf{r} \right)\frac12 \left(\mathbf{l}^{\intercal}\Omega\mathbf{r} \right)^2 {\rm d}\alpha \right)\right  \|_1 \\
        & \quad + \left \| \mathcal{W}\left(\int_{0}^{1}(1-\alpha)\exp\left( -\frac14 \left(\Omega\mathbf{r}\right)^{\intercal}(V+\alpha X)\Omega\mathbf{r} + i(\mathbf{m}+\alpha\mathbf{l})^{\intercal}\Omega\mathbf{r}\right)\frac{i}{4}\left[\left(\Omega\mathbf{r}\right)^{\intercal}X\Omega\mathbf{r}\right]\left[\mathbf{l}^{\intercal}\Omega\mathbf{r}\right] {\rm d}\alpha  \right)\right  \|_1 \\ 
        &\leq \frac{1}{32}\|\tilde{X}\|^2_{\infty}\sum_{j,k,l,m=1}^{2n}\left\|\mathcal{W}\left(\int_{0}^{1}(1-\alpha)\exp\left( -\frac14 \left(\Omega\mathbf{r}\right)^{\intercal}(V+\alpha X)\Omega\mathbf{r} + i(\mathbf{m}+\alpha\mathbf{l})^{\intercal}\Omega\mathbf{r}\right)(B\mathbf{r})_j(B\mathbf{r})_k(B\mathbf{r})_l(B\mathbf{r})_m  \right)\right\|_1 \\
        & \quad + \frac12 \|\mathbf{l}\|^2_{\infty}\sum_{j,k=1}^{2n}\left\|\mathcal{W}\left(\int_{0}^{1}(1-\alpha)\exp\left( -\frac14 \left(\Omega\mathbf{r}\right)^{\intercal}(V+\alpha X)\Omega\mathbf{r} + i(\mathbf{m}+\alpha\mathbf{l})^{\intercal}\Omega\mathbf{r}\right)(\Omega\mathbf{r})_j(\Omega\mathbf{r})_k\right)\right\|_1 \\
        & \quad + \frac14 \|\tilde{X}\|_{\infty} \|\mathbf{l}\|_{\infty}\sum_{j,k,l=1}^{2n}\left\|\mathcal{W}\left(\int_{0}^{1}(1-\alpha)\exp\left( -\frac14 \left(\Omega\mathbf{r}\right)^{\intercal}(V+\alpha X)\Omega\mathbf{r} + i(\mathbf{m}+\alpha\mathbf{l})^{\intercal}\Omega\mathbf{r}\right)(\Omega\mathbf{r})_j(B\mathbf{r})_k(B\mathbf{r})_l\right)\right\|_1 \, .
    \ee 
    Here $\tilde{X} = \Omega^{\intercal}X\Omega$ as in Lemma~\ref{lemma:cont_pol_schw}, and $B$ is an orthogonal matrix such that $B\tilde{X}B^{\intercal} = D$, with $D$ diagonal. We recall that $\|\tilde{X}\|_{\infty}=\|X\|_{\infty}$. The three sums in the above expression are all bounded thanks to Corollary~\ref{coroll:bound_transform}. Moreover, since we $(X,\mathbf{l})$ belongs to a finite-dimensional vector space, then there exists a constant $K$ such that $\|(X,\mathbf{l})\| \leq K (\| X \|_{\infty} + \| \mathbf{l} \|_{\infty})$. Given all these pieces of information, we can finally evaluate the limit~\eqref{eq:lim_cen} as
    \bb
        &\lim_{\left\|(X,\mathbf{l}) \right\|\to 0} \frac{\left\|\omega(V+X,\mathbf{m}+\mathbf{l})-\omega(V,\mathbf{m})-\mathcal{W}\left(g^{(1)}_{V,X,\mathbf{m},\mathbf{l}}\right)\right\|_1}{\left\|(X,\mathbf{l}) \right\|} \\
        & \quad = \lim_{\left\|(X,\mathbf{l}) \right\|\to 0} \frac{\left\|\mathcal{W}\left( g^{(2)}_{V',X,\mathbf{m}',\mathbf{l}} \right)\right\|_1}{\left\|(X,\mathbf{l}) \right\|} \\ 
        &\quad \leq \lim_{\left\|(X,\mathbf{l}) \right\|\to 0}\frac{C_1\| X \|^2_{\infty} + C_2 \|\mathbf{l} \|^2_{\infty} + C_3 \|X \|_{\infty}\|\mathbf{l} \|_{\infty} }{K (\| X \|_{\infty} + \| \mathbf{l} \|_{\infty})} = 0 \, .
    \ee 
    Here $C_1$, $C_2$, and $C_3$ are positive constants. Since $\mathcal{W}\left(g^{(1)}_{V,X,\mathbf{m},\mathbf{l}}\right)$ is continuous in trace-norm as a function of $(V,\mathbf{m})$ we have proved that the map $(V,\mathbf{m})\mapsto \omega(V,\mathbf{m})$ belongs to $C^1\left( \mathcal{P}^+_{2n}(\R)\times \R^{2n} \to \TT(\mathcal{H}) \right)$.

\end{proof}

The theorem above allows us to write any difference of Gaussian state as the first-order integral remainder.

\begin{cor} \label{cor:int_rem}
For any $V$ and $W$ quantum covariance matrices, and $\mathbf{m}, \mathbf{t} \in \R^{2n}$, the difference between two Gaussian states reads:
\bb 
    \rho(W,\mathbf{t}) - \rho(V,\mathbf{m}) = \int_{0}^{1} \partial_{X,\mathbf{x}}\rho(V + \alpha X, \mathbf{m} + \alpha \mathbf{x}) {\rm d}\alpha \, ,
\ee 
where $X\coloneqq W - V$ and $\mathbf{x} \coloneqq \mathbf{t} - \mathbf{m}$.
\end{cor}

\begin{proof}
    A straightforward consequence of Theorem~\ref{thm:gauss_regularity} is that we can write the difference of two Gaussian states as the first-order integral remainder of the Taylor formula~\cite{Zeidler1995}. The only caveat we have to address is the fact that the integrand must be integrable. Since any Gaussian state $\rho(V,\mathbf{m})$ in continuously differentiable as a function of the pair $(V,\mathbf{m})$ considering the trace-norm, we have that, for any $X$ and $\mathbf{x}$, any directional derivative of a Gaussian state $\partial_{X,\mathbf{x}}\rho(V), \mathbf{m})$ is continuous in trace-norm. Then, it is also Bochner integrable. Hence, we conclude.
\end{proof}
Finally, we would like to stress that the tools used in the proof suffices to prove that Gaussian states are infinitely differentiable:

\begin{remark}
It is possible to extend the proof of Theorem~\ref{thm:gauss_regularity} to prove that, for any number of modes $n\in\mathbb{N}$, the map acting as $(V,\mathbf{m}) \mapsto \omega(V,\mathbf{m})$ is infinitely differentiable; that is, it belongs to $C^{\infty}\left( \mathcal{P}^+_{2n}(\R)\times \R^{2n} \to \TT(\mathcal{H}) \right)$.
\end{remark}

\section{Alternative approach based on the Gaussian noise channel}\label{Sec_proof_LL}
In this section, we present an alternative approach to establish stringent upper bounds on the trace distance between Gaussian states in terms of the norm distance of their first moments and covariance matrices. The forthcoming Theorem~\ref{thm0_ll} provides an upper bound that, while less tight than the one established using the derivative approach in Section~\ref{Sec_proof_derivative}, is proved using an insightful different method. This approach takes advantage of novel properties of the Gaussian noise channel introduced in Subsection~\ref{subsubsec_gauss_noise}, which may hold independent interest. We now state the bound derived in this section.
\begin{thm}\label{thm0_ll}
    Let $\rho(V,\mathbf{m})$ and $\rho(W,\mathbf{t})$ be Gaussian states with covariance matrices $V,W$ and first moments $\mathbf{t},\mathbf{m}$, respectively. Then, it holds that
\bb
\frac12 \left\| \rho(V,\mathbf{m}) - \rho(W,\mathbf{t}) \right\|_1 \leq \, \frac{1+\sqrt{5}}{8}\max\big\{\|V\|_\infty,\,\|W\|_\infty\big\}\, \|V-W\|_1 +\sqrt{1 - e^{-\frac12 \min\{\|V\|_\infty,\,\|W\|_\infty\}\, \|\mathbf{t}-\mathbf{m}\|_2^2}}\,.
 \ee
\end{thm}
\begin{proof}
    It suffices to apply the forthcoming Theorem~\ref{bound_trace_distance_thm} with the choice $\bar\phi(x)\coloneqq \frac{1+\sqrt{5}}{2}x$. Indeed, using the elementary bounds $\sinh(x/2) \leq \frac12\, x e^{x/2}$ and $1 - e^{-\frac{x^2}{1+4x}} \leq x^2 \leq x^2 e^x$, it is not difficult to show that 
    \bb
        \phi(x) \leq \frac{1+\sqrt{5}}{2}x \,,
    \ee
    where the function $\phi(x)$ is defined in the statement of Theorem~\ref{bound_trace_distance_thm} (note that $\frac{1+\sqrt{5}}{2}x$ is in fact the dominant term in the Taylor expansion of $\phi$ around $0$, i.e.\ $\phi(x) = \frac{1+\sqrt{5}}{2}x + O\big(x^2\big)$ for $x\to 0$.)
\end{proof}
The forthcoming Theorem~\ref{bound_trace_distance_thm} is the key result of this section, which directly implies Theorem~\ref{thm0_ll}.

\begin{thm} \label{bound_trace_distance_thm}
Let $\widebar{\phi}: [0,\infty) \to \R^+$ be a concave function such that $\widebar{\phi}(x) \geq \phi(x)$ for all $x\geq 0$, where the function $\phi: [0,\infty) \to [0,1]$ is defined by
\bb
\phi(x) \coloneqq \frac{1-e^{-x}}{2} + e^{-x/2} \sqrt{\sinh^2(x/2) + 1 - e^{-\frac{x^2}{1+4x}}}\, .
\label{phi_function}
\ee
    Let $\rho(V,\mathbf{m})$ and $\rho(W,\mathbf{t})$ be Gaussian states with covariance matrices $V,W$ and first moments $\mathbf{t},\mathbf{m}$, respectively. Then, it holds that
\bb
\frac12 \left\| \rho(V,\mathbf{m}) - \rho(W,\mathbf{t}) \right\|_1 \leq  2\, \widebar{\phi}\!\left(\frac18 \max\big\{\|V\|_\infty,\,\|W\|_\infty\big\}\, \|V-W\|_1 \right) +\sqrt{1 - e^{-\frac12 \min\{\|V\|_\infty,\,\|W\|_\infty\}\, \|\mathbf{t}-\mathbf{m}\|_2^2}}\,.
\label{bound_trace_distance} 
\ee
\end{thm}

\begin{rem}
By looking at a plot of the function $\phi$, it seems clear that $\phi$ itself is concave. However, proving this explicitly seems quite difficult. If this could be done, then in the above Theorem~\ref{bound_trace_distance_thm} we could take $\widebar{\phi} \coloneqq \phi$. Since $\widebar{\phi}$ is bounded, this would give us a bound on the trace distance between any two Gaussian states that is both faithful and bounded. Note that the right-hand side of~\eqref{eq_ineq_main}, on the contrary, is not bounded. In any case, we can definitely take $\widebar{\phi}(x) \coloneqq \frac{1+\sqrt{5}}{2}x$, which yields Theorem~\ref{thm0_ll} above.
\end{rem}

Before proving Theorem~\ref{bound_trace_distance_thm}, let us establish some preliminary results about the \emph{Gaussian noise channel} $\NN_K$, defined in subsection~\ref{subsubsec_gauss_noise}. Here, we recall that, given $K\ge0$, the Gaussian noise channel $\NN_K$ is a Gaussian channel that acts on covariance matrices as $V\mapsto V+K$ and leaves unchanged the first moments.

\begin{lemma}\label{lemma_K_pos}
    Let $K$ be a $2n\times2n$ real positive semi-definite matrix. Let $\widebar{\phi}: [0,\infty) \to \R^+$ be a function defined as in Theorem~\ref{bound_trace_distance_thm}. Then for all $M\ge1$ it holds that
    \bb
        \sup_{\rho_\G:\ \|V(\rho_\G)\|_\infty \leq M} \frac12 \left\| \left(\NN_K - \Id \right)(\rho_\G) \right\|_1\le\widebar{\phi}\left(\frac14\, M \Tr K \right)\,,
    \ee
    where the supremum is over all the Gaussian states $\rho_G$ whose covariance matrix $V(\rho_\G)$ has operator norm at most $M$. In particular, it holds that
    \bb
            \sup_{\rho_\G:\ \|V(\rho_\G)\|_\infty \leq M} \frac12 \left\| \left(\NN_K - \Id \right)(\rho_\G) \right\|_1\le\frac{1+\sqrt{5}}{8}\ M \Tr K \,.
    \ee
\end{lemma}

\begin{proof} 
Let us start by upper bounding the quantity $ \frac12 \left\| \left(\NN_K - \Id \right)(\ketbra{0}) \right\|_1$. Note that $\NN_K(\ketbra{0})$ is a Gaussian state with zero first moment and covariance matrix $\id + K$. Let us consider the Williamson decomposition
\bb
\id + K = S D S^\intercal\, ,\qquad D = \bigoplus_j d_j \id_2\, ,
\ee
with $S$ symplectic and $d_j\geq 1$. Moreover, let us define
\bb
\gamma \coloneqq (\id+K) \# \big( \Omega (\id+K)^{-1}\Omega^\intercal\big)\,,
\label{tighter_estimate_proof_eq03}
\ee
where we introduced the \emph{matrix geometric mean} defined as $A\#B\coloneqq \sqrt{A}\sqrt{A^{-1/2}BA^{-1/2}}\sqrt{A}$ for all positive definite matrices $A,B$~\cite{BHATIA}. By exploiting that for any invertible matrix $X$ it holds that~\cite{BHATIA} 
\bb
    (XAX^\dagger)\#(XBX^\dagger)=X(A\#B)X^\dagger\,,
\ee
it follows that
\bb
\gamma = SS^\intercal \,.
\label{tighter_estimate_proof_eq3}
\ee
In particular, note that $\gamma$ is the covariance matrix associated with the pure Gaussian state $U_S\ket{0}$. Introducing a parameter $\mu\geq 0$ to be fixed later, we now write that
\bb
\frac12 \left\| \left(\NN_K - \Id \right)(\ketbra{0}) \right\|_1 &= \frac12 \left\| U_S^\dag \NN_K(\ketbra{0}) U_S^{\vphantom{\dag}} - U_S^\dag \ketbra{0} U_S^{\vphantom{\dag}} \right\|_1 \\
&\leq \frac12 \left\| U_S^\dag \NN_K(\ketbra{0}) U_S^{\vphantom{\dag}} - \mu \ketbra{0} \right\|_1 + \frac12 \left\| \mu \ketbra{0} - U_S^\dag \ketbra{0} U_S^{\vphantom{\dag}} \right\|_1 \\
&\eqt{(i)} \frac{1-\mu}{2} + \left(\mu - \prodno_j \frac{2}{1+d_j}\right)_+ + \sqrt{\left(\frac{1+\mu}{2}\right)^2 - \mu |\!\braket{0|U_S|0}\!|^2} \\
&\eqt{(ii)} \frac{1-\mu}{2} + \left(\mu - \prodno_j \frac{2}{1+d_j}\right)_+ + \sqrt{\left(\frac{1+\mu}{2}\right)^2 - \mu\, \frac{2^n}{\sqrt{\det(\id+\gamma)}}} \\
&\leqt{(iii)} \frac{1-\mu}{2} + \left(\mu - e^{-\frac14 \Tr K}\right)_+ + \sqrt{\left(\frac{1+\mu}{2}\right)^2 - \mu\, \frac{2^n}{\sqrt{\det(\id+\gamma)}}}\, .
\label{tighter_estimate_proof_eq4}
\ee
The justification of the above manipulations is as follows. In~(i) we used the fact that $U_S^\dag \NN_K(\ketbra{0}) U_S^{\vphantom{\dag}}$ is a Gaussian state with zero first moment and covariance matrix equal to $S^{-1} (\id+K) S^{-\intercal} = D$, and it is thus a product of thermal states with mean photon number $(d_j-1)/2$~\cite{BUCCO}; the trace distance from the vacuum state is easy to compute, given that the two density operators commute. In~(ii) we substituted the formula~\cite[Eq~(4.51)]{BUCCO}. In~(iii), instead, we observed that
\bb
\ln \prod_j \frac{1\!+\!d_j}{2} = \sum_j \ln \left(1 + \frac12(d_j\!-\!1) \right) \leq \frac12 \sum_j (d_j-1) = \frac14 \left( \Tr D - 2n \right) \leq \frac14 \big( \Tr [\id\! +\! K] - 2n \big) = \frac14 \Tr K\, ,
\ee
where the inequality $\Tr D \leq \Tr[\id+K]$ follows from~\cite[Theorem~5(i)]{bhatia15}. 

We can simplify~\eqref{tighter_estimate_proof_eq4} by choosing $\mu = e^{-\frac14 \Tr K}$, which gives immediately
\bb
\frac12 \left\| \left(\NN_K - \Id \right)(\ketbra{0}) \right\|_1 &\leq \frac{1-e^{-\frac14 \Tr K}}{2} + \sqrt{\left(\frac{1+e^{-\frac14 \Tr K}}{2}\right)^2 - e^{-\frac14 \Tr K}\, \frac{2^n}{\sqrt{\det(\id+\gamma)}}} \\
&= \frac{1-e^{-\frac14 \Tr K}}{2} + \sqrt{\left(\frac{1 - e^{-\frac14 \Tr K}}{2}\right)^2 + e^{-\frac14 \Tr K}\left(1 - \frac{2^n}{\sqrt{\det(\id+\gamma)}}\right)} .
\label{tighter_estimate_proof_eq5}
\ee
Numerically, this choice of $\mu$ seems to be optimal.

We now set out to manipulate the term containing the determinant of $\id+\gamma$. Since $\gamma$ is a pure covariance matrix, it has spectrum of the form $\bigcup_j \big\{ g_j,\, 1/g_j\big\}$, with $g_j > 0$. Thus,
\begin{align}
\sqrt{\det(\id+\gamma)}\ &=\ \prod_j \sqrt{2 + g_j + 1/g_j}   \\
&\leqt{(iv)}\ \prod_j 2 \left( 1 + \frac18 \left( g_j + 1/g_j - 2 \right) \right)   \\
&=\ 2^n \exp\left[ \sumno_j \ln \left( 1 + \frac18 \left( g_j + 1/g_j - 2 \right) \right) \right]   \\
&\leqt{(v)}\ 2^n \exp\left[ \frac18 \sumno_j \left( g_j + 1/g_j - 2 \right) \right]   \\
&= 2^n \exp\left[ \frac18 \Tr\left[ \gamma - \id \right] \right] \label{tighter_estimate_proof_eq6} \\
&\leqt{(vi)}\ 2^n \exp\left[ \frac18 \Tr\left[ \frac{\id + K + \Omega(\id+K)^{-1}\Omega^\intercal}{2} - \id \right] \right]   \\
&\leqt{(vii)}\ 2^n \exp\left[ \frac{1}{16} \Tr\left[ 
(\id+K)^{-1} - \id + K \right] \right]   \\
&=\ 2^n \exp\left[ \frac{1}{16} \Tr \left[ \frac{K^2}{\id+K} \right] \right]   \\
&\leqt{(viii)}\ 2^n \exp\left[ \frac{1}{16} \frac{(\Tr K)^2}{1+\Tr K} \right]   .
\end{align}
The `miracle' here has happened on the second-to-last line, where the two linear terms have cancelled out. In the above derivation: (iv)~follows from the elementary inequality
\bb
\sqrt{2+x+1/x} \leq 2 + \frac14 \left(x+1/x-2\right) ,
\ee
valid for all $x > 0$; in~(v) we used simply $\ln(1+x) \leq x$; (vi)~holds due to the definition of $\gamma$ given in~\eqref{tighter_estimate_proof_eq3}, together with the AM-GM inequality $A\#B \leq (A+B)/2$ for matrix means~\cite{BHATIA}; and in~(vii) we used the cyclicity of the trace. 
To justify~(viii), define, 
for an arbitrary $A\geq 0$ with spectral decomposition $A = \sum_i a_i \ketbra{i}$, the probability distribution $p_i \coloneqq \frac{a_i}{\Tr A}$; then, we have that
\bb
\Tr \left[ \frac{A^2}{\id+A} \right] &= \sum_i \frac{a_i^2}{1+a_i} \\
&= \sum_i \frac{p_i^2 (\Tr A)^2}{1 + p_i \Tr A} \\
&= \frac{(\Tr A)^2}{1 + \Tr A} \sum_i p_i\, \frac{p_i (1 + \Tr A)}{1 + p_i \Tr A} \\
&\leq \frac{(\Tr A)^2}{1 + \Tr A} \sum_i p_i \\
&= \frac{(\Tr A)^2}{1 + \Tr A}\, .
\ee
This completes the justification of~\eqref{tighter_estimate_proof_eq6}.

Plugging~\eqref{tighter_estimate_proof_eq6} into~\eqref{tighter_estimate_proof_eq5} yields the estimate
\bb
 \frac12 \left\| \left(\NN_K - \Id \right)(\ketbra{0}) \right\|_1\leq \phi\left(\frac14 \Tr K\right) \leq \widebar{\phi}\left(\frac14 \Tr K\right) \,.
\label{tighter_estimate_proof_eq8}
\ee

Now, for fixed $M\ge0$, let us consider the following quantity:
\bb
    \sup_{\rho_\G:\ \|V(\rho_\G)\|_\infty \leq M} \frac12 \left\| \left(\NN_K - I \right)(\rho_\G) \right\|_1\, ,
\label{tighter_estimate_proof_eq9}
\ee
where the supremum is over all Gaussian states $\rho_\G$ whose covariance matrix $V(\rho_\G)$ has operator norm at most $M$. The supremum is clearly achieved when $\rho_\G$ has zero first moment, because displacements commute with $\NN_K$, and also when $\rho_\G$ is pure, because any mixed Gaussian state $\rho_\G$ with covariance matrix $V(\rho_\G)$ can be written as a convex combination of displaced pure Gaussian states with covariance matrices $\theta \leq V(\rho_\G)$ (see Lemma~\ref{lemma_mixed_gauss}). Let $\theta$ be a generic pure covariance matrix such that $\theta \leq V(\rho_\G)$, so that $\|\theta\|_\infty\leq M$, and let $\psi_\G$ be the pure Gaussian state with covariance matrix $\theta$ and zero first moment. We write $\theta = ZZ^\intercal$, with $Z$ symplectic, so that $\psi_\G = U_Z^{\vphantom{\dag}} \ketbra{0} U_Z^\dag$. 
We now have that
\begin{align}
\frac12 \left\| \left(\NN_K - I \right)(\psi_G) \right\|_1\ &\eqt{(ix)}\ \frac12 \left\| U_Z^{\dag} \NN_K\big( U_Z^{\vphantom{\dag}} \ketbra{0} U_Z^{\dag} \big) U_Z^{\vphantom{\dag}} - \ketbra{0} \right\|_1   \\
&\eqt{(x)}\ \frac12 \left\| \NN_{Z^{-1} K Z^{-\intercal}}(\ketbra{0}) - \ketbra{0} \right\|_1   \\
&\leq\ F\big(Z^{-1} K Z^{-\intercal}\big)   \\
&\leqt{(xi)}\ \widebar{\phi}\left(\frac14 \Tr Z^{-1} K Z^{-\intercal} \right)   \\
&= \widebar{\phi}\left(\frac14 \Tr \theta^{-1} K \right) \label{tighter_estimate_proof_eq15} \\
&\leqt{(xii)}\ \widebar{\phi}\left(\frac14\, \big\| \theta^{-1}\big\|_\infty \Tr K \right)   \\
&\eqt{(xiii)}\ \widebar{\phi}\left(\frac14\, \| \theta\|_\infty \Tr K \right)   \\
&\leqt{(xiv)}\ \widebar{\phi}\left(\frac14\, M \Tr K \right)\, .  
\end{align}
Here: (ix)~is by unitary invariance; (x)~holds because the Gaussian channel given by $\rho \mapsto U_Z^{\dag} \NN_K\big( U_Z^{\vphantom{\dag}} \rho U_Z^{\dag} \big) U_Z^{\vphantom{\dag}}$ acts on input covariance matrices as $V \mapsto Z^{-1} \left( ZVZ^\intercal +K \right) Z^{-\intercal} = V + Z^{-1} K Z^{-\intercal}$, does not change the first moment, and hence it coincides with $\NN_{Z^{-1} K Z^{-\intercal}}$; in~(xi) we employed~\eqref{tighter_estimate_proof_eq8}, (xii)~follows from H\"older's inequality, remembering that $K\geq 0$ is positive semi-definite and observing that $\widebar{\phi}\geq \phi$, being concave, must be non-decreasing, because $\lim_{x\to\infty} \phi(x) = 1$; in~(xiii) we noticed that $\theta^{-1} = \Omega \theta \Omega^\intercal$, because $\theta$ is symplectic; finally, in~(xiv) we remembered that $\|\theta\|_\infty\leq M$, leveraging once again the fact that $\widebar{\phi}$ is non-decreasing.
\end{proof}

We are now ready to prove Theorem~\ref{bound_trace_distance_thm}.
\begin{proof}[Proof of Theorem~\ref{bound_trace_distance_thm}]
We use a similar idea of~\cite{mele2024learningquantumstatescontinuous}. First, given a covariance matrix $V$, let $\rho(V,0)$ be the Gaussian state with zero first moment and covariance matrix equal to $V$. Given two covariance matrices $V$ and $W$, our crucial idea is to consider a matrix $X$ (to be chosen later) such that 
\bb\label{constX}
    X&\ge0\,,\\
    X+V-W&\ge0\,,
\ee
so that 
\bb
    \NN_{X}\!\left(\rho(V,0)\right)=\rho\!\left(X+V,0\right)=\NN_{X+V-W}\!\left(\rho(W,0)\right)\,.
\ee
Consequently, we have that
\bb
    \frac12\|\rho(V,0)-\rho(W,0)\|_1&\le \frac12\|\rho(V,0)-\NN_{X}(\rho(V,0))\|_1+\frac12\|\NN_{X+V-W}(\rho(W,0))-\rho(W,0)\|_1\\
    &\le\frac12\|(\NN_{X}-\Id)(\rho(V,0))\|_1+\frac12\|(\NN_{X+V-W}-\Id)(\rho(W,0))\|_1\\
    &\leqt{(i)} \widebar{\phi}\left(\frac14\, \|V\|_\infty \Tr X \right)+\widebar{\phi}\left(\frac14\, \|W\|_\infty \Tr[ X+V-W] \right)\\
    &\leqt{(ii)} 2\widebar{\phi}\left(\frac18\, \|V\|_\infty \Tr X +\frac18\, \|W\|_\infty \Tr[ X+V-W] \right)\\
    &\le 2\widebar{\phi}\left(\frac18\, \max(\|V\|_\infty,\|W\|_\infty ) \Tr [2X+V-W] \right)\,,
\ee
where in (i) we exploited Lemma~\ref{lemma_K_pos} and (ii) follows by concavity of $\widebar{\phi}$. Now let us choose $X$ in order to minimise $\Tr X$, satisfying the constraints in~\eqref{constX}. That is, we choose 
\bb
    X\coloneqq (W-V)_+= \frac{W-V+|W-V|}{2}\,.
\ee
Hence, we obtain that
\bb
    \frac12\|\rho(V,0)-\rho(W,0)\|_1&\le2\widebar{\phi}\left(\frac18\, \max(\|V\|_\infty,\|W\|_\infty ) \Tr|W-V| \right)\\
    &=2\widebar{\phi}\left(\frac18\, \max(\|V\|_\infty,\|W\|_\infty ) \|W-V\|_1 \right)\,.
\ee
In the general case of Gaussian states with non-zero first moments, we can apply the following reasoning. Let $\rho(V,\mathbf{m})$ be the Gaussian state with first moment $\mathbf{m}$ and covariance matrix $V$. For any first moments $\mathbf{m},\mathbf{t}$ and covariance matrices $V,W$, we have that
\bb\label{eq_LL_1}
    \frac12\left\|\rho(V,\mathbf{m})-\rho(W,\mathbf{t})\right\|_1&\le \frac12\|\rho(V,\mathbf{m})-\rho(W,\mathbf{m})\|_1+\frac12\|\rho(W,\mathbf{m})-\rho(W,\mathbf{t})\|_1\\
    &=\frac12\|\rho_{V,0}-\rho_{W,0}\|_1+\frac12\|\rho(W,\mathbf{m})-\rho(W,\mathbf{t})\|_1\\
    &\le 2\widebar{\phi}\left(\frac18\, \max(\|V\|_\infty,\|W\|_\infty ) \|W-V\|_1 \right)  +\frac12\|\rho(W,\mathbf{m})-\rho(W,\mathbf{t})\|_1\,.
\ee
Finally, by writing the Williamson decomposition $W=SDS^\intercal$ and applying Lemma~\ref{lemma_mixed_gauss}, we have that
\bb\label{eq_LL_2}
    \frac12\|\rho(W,\mathbf{m})-\rho(W,\mathbf{t})\|_1&\leqt{(iii)}\frac12\|\rho_{SS^\intercal,\mathbf{m}}-\rho_{SS^\intercal,\mathbf{t}}\|_1\\
&\eqt{(iv)} \sqrt{1 - \Tr [\rho_{SS^\intercal,\mathbf{m}}\rho_{SS^\intercal,\mathbf{t}}]} \\
&\eqt{(v)} \sqrt{1 - e^{-\frac12 (\mathbf{m}-\mathbf{t})^\intercal (SS^\intercal)^{-1} (\mathbf{m}-\mathbf{t})}} \\
&\leq \sqrt{1 - e^{-\frac12 \|(SS^\intercal)^{-1}\|_\infty \|\mathbf{m}-\mathbf{t}\|^2}} \\
&\eqt{(vi)}\sqrt{1 - e^{-\frac12 \|SS^\intercal\|_\infty \|\mathbf{m}-\mathbf{t}\|^2}} \\
&\leqt{(vii)} \sqrt{1 - e^{-\frac12 \|W\|_\infty \|\mathbf{m}-\mathbf{t}\|^2}}\, .
\ee
Here, in (iii), we applied the monotonicity of the trace norm under quantum channels; in (iv), we exploited that $\rho_{SS^\intercal,\mathbf{m}}$ and $\rho_{SS^\intercal,\mathbf{t}}$ are pure states; in (v), we employed the known formula for the overlap between Gaussian states~\cite[Eq~(4.51)]{BUCCO}; in (vi), we exploited that for any symplectic matrix $T$ it holds that 
\bb
    \|T^{-1}\|_\infty= \|\Omega T^\intercal\Omega^\intercal\|_\infty=\|T^\intercal\|_\infty=\|T\|_\infty\,,
\ee
where we exploited that $T\Omega T^\intercal =\Omega$ and the fact that $\Omega$ is orthogonal; finally, in (vii), we used the fact that $SS^\intercal\le SDS^\intercal=W$. By combining~\eqref{eq_LL_1} and~\eqref{eq_LL_2}, we conclude the proof of~\eqref{bound_trace_distance}.
\end{proof}

\section{Sample complexity for estimating the covariance matrix using heterodyne measurements}\label{sec_sample}
In this section, we analyse an improved sample complexity upper bound for estimating the covariance matrix of a Gaussian state to a given target accuracy, measured with respect to Schatten norms. The strategy relies on the fact that performing heterodyne measurements on a Gaussian state is equivalent to sampling from a classical Gaussian probability distribution whose covariance matrix and mean vector are directly related to the true quantum covariance matrix and mean vector of the underlying Gaussian state.

The heterodyne measurement is a positive operator-valued measure (POVM) defined by the set of operators 
\[
\left\{ \frac{1}{(2\pi)^n} \ket{\mathbf{r}}\!\bra{\mathbf{r}} \right\}_{\mathbf{r} \in \R^{2n}},
\]
where \(\ket{\mathbf{r}} \coloneqq \hat{D}_{\mathbf{r}} \ket{0}\) is a coherent state~\cite{BUCCO}, and \(\hat{D}_{\mathbf{r}}\) denotes the displacement operator. This measurement constitutes an experimentally feasible scheme that is routinely implemented in CV quantum experiments~\cite{BUCCO}.

For a quantum state \(\rho\), the probability distribution associated with these POVMs is described by the function 
\(Q_\rho(\mathbf{r}) : \R^{2n} \to \R\), defined as
\begin{align}
    Q_\rho(\mathbf{r}) \coloneqq \frac{1}{(2\pi)^n} \bra{\mathbf{r}} \rho \ket{\mathbf{r}}\,,
\end{align}
commonly referred to as the Husimi function.

It is well known~\cite{BUCCO} that for a Gaussian state \(\rho\), the Husimi function is itself a Gaussian probability distribution. It is characterised by the mean vector \(\mathbf{m}(\rho)\) and the covariance matrix \(\frac{V(\rho) + \mathbb{1}}{2}\), where \(V(\rho)\) is the covariance matrix of \(\rho\). Specifically,
\begin{align} \label{husimi_gaussian_state}
    Q_\rho(\mathbf{r}) = \mathcal{N}\!\left(\mathbf{m}(\rho), \frac{V(\rho) + \mathbb{1}}{2} \right)\!(\mathbf{r})\,,
\end{align}
where \(\mathcal{N}(\mathbf{m}, V)(\mathbf{r})\) denotes a Gaussian probability distribution with mean vector \(\mathbf{m}\) and covariance matrix \(V\), explicitly given by
\begin{align}
    \mathcal{N}(\mathbf{m}, V)(\mathbf{r}) \coloneqq \frac{\exp\!\left(-\frac{1}{2} (\mathbf{r} - \mathbf{m})^\intercal V^{-1} (\mathbf{r} - \mathbf{m}) \right)}{(2\pi)^n \sqrt{\det V}}\,.
\end{align}

We now recall established results in the classical literature~\cite{lugosi2017subgaussianestimatorsmeanrandom,vershynin_2018} concerning explicit bounds for estimating the mean vector and covariance matrix from samples of a multivariate Gaussian distribution.
Using Eq.~(1.1) from~\cite{lugosi2017subgaussianestimatorsmeanrandom}, which follows from results such as the Hanson-Wright inequality~\cite{vershynin_2018}, we state the following lemma for estimating the mean vector:

\begin{lemma}[(First moment estimation~\cite{lugosi2017subgaussianestimatorsmeanrandom})]
\label{lem:first_moment_estimation}
Let $\varepsilon, \delta > 0$ denote the desired accuracy and failure probability, respectively. Consider a multivariate Gaussian distribution $\mathcal{N}(\mu, \Sigma)$, where $\mu \in \R^n$ is the mean vector and $\Sigma \in \R^{n \times n}$ is the covariance matrix. Suppose $X_1, X_2, \dots, X_N$ are independent and identically distributed (i.i.d.) samples drawn from $\mathcal{N}(\mu, \Sigma)$. Define the empirical mean estimator as:
\[
\hat{\mu}_N = \frac{1}{N} \sum_{i=1}^N X_i.
\]
Then, with probability at least $1 - \delta$, the following bound holds:
\begin{align}
    \|\hat{\mu}_N - \mu\|_2 &\leq \sqrt{\frac{\|\Sigma\|_1}{N}} + \sqrt{\frac{2 \|\Sigma\|_\infty \log\left( \frac{1}{\delta} \right)}{N}}.
\end{align}
\end{lemma}

For estimating the covariance matrix, we use the result of~\cite[Example 6.3]{Wainwright_2019} (also see~\cite[Page 31]{vershynin2011introductionnonasymptoticanalysisrandom}), which provides explicit sample complexity bound under the assumption of a zero-mean Gaussian distribution:
\begin{lemma}[(Covariance matrix estimation~\cite{Wainwright_2019})]
\label{lem:covariance_estimation}
Let $\varepsilon, \delta > 0$ denote the desired accuracy and failure probability, respectively. Consider a multivariate Gaussian distribution $\mathcal{N}(0, \Sigma)$, where $0 \in \R^n$ is the mean zero vector and $\Sigma \in \R^{n \times n}$ is the covariance matrix. Suppose $X_1, X_2, \dots, X_N$ are independent and identically distributed (i.i.d.) samples drawn from $\mathcal{N}(0, \Sigma)$. Define the empirical covariance matrix as:
\[
\hat{\Sigma}_N = \frac{1}{N} \sum_{i=1}^N X_i X_i^T.
\]
Then, with probability at least $1 - \delta$, the following bound holds:
\begin{align}
\|\hat{\Sigma}_N - \Sigma\|_\infty &\leq \| \Sigma\|_{\infty} \left( 2 \sqrt{\frac{n}{N}} + 2 \sqrt{\frac{2}{N} \ln\left(\frac{2}{\delta}\right)}  + \left( \sqrt{\frac{n}{N}} + \sqrt{\frac{2}{N} \log\left(\frac{1}{\delta}\right)} \right)^2 \right).
\end{align}
\end{lemma}

Combining the two previous Lemmas, we derive explicit sample complexity bounds for estimating the covariance matrix of a multivariate Gaussian distribution without assuming that the mean vector is zero.
\begin{lemma}[(Estimation of covariance matrix and mean vector~\cite{vershynin_2018})]
\label{le:classicalemma}
Let $\varepsilon, \delta > 0$ denote the desired accuracy and failure probability, respectively. Consider a multivariate normal distribution $\mathcal{N}(\mu, \Sigma)$, where $\mu \in \R^n$ is the mean vector and $\Sigma \in \R^{n \times n}$ is the covariance matrix. Let $X_1, X_2, \dots, X_N$ be i.i.d. samples drawn from $\mathcal{N}(\mu, \Sigma)$. Assume that the sample size $N$ satisfies the condition
\begin{align}
    N \ge  \left(2^4+\frac{2^{10}}{n}\right) \frac{\|\Sigma\|^2_\infty}{\varepsilon^2}\left( n+\log\left(\frac{2}{\delta}\right)\right).
\end{align}
Let $\hat{\mu}_N = \frac{1}{N} \sum_{i=1}^N X_i$ be the sample mean, and let $\hat{\Sigma}_N$ be the sample covariance matrix defined as
\begin{align}
    \hat{\Sigma}_N = \frac{1}{N-1} \sum_{i=1}^N (X_i - \hat{\mu}_N)(X_i - \hat{\mu}_N)^T.
\end{align}
Then, with probability at least $1 - \delta$, the following bounds hold:
\begin{align}
    \|\hat{\Sigma}_N - \Sigma\|_\infty &\leq \varepsilon, \\
    \|\hat{\mu}_N - \mu\|_2 &\leq \frac{\varepsilon}{2 \| \Sigma\|^{1/2}_{\infty}}.
\end{align}
\end{lemma}

\begin{proof}
We first consider
\begin{align}
    \hat{W}_N \coloneqq \frac{1}{N} \sum_{i=1}^N (X_i - \mu)(X_i - \mu)^T.
\end{align}
We have that:
\begin{align}
\label{eq:bound0}
    \|\hat{\Sigma}_N - \Sigma\|_\infty &\leq \|\hat{\Sigma}_N - \hat{W}_N\|_\infty + \|\hat{W}_N - \Sigma\|_\infty.
\end{align}
We now analyse the second term \(\|\hat{W}_N - \Sigma\|_\infty\). Note that the random variables \(X_i - \mu\) appearing in \(\hat{W}_N\) are distributed as \(\mathcal{N}(0, \Sigma)\). Consequently, \(\hat{W}_N\) serves as an unbiased estimator for the covariance matrix \(\Sigma\), according to Lemma~\ref{lem:covariance_estimation}. Specifically, if 
\begin{align}
    N \ge  \left(C_1+\frac{C_2}{n}\right) \frac{\|\Sigma\|^2_\infty}{\varepsilon^2}\left( n+\log\left(\frac{2}{\delta}\right)\right)
\end{align}
where $C_1, C_2\ge1$ are constants to be fixed later, then, with probability at least \(1 - \delta/2\), we have
\begin{align}
\label{eq:bound1}
\|\hat{W}_N - \Sigma\|_\infty &\leq \| \Sigma\|_{\infty} \left( 2 \sqrt{\frac{n}{N}} + 2 \sqrt{\frac{2}{N} \log\left(\frac{2}{\delta}\right)}  + \left( \sqrt{\frac{n}{N}} + \sqrt{\frac{2}{N} \log\left(\frac{2}{\delta}\right)} \right)^2 \right)\nonumber\\
&\leq \left(2 \sqrt{\frac{1}{C_1}} + 2 \sqrt{\frac{2}{C_1+C_2}}  + \left( \sqrt{\frac{1}{C_1} } + \sqrt{\frac{2}{C_1+C_2} } \right)^2\right) \varepsilon,
\end{align}
where we used that $\| \Sigma\|_{\infty}\ge 1$.
Before analysing the first term \(\|\hat{\Sigma}_N - \hat{W}_N\|_\infty\), we consider the first moment. Using Lemma~\ref{lem:first_moment_estimation}, we have the following bound that holds with probability $\ge 1 - \delta/2$:
\begin{align}
\label{eq:estfirst}
    \|\hat{\mu}_N - \mu\|_2 &\leq \sqrt{\frac{\|\Sigma\|_1}{N}} + \sqrt{\frac{2 \|\Sigma\|_\infty \log\left( \frac{2}{\delta} \right)}{N}} \leq \left(\sqrt{\frac{2}{C_1}}+\sqrt{\frac{2}{C_1+C_2}}\right)\frac{\varepsilon}{\| \Sigma\|^{1/2}_{\infty}}.
\end{align}

By the union bound, with probability at least $1 - \delta$, both bounds hold simultaneously for the covariance matrix and the mean vector.

Finally, we consider the difference $\hat{\Sigma}_N - \hat{W}_N$. Note that this can be written as
\begin{align}
    \hat{\Sigma}_N - \hat{W}_N &= \frac{1}{1 - \frac{1}{N}} \left( \frac{\hat{W}_N}{N} - (\hat{\mu}_N - \mu)(\hat{\mu}_N - \mu)^T \right)\nonumber\\
    &= \frac{1}{1 - \frac{1}{N}} \left( \frac{\hat{W}_N - \Sigma}{N} + \frac{\Sigma}{N} - (\hat{\mu}_N - \mu)(\hat{\mu}_N - \mu)^T \right).
\end{align}
Taking the operator norms on both sides and applying the triangle inequality, we obtain:
\begin{align}
\label{eq:bound2}
    \|\hat{\Sigma}_N - \hat{W}_N\|_\infty &= \frac{1}{1 - \frac{1}{N}} \left( \frac{\|\hat{W}_N - \Sigma\|_\infty}{N} + \frac{\|\Sigma\|_\infty}{N} + \|(\hat{\mu}_N - \mu)(\hat{\mu}_N - \mu)^T \|_{\infty} \right)\nonumber \\
    &\leq 2\frac{\|\hat{W}_N - \Sigma\|_\infty}{N} + 2\frac{\|\Sigma\|_\infty}{N} + 2\|\hat{\mu}_N - \mu\|_2^2\nonumber\\
    &\leq 2\left(2 \sqrt{\frac{1}{C_1}} +  2 \sqrt{\frac{2}{C_1+C_2}}  + \left(\sqrt{\frac{1}{C_1} } + \sqrt{\frac{2}{C_1+C_2} } \right)^2\right)\frac{\varepsilon^3}{C_1+C_2}\nonumber\\&+2\frac{\varepsilon^2}{C_1+C_2} + 2\left(\sqrt{\frac{2}{C_1}}+\sqrt{\frac{2}{C_1+C_2}}\right)^2\varepsilon^2.
\end{align}
where we have applied the bound from Eq.~\eqref{eq:bound1} and Eq.~\eqref{eq:estfirst}.

Combining the bounds in Eqs.~\eqref{eq:bound0},~\eqref{eq:bound1}, and~\eqref{eq:bound2}, and by choosing $C_1=2^4, C_2=2^{10}$, we conclude the proof of the lemma.
\end{proof}
As an application of the previous lemma, we derive explicit sample complexity bounds for estimating the covariance matrix of an unknown Gaussian state using heterodyne measurements. 

\begin{thm}[(Explicit sample complexity bound for estimating the first moment and covariance matrix of a Gaussian state)]
\label{thm:covariance_est}
Let $\varepsilon, \delta \in (0,1)$ denote the desired accuracy and failure probability, respectively. 
Let \(\rho\) be an unknown \(n\)-mode Gaussian state, with \(\mathbf{m}\) and \(V\) denoting the first moment and covariance matrix of \(\rho\), respectively. Let \(N\) be the number of copies of the unknown state, which we assume to satisfy the condition
    \[
    N \ge  \left(2^{8}+\frac{2^{13}}{n}\right) \frac{\|V\|^2_\infty}{\varepsilon^2}\left( 2n+\log\left(\frac{2}{\delta}\right)\right).
    \]
    Then, heterodyne measurements performed on each of these \(N\) copies of the unknown state \(\rho\) suffice to construct a (valid) covariance matrix \(\tilde{V}\) and a vector $\hat{\mathbb{m}}$ such that, with probability at least \(1 - \delta\), the following bounds hold:
    \begin{align}
        \|\tilde{V} - V\|_\infty &\leq \varepsilon, \\
        \|\tilde{\mathbf{m}} - \mathbf{m}\|_2 &\leq \frac{\varepsilon}{4\|V\|_\infty^{1/2}}.
    \end{align}
\end{thm}

\begin{proof}
    By performing heterodyne measurements on the Gaussian state \(\rho\), we sample from the classical probability distribution
    \begin{equation}
        \mathcal{N}\left[\mathbf{m}(\rho), \frac{V + \mathbb{1}}{2}\right](\mathbf{r}).
    \end{equation}
    Using Lemma~\ref{le:classicalemma}, we assume that the number of copies \(N\) satisfies
    \begin{equation}
        N \geq \left(2^{8}+\frac{2^{13}}{n}\right) \frac{\|V\|_\infty^2}{\varepsilon^2} \left( 2n + \log\left(\frac{2}{\delta}\right)\right)\geq \left(2^{4}+\frac{2^{10}}{2n}\right) \frac{\left\|\frac{V + \mathbb{1}}{2}\right\|_\infty^2}{\left(\frac{\varepsilon}{4}\right)^2} \left( 2n + \log\left(\frac{2}{\delta}\right)\right),
    \end{equation}
    which allows us to construct estimators \(\hat{\Sigma}\) and \(\tilde{\mathbf{m}}\) such that, with probability at least \(1 - \delta\), the following bounds hold:
    \begin{align}
        \|\hat{\Sigma} - \frac{V + \mathbb{1}}{2}\|_\infty &\leq \frac{\varepsilon}{4}, \\
        \|\tilde{\mathbf{m}} - \mathbf{m}\|_2 &\leq \frac{1}{2\left\|\frac{V + \mathbb{1}}{2}\right\|_\infty^{1/2}} \cdot \frac{\varepsilon}{4} \leq \frac{\varepsilon}{4 \|V\|_\infty^{1/2}}.
    \end{align}
    Consequently, we define the estimator for the covariance matrix as
    \begin{equation}
        \hat{V} = 2 \hat{\Sigma} - \mathbb{1}.
    \end{equation}
    From this, we deduce that
    \begin{equation}
        \|\hat{V} - V\|_\infty \leq \frac{\varepsilon}{2}.
    \end{equation}

    However, \(\hat{V}\) may not satisfy the uncertainty relation, and thus may not be a proper covariance matrix. To ensure that \(\tilde{V}\) is a proper covariance matrix, we define
    \begin{equation}
        \tilde{V} = \hat{V} + \frac{\varepsilon}{2} \mathbb{1},
    \end{equation}
    so that by the triangle inequality, we have
    \begin{equation}
        \|\tilde{V} - V\|_\infty \leq \varepsilon.
    \end{equation}

    Next, we show that \(\tilde{V}\) satisfies the uncertainty relation \(\tilde{V} + i \Omega_n \geq 0\). We proceed as follows:
    \begin{align}
        \tilde{V} + i \Omega_n &= \hat{V} + i \Omega_n + \frac{\varepsilon}{2} \mathbb{1}, \\
        &\geq \hat{V} - V + \frac{\varepsilon}{2} \mathbb{1}, \\
        &\geq -\|\hat{V} - V\|_\infty \mathbb{1} + \frac{\varepsilon}{2} \mathbb{1}, \\
        &\geq 0.
    \end{align}
    In the first inequality, we used the fact that \(V + i \Omega_n \geq 0\). In the second inequality, we applied the result that for any operator \(\Theta\), we have \(\Theta \geq -\|\Theta\|_\infty \mathbb{1}\). Therefore, \(\tilde{V}\) satisfies the uncertainty relation and is a proper covariance matrix, completing the proof.
\end{proof}
\begin{remark}
    Note that the sample complexity bound in the previous theorem requires knowledge of an upper bound on the covariance matrix \(\|V\|_{\infty}\) of the unknown Gaussian state \(\rho\) in order to predict the number of measurements needed to accurately estimate the first moment and covariance matrix of the unknown state. However, if we assume that the unknown state \(\rho\) satisfies an energy constraint, i.e., 
    \[
    \Tr[\rho \hat{E}] \leq E, \quad \text{where} \quad \hat{E} \coloneqq \frac{1}{2} \hat{\mathbf{R}}^\intercal \hat{\mathbf{R}},
    \]
    denotes the energy operator~\cite{BUCCO}, then we have the bound
    \[
    \|V\|_{\infty} \leq 4E.
    \]
\end{remark}
This follows because
\begin{equation}
\label{eq:boundencov}
    \|V\|_{\infty} \leq \Tr V = 4 \Tr[\hat{E} \rho] - 2 \|\mathbf{m}\|_2^2 \leq 4E.
\end{equation}

By leveraging the trace-distance bounds introduced earlier, we can derive the following result:

\begin{thm}[(Upper bound on the sample complexity of tomography of Gaussian states)]
\label{th:appmain2}
Let $\varepsilon, \delta \in (0,1)$ denote the desired accuracy and failure probability, respectively. 
Let \(\rho\) be an unknown \(n\)-mode Gaussian state with an unknown covariance matrix \(V\) and first moment \(\mathbf{m}\). A number 
\begin{equation}
    N \geq \left(2^{9}+\frac{2^{14}}{n}\right)  \frac{\|V\|_1^2 \|V\|_\infty^2}{\varepsilon^2} \left( 2n + \log\left(\frac{2}{\delta}\right) \right)
\end{equation}
of copies of \(\rho\) suffices to construct a classical description of a Gaussian state \(\tilde{\rho}\) such that
\begin{equation}
    \Pr\left( \frac{1}{2} \|\tilde{\rho} - \rho\|_1 \leq \varepsilon \right) \geq 1 - \delta.
\end{equation}
\end{thm}

\begin{proof}
Using Theorem~\ref{thm:covariance_est}, a number of copies 
\begin{equation}
    N \geq \left(2^{8}+\frac{2^{13}}{n}\right)  \frac{\|V\|_\infty^2}{\varepsilon'^2} \left( 2n + \log\left(\frac{2}{\delta}\right) \right)
\end{equation}
suffices to construct estimates \(\tilde{\mathbf{m}}\) and \(\tilde{V}\), with the following guarantees (with probability at least \(1-\delta\)):
\bb
    \|\tilde{V} - V\|_\infty &\leq \varepsilon', \\
    \|\tilde{\mathbf{m}} - \mathbf{m}\|_2 &\leq \frac{\varepsilon'}{4 \|V\|_\infty^{1/2}}.
\ee
Define \(\tilde{\rho}\) as the Gaussian state with first moment \(\tilde{\mathbf{m}}\) and covariance matrix \(\tilde{V}\). Substituting these bounds into Eq.~\eqref{eq_smsmsm}, we obtain:
\begin{align}
    \frac{1}{2}\|\rho - \tilde{\rho}\|_1 &\leq 
    \frac{1 + \sqrt{3}}{8} \max(\|V\|_1, \|\tilde{V}\|_1) \|V - \tilde{V}\|_\infty 
    + \sqrt{\frac{\min(\|V\|_\infty, \|\tilde{V}\|_\infty)}{2}} \|\mathbf{m} - \tilde{\mathbf{m}}\|_2 \nonumber \\
    &\leq \frac{\varepsilon'}{2} \max(\|V\|_1, \|\tilde{V}\|_1) 
    +\frac{ \varepsilon'}{4\sqrt{2}} \nonumber \\
    &\leq \frac{\varepsilon'}{2} \|V\|_1+ \frac{\varepsilon'^2}{2}   + \frac{\varepsilon'}{4\sqrt{2}}.
\end{align}

Setting \(\varepsilon' = \frac{\varepsilon}{\sqrt{2} \|V\|_1}\) ensures that $\frac{1}{2}\|\rho - \tilde{\rho}\|_1 \leq \varepsilon$ holds with probability at least \(1-\delta\). 
\end{proof}

To address the energy-constrained case $\Tr[\rho \hat{E}]\le E$, note the following bound (as in Eq.\eqref{eq:boundencov}:
\begin{equation}
    \|V\|_\infty \leq \|V\|_1 \leq 4E.
\end{equation}
This implies the following corollary:

\begin{cor}[(Upper bound on the sample complexity for energy-constrained states)]
\label{cor:energy_constraint}
Let $\varepsilon, \delta \in (0,1)$ denote the desired accuracy and failure probability, respectively. 
Let \(\rho\) be an unknown \(n\)-mode Gaussian state satisfying the energy constraint \(\Tr[\rho \hat{E}] \leq E\). A number 
\begin{equation}
    N \geq \left(2^{15}+\frac{2^{22}}{n}\right) \frac{E^4}{\varepsilon^2} \left( 2n + \log\left(\frac{2}{\delta}\right) \right)
\end{equation}
of copies of \(\rho\) suffices to construct a classical description of a Gaussian state \(\tilde{\rho}\) such that
\begin{equation}
    \Pr\left( \frac{1}{2}\|\tilde{\rho} - \rho\|_1 \leq \varepsilon \right) \geq 1 - \delta.
\end{equation}
\end{cor}
The constant factors in our sample-complexity bounds are likely overestimated due to the conservative approximations employed in the proofs to ensure simplicity and generality across all parameter regimes. In practice, one can anticipate considerably better constants.
 
\end{document}